\documentclass[11pt]{amsart} 

\usepackage[utf8]{inputenc} 

\usepackage{geometry}
\usepackage{graphbox,graphicx}

\usepackage[obeyspaces]{url}

\usepackage{amsmath, amssymb, amsthm}
\usepackage{caption}
\usepackage{subcaption} 
\usepackage{appendix}
\usepackage{xcolor}
\usepackage{stmaryrd} 
\usepackage{tikz}

\newtheorem{theorem}{Theorem}[section]
\newtheorem{cor}[theorem]{Corollary}

\newtheorem{prop}[theorem]{Proposition}
\theoremstyle{definition}

\theoremstyle{remark}
\newtheorem{remark}[theorem]{Remark}
\theoremstyle{conjecture}
\newtheorem{conj}[theorem]{Conjecture}

\numberwithin{equation}{section}

\newcommand{\bU}{\mathbf{U}}

\newcommand{\Tr}{\mathrm{Tr}\,}

\newcommand{\linc}{l^{\boxslash}}
\newcommand{\ldec}{l^{\boxbslash}}

\newcommand{\rmhard}{\mathrm{hard}}
\newcommand{\rmsoft}{\mathrm{soft}}
\newcommand{\rmAi}{\mathrm{Ai}}

\newcommand{\bbE}{\mathbb{E}}
\newcommand{\Var}{\mathrm{Var}}

\def\picHamLines{\begin{tikzpicture}[scale=0.4, every node/.style={transform shape}]

\draw[thick] (0,0)--(0,9)--(9,9)--(9,0)--(0,0);
\node[below, left] at (-0.2, -0.2) {\Huge $(0,0)$};
\node[below, right] at (9.2, -0.2) {\Huge $(1,0)$};
\node[above, left] at (-0.2, 9.2) {\Huge $(0,1)$};
\node[above,right] at (9.2, 9.2) {\Huge $(1,1)$};

\draw[ultra thick,color=white!25!black,fill=black] (2,1) circle (3pt);
\draw[ultra thick,color=white!25!black,fill=black] (5,2) circle (3pt);
\draw[ultra thick,color=white!25!black,fill=black] (6,3) circle (3pt);
\draw[ultra thick,color=white!25!black,fill=black] (8,4) circle (3pt);
\draw[ultra thick,color=white!25!black,fill=black] (7,5) circle (3pt);
\draw[ultra thick,color=white!25!black,fill=black] (3,6) circle (3pt);
\draw[ultra thick,color=white!25!black,fill=black] (4,7) circle (3pt);
\draw[ultra thick,color=white!25!black,fill=black] (1,8) circle (3pt);

\draw (0,0)--(2,1)--(5,2)--(6,3)--(8,4)--(9,9);
\draw (6,3)--(7,5)--(9,9);

\end{tikzpicture}}

\def\picHamPerm{\begin{tikzpicture}[scale=0.4, every node/.style={transform shape}]

\draw[thick] (0,0)--(0,9)--(9,9)--(9,0)--(0,0);

\draw[ultra thick,color=white!25!black,fill=black] (2,1) circle (3pt);
\draw[ultra thick,color=white!25!black,fill=black] (5,2) circle (3pt);
\draw[ultra thick,color=white!25!black,fill=black] (6,3) circle (3pt);
\draw[ultra thick,color=white!25!black,fill=black] (8,4) circle (3pt);
\draw[ultra thick,color=white!25!black,fill=black] (7,5) circle (3pt);
\draw[ultra thick,color=white!25!black,fill=black] (3,6) circle (3pt);
\draw[ultra thick,color=white!25!black,fill=black] (4,7) circle (3pt);
\draw[ultra thick,color=white!25!black,fill=black] (1,8) circle (3pt);

\draw (0,0)--(2,1)--(5,2)--(6,3)--(8,4)--(9,9);
\draw (6,3)--(7,5)--(9,9);

\node[below] at (1,-0.2) {\Huge $1$};
\node[below] at (2,-0.2) {\Huge $2$};
\node[below] at (3,-0.2) {\Huge $3$};
\node[below] at (4,-0.2) {\Huge $4$};
\node[below] at (5,-0.2) {\Huge $5$};
\node[below] at (6,-0.2) {\Huge $6$};
\node[below] at (7,-0.2) {\Huge $7$};
\node[below] at (8,-0.2) {\Huge $8$};

\node[left] at (-0.2,8) {\Huge $1$};
\node[left] at (-0.2,1) {\Huge $2$};
\node[left] at (-0.2,6) {\Huge $3$};
\node[left] at (-0.2,7) {\Huge $4$};
\node[left] at (-0.2,2) {\Huge $5$};
\node[left] at (-0.2,3) {\Huge $6$};
\node[left] at (-0.2,5) {\Huge $7$};
\node[left] at (-0.2,4) {\Huge $8$};

\end{tikzpicture}}

\begin{document}

\title[Finite size corrections and longest
increasing subsequences]{Finite size corrections relating to distributions of the length of longest
increasing subsequences}
\author{Peter J. Forrester and Anthony Mays}
\address{School of Mathematics and Statistics, 
ARC Centre of Excellence for Mathematical \& Statistical Frontiers,
University of Melbourne, Victoria 3010, Australia}
\email{pjforr@unimelb.edu.au}

\date{\today}


\begin{abstract}
Considered are the large $N$, or large intensity, forms of the distribution of the length of the longest increasing
subsequences for various models. Earlier work has established that after centring and scaling, the limit laws for  these distributions relate to certain 
distribution functions at the hard edge known from random matrix theory. 
By analysing the hard to soft edge transition, we supplement and extend results of Baik and Jenkins for
the Hammersley model and symmetrisations, which give  that the leading 
correction is proportional to $z^{-2/3}$, where $z^2$ is the intensity of the Poisson rate,
and provides a functional form as derivates of the limit law.
Our methods give the functional form 
 both in terms of Fredholm operator theoretic quantities, and in terms of Painlev\'e transcendents. For random
permutations and their symmetrisations, numerical analysis of exact enumerations and simulations gives compelling evidence that the leading
corrections are proportional to $N^{-1/3}$, and moreover provides an approximation to their graphical forms.

\end{abstract}


\maketitle

\section{Introduction}

Taking a viewpoint of random matrix theory in probability theory, it is very natural to ask
about the rate of convergence to universal laws. Consider for example the spacing distribution,
$p_2(s)$ say, between consecutive eigenvalues in ensembles with unitary symmetry. 
Here the subscript is the Dyson index $\beta = 2$ for unitary symmetry. The corresponding
universal law, obtained by taking the large $N$ limit of an ensemble with unitary
symmetry and scaling the mean spacing to unity,
tells us that \cite{Dy62a}
\begin{equation}\label{1.1}
p_2(s) = {d^2 \over d s^2} \log \Big ( 1 - \mathbb K_{(0,s)}^{\rm sine} \Big ),
\end{equation}
where $K_{(0,s)}^{\rm sine}$ is the Fredholm determinant of the integral operator 
on $(0,s)$ with the so-called sine kernel
\begin{equation}\label{1.1a}
K^{\rm sine}(x,y) = {\sin \pi (x - y) \over \pi (x - y)}.
\end{equation}
An example of an ensemble with unitary symmetry is the set of $N \times N$
unitary matrices chosen with Haar (uniform) measure. With $p_{2,N}^{U(N)}(s)$
denoting the spacing distribution between consecutive eigenvalues in this
ensemble, the limit theorem relating to (\ref{1.1}) is that
\begin{equation}\label{1.1b}
\lim_{N \to \infty} \Big ( { 2 \pi \over N}  \Big )^2 p_{2,N}^{U(N)}(2 \pi s/N ) = p_2(s).
\end{equation}
The rate of convergence question may be posed by asking for a tight bound on
\begin{equation}\label{1.1c}
\sup_{0\le s \le  N} \Big | p_{2,N}^{U(N)}(2 \pi s/N )  - p_2(s)  \Big |.
\end{equation}

Less ambitious, but more in keeping with an applied mathematics viewpoint on this
aspect of random matrix theory, is to ask for the leading term in the large $N$
asymptotic expansion of the difference 
\begin{equation}\label{1.1d}
p_{2,N}^{U(N)}(2 \pi s/N )  - p_2(s)
\end{equation}
occurring in (\ref{1.1c}) for $s$ fixed. Indeed this question is central to probing the Keating--Snaith
hypothesis \cite{KS00a} relating the statistical distribution of the eigenvalues of Haar distributed
random unitary matrices to the statistical distribution of the zeros of the Riemann zeta
function on the critical line \cite{Od01,BBLM06,FM15,BFM17}. Here one uses Odlyzko's
data set \cite{Od01} of over $10^9$ high precision consecutive zeros about a zero number near $10^{23}$
to obtain the empirical spacing distribution. From a graphical viewpoint this appears identical
to $p_2(s)$. This is in keeping with the Montgomery--Odlyzko law (see e.g.~\cite{Sa04}) equating 
the scaled local statistics of the Riemann zeros infinitely high up the critical line
to the limiting bulk scaled eigenvalue statistics from any random matrix ensemble with
unitary symmetry. However there are finite size effects --- even though $10^{23}$ is huge
on an absolute scale, it is the logarithm of the zero number which is the relevant measure
of size. The extraordinary statistics provided by Odlyzko's data set allows for the functional
form of the analogue of the difference (\ref{1.1d}), where now $p_{2,N}^{U(N)}(2 \pi s/N ) $ is
replaced by the empirically determined spacing distribution, to be accurately determined. The
Keating--Snaith hypothesis predicts that this difference will be identical to the difference (\ref{1.1d})
for an appropriate value of $N$ and rescaling of $s$. Hence the applied interest in
(\ref{1.1d}) for fixed $s$ and large $N$, a study of which was undertaken in \cite{BBLM06,FM15,BFM17}.

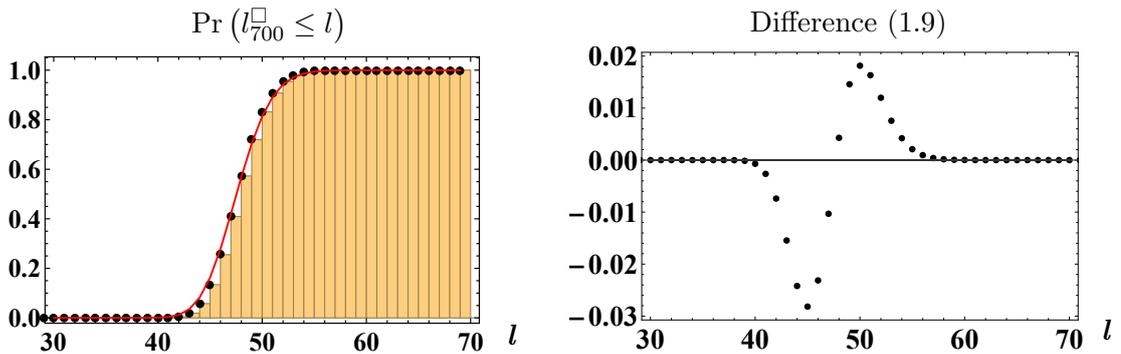
\begin{figure}
\centering
     \begin{minipage}{0.45\textwidth}
	\begin{tikzpicture}[scale=1, every node/.style={transform shape}]
         \node[inner sep=0pt] at (0,0) {\includegraphics[height=4cm, align=t]{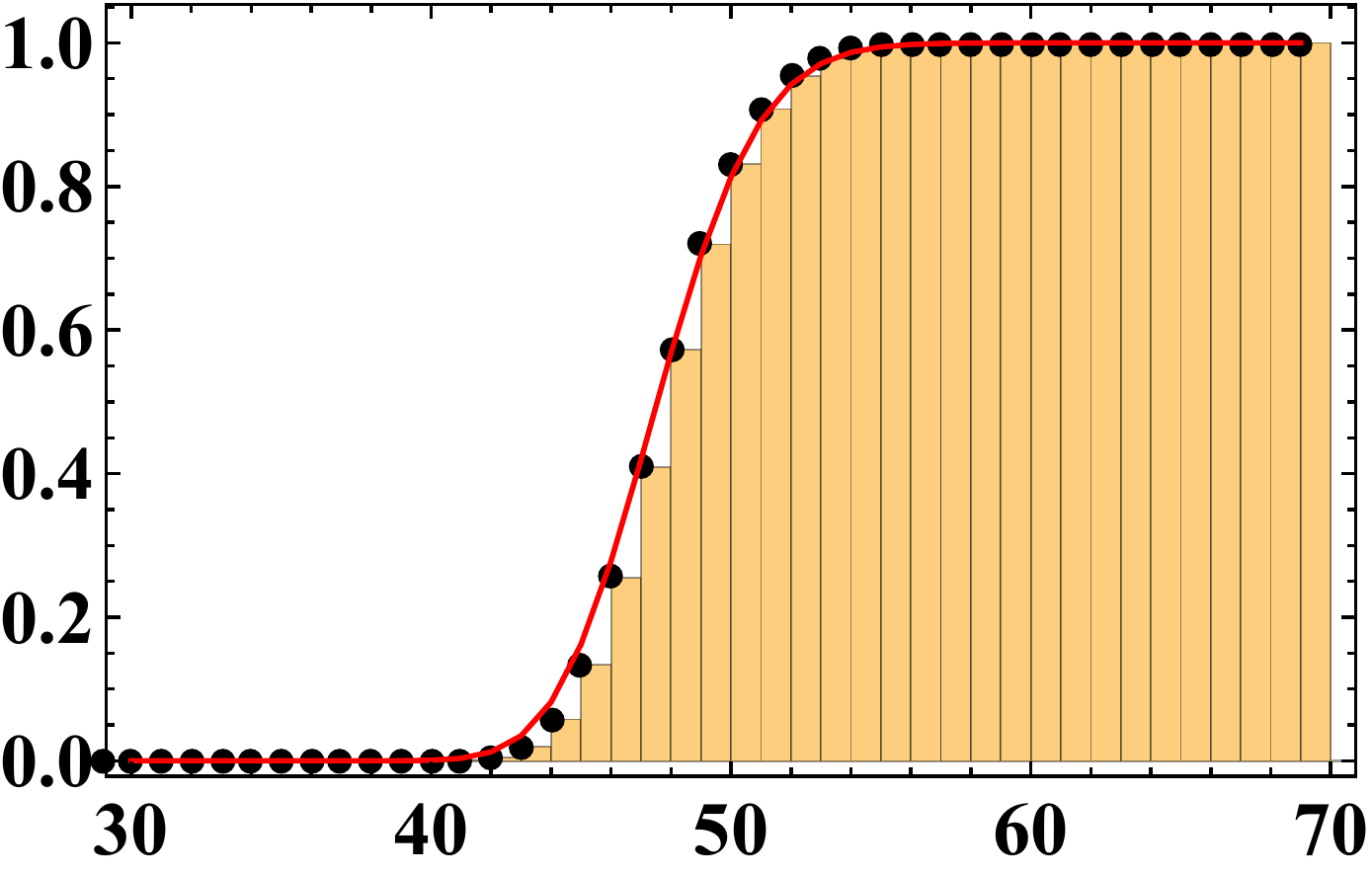}};
         \node at (3.5,-1.7) {\textbf{\textit{l}}};
         \node at (0.3,2.4) {$\Pr\big( l^{\Box}_{700} \leq l\big)$};
         \end{tikzpicture}
	\end{minipage} \quad
	\begin{minipage}{0.45\textwidth}
	\begin{tikzpicture}[scale=1, every node/.style={transform shape}]
         \node[inner sep=0pt] at (0,0) {\includegraphics[height=4.17cm, align=t]{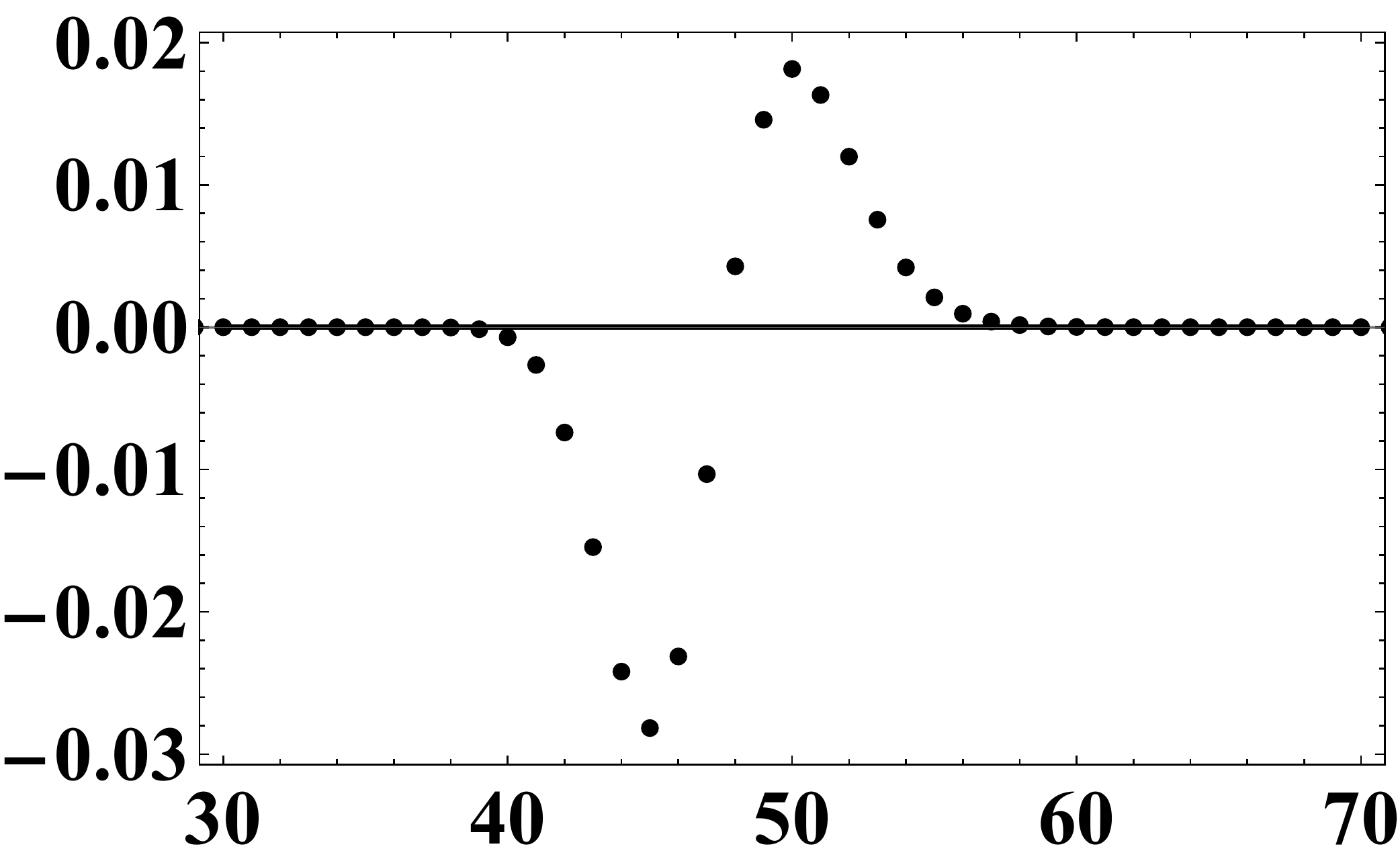}};
         \node at (3.7,-1.75) {\textbf{\textit{l}}};
         \node at (0.3,2.3) {Difference \eqref{1.1h}};
         \end{tikzpicture}
	\end{minipage}
        \caption{On the left we have the empirical CDF of the longest increasing subsequences of 1,000,000 random permutations of length $N= 700$, along with the calculation of the exact CDF using $c_{700}^{\Box} (l)$ in \eqref{e:Gtlexp} [black dots] and the limiting CDF given by the second term in (\ref{1.1h}) [red curve]. On the right is plotted the difference \eqref{1.1h}.}
        \label{f:AsymmLISn700}
\end{figure}

Our interest in the present work relates to the finite size corrections of another applied problem from random matrix theory, this time in the field of combinatorics. Take the set of the first $N$ positive integers, and choose a permutation uniformly
at random. The question is to specify the statistics of the longest increasing subsequence
length, $l_N^\square$ say, from the permutation in the large $N$ limit. To explain the notion of the longest increasing subsequence, suppose $N = 8$ and the permutation is the ordered list $2,5,6,8, 7,3,4,1$. The longest subsequences of increasing numbers in this list have length 4: $2,5,6,8$ and $2,5,6,7$. A result of Logan and Shepp 
\cite{LS77} (see the
Introduction in \cite{OR00} for more context) gives that asymptotically the expected
length is $2 \sqrt{N}$. The tie in with random matrix theory is the limit theorem  \cite{BDJ98}
\begin{equation}\label{1.1e}
\lim_{N \to \infty} {\rm Pr} \Big ( {l_N^\square - 2 \sqrt{N} \over N^{1/6}}  \le t \Big ) = E_2^{\rm soft}(0;(t,\infty)),
\end{equation}
where $E_2^{\rm soft}(0,(t,\infty))$ is the probability that, after centring and scaling about the
largest eigenvalue, the interval $(t,\infty)$ is free of eigenvalues
in an Hermitian random
matrix ensemble with unitary symmetry, or equivalently
$E_2^{\rm soft}(0,(t,\infty))$ is the distribution of the scaled largest eigenvalue. This quantity has the exact evaluations  \cite{Fo93a,TW94a}
\begin{align}\label{1.1f}
E_2^{\rm soft}(0;(t,\infty)) & = \det \Big ( \mathbb I - \mathbb K_{(t,\infty)}^{\rm soft} \Big )      \nonumber     \\
& = \exp \bigg ( - \int_t^\infty ( t - s) q^2_0(s) \, ds \bigg ).
\end{align}
In the first of these expressions, $\mathbb{K}_{(t,\infty)}^{\rm soft}$ is the integral operator on $(t,\infty)$ with
kernel
\begin{equation}\label{1.1g}
K^{\rmsoft}(x,y)= {{\rm Ai}(x)  {\rm Ai}'(y)  - {\rm Ai}(y)  {\rm Ai}'(x)  \over x - y},
\end{equation}
where ${\rm Ai}(u)$ denotes the Airy function. In the second expression, $q_0(t)$ is the solution of the
particular Painlev\'e II equation $q'' = s q + 2 q^3$ satisfying the boundary condition $q_0(s) \sim
{\rm Ai}(s)$ as $s \to \infty$. Much more about the mathematics relating to the longest
increasing subsequence problem for a random permutation can be found in \cite{AD99}
and \cite{Ro15}.

In analogy with our discussion of studies in random matrix theory motivated by 
Odlyzko's data for the Riemann zeros, an immediate question is to inquire about the
large $N$ form of the difference
\begin{equation}\label{1.1h}
 {\rm Pr} ( l_N^\square \le l)  - E_2^{\rm soft} \left( 0;\left( \frac{l- 2\sqrt{N}}{N^{1/6}} ,\infty \right) \right).
 \end{equation}
 There are various ways to generate exact and simulated data for this quantity; these are discussed
 in \S 4.
 In Figure \ref{f:AsymmLISn700}, in the first panel we plot the histogram corresponding to the 
 empirical cumulative distribution function (CDF) for
 $l^{\Box}_{700}$, computed from the longest increasing subsequence length of $10^6$
 random permutations. Also plotted, as black dots, is the exact CDF as a function of the positive
 integer $l$, calculated as described in Section \ref{S4.2} below, and plotted as a red line is the quantity $E_2^{\rmsoft}( 0; ((l- \sqrt{N})/N^{1/6}, \infty) ) |_{N = 700}$ with $l$ varying continuously. In the second panel the difference (\ref{1.1h}) is displayed.  In relation to the functional form in the second panel, we would like to know its dependence on $N$, and its functional form at next-to-leading order.  Our result of Conjecture \ref{C1} asserts the large $N$ asymptotic expansion
\begin{equation}\label{1.1i}
 {\rm Pr} \Big ( {l_N^\square - 2 \sqrt{N} \over N^{1/6}}  \le t \Big ) =  F_{2,0}(t^*) + {1 \over N^{1/3}} F_{2,1}(t) +  \cdots, \quad t^* := ([2 \sqrt{N} + t N^{1/6}] - 2 \sqrt{N})/N^{1/6},
  \end{equation}
  where $ F_{2,0}(t)  = E_2^{\rm soft}(0;(t,\infty))$ as known from (\ref{1.1e}), while the functional form $F_{2,1}(t)$
  remains unknown as an analytic function, but can be approximated graphically; see Figure~\ref{f:CompDelta2}.
  

\begin{figure}
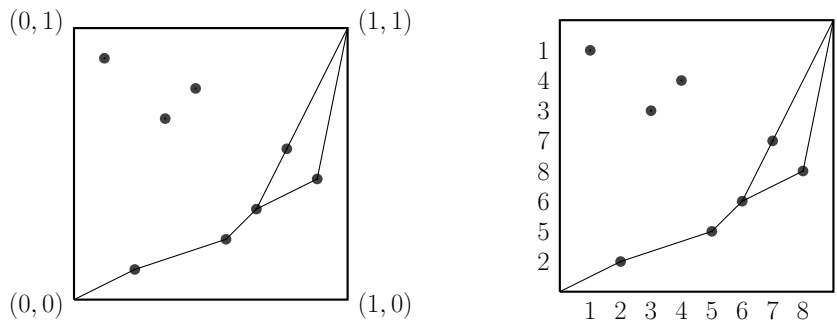

\centering
\picHamLines \qquad \qquad \picHamPerm
\caption{An example of the Hammersley process, where the (random) number of points is $N= 8$. On the left we have eight points marked in the unit square, and the longest paths we can make from this arrangement of points (as measured by the number of points on the path) using only segments with positive slope have length four, \textit{i.e.} $l^{\Box} =4$. The mapping to the permutation $(2,5,6,8, 7,3,4,1)$ is in the diagram on the right, where we number the points sequentially on the horizontal axis, and then move up the vertical axis listing the label of each point as we reach it. The two longest paths correspond to the subsequences $(2,5,6,8)$ and $(2,5,6,7)$.}
\label{f:HammPath}
\end{figure}

 There is a well known Poissonized form of the longest increasing subsequence problem known as
 the Hammersley process (see e.g.~\cite[\S 10.6]{Fo10}). 
Each permutation of length $N$ appears with probability
 $ z^{2N}/ N!$ and is represented by $N$ points in the unit square, where $z$ is the Poisson rate of the number of these points. The length of the longest
 increasing subsequence is now the maximum over the number of points that can be joined using
 straight line segments starting from the origin and finishing at $(1,1)$,
 or equivalently the maximum number of segments in a path through the dots which always goes up and
 to the right.
  The example of the permutation
 for $N = 8$ given in the second paragraph is illustrated in Figure \ref{f:HammPath} from this viewpoint. Define this
 length to be the random variable $l^\square = l^\square(z)$.
 From the definition,
 \begin{equation}\label{1.1j}
  {\rm Pr} ( l^\square  \le l) =   e^{- z^2} \sum_{N=0}^\infty {z^{2N} \over N!}  {\rm Pr} ( l_N^\square  \le l) .
  \end{equation}
  Analogous to (\ref{1.1e}), this quantity satisfies the limit law \cite{BDJ98, BF03}
 \begin{equation}\label{1.1k}   
  \lim_{z\to \infty} \Pr \left(\frac{l^{\Box} -2z}{z^{1/3}} \leq t \right)= E_2^{\rmsoft} \Big( 0; (t, \infty) \Big).
    \end{equation}
    The study \cite{FF11} showed the existence of an expansion analogous to (\ref{1.1i}) with $N$ replaced by $z^2$, while most significantly from the present viewpoint, the subsequent work of Baik and Jenkins \cite[Theorem 1.3]{BJ13} provided an explicit functional form of the analogue of $F_{2,1}(s)$.

    The reasons why ${\rm Pr} ( l^\square  \le l)$ is more tractable than ${\rm Pr} ( l_N^\square \le l)$ for asymptotic
    analysis are either of the formulas \cite{Ge90,Ra98}
    \begin{align}
\label{e:lBoxE2HardAv} \Pr(l^{\Box} \leq l) = e^{-z^2} \left\langle e^{2z \sum_{j=1}^l \cos \theta_j} \right\rangle_{U(l)},
\end{align}
where $\langle \cdot \rangle_{U(l)}$ is the average over the unitary group of degree $l$ with Haar measure,
or \cite{BF03}
\begin{align}
\label{e:lBoxE2Hard} \Pr(l^{\Box} \leq l)= E_2^{\rmhard} \Big( 0; (0, 4z^2); l \Big),
\end{align}
where $E_2^{\rmhard} \Big( 0; (0, s); a \Big)$ is the hard edge scaled probability that in the Laguerre
unitary ensemble with parameter $a$ (recall the Laguerre weight is $x^a e^{-x}$) the interval $(0,s)$
is free of eigenvalues. The first of these was used in \cite{BDJ98} to prove (\ref{1.1k}), while an alternative
proof of (\ref{1.1k}) using (\ref{e:lBoxE2Hard}) was given in \cite{BF03}. 

 As with the work \cite{BJ13}, the question of interest is now to quantify the large $z$ expansion of
   \begin{equation}\label{1.1l}   
 \Pr \left(\frac{l^{\Box} -2z}{z^{1/3}} \leq t \right) - E_2^{\rmsoft} \Big( 0; (t, \infty) \Big).
 \end{equation}  
 In this setting a precise statement can be made, supplementing the result known from \cite{BJ13};
 in relation to the latter see \S 2.2. To state our result requires introducing the integral
 operator 
 $\mathbb{L}_{(t, \infty)}$  on $(t, \infty)$ with kernel
\begin{align}
\nonumber L (x,y):= - \frac{1}{2^{1/3} (x-y)} &\Bigg[ \frac{x-y}{5} \Big( {\rm Ai}(x) {\rm Ai}'(y) + {\rm Ai}'(x) {\rm Ai} (y) \Big) \\
\label{d:Lxy} & + \frac{x^3 -y^3}{30} {\rm Ai}(x) {\rm Ai }(y) - \frac{x^2 -y^2}{30} {\rm Ai}'(x) {\rm Ai}' (y) \Bigg].
\end{align}

 \begin{prop}\label{P1}
 With $[ u ]$ denoting the integer part of a positive real number $u$, let
  \begin{equation}\label{1.1u} 
  \tilde{t} = ([2z+ t z^{1/3}] - 2 z)/ z^{1/3}.
  \end{equation}
  For large $z$ we have
 \begin{equation}\label{1.1m}   
   \Pr \left(\frac{l^{\Box} -2z}{z^{1/3}} \leq t \right) = F_{2,0}^{ \rm H }(\tilde{t}) + {1 \over (2z)^{2/3}} F_{2,1}^{\rm H}(t) + {\rm O}(z^{-4/3}),
  \end{equation}
  where $ F_{2,0}^{\rm H}(\tilde{t})  = E_2^{\rm soft}(0,(\tilde{t},\infty))$ and 
  \begin{align}
\label{e:Corr2t}  F_{2,1}^{\rm H }(t) &= - \det \Big( \mathbb I - \mathbb{K}^{\rmsoft}_{(t, \infty)} \Big) \Tr \Big( (\mathbb I - \mathbb{K}^{\rmsoft}_{(t, \infty)})^{-1} \mathbb{L}_{(t, \infty)} \Big).
\end{align}
    \end{prop}   
    
    In keeping with the two characterisations of $E_2^{\rm soft}(0; (t, \infty) )$ given in (\ref{1.1f}), one as
    a Fredholm determinant and  the other  in terms of the Painlev\'e transcendent $q_0(s)$, the correction
    $F_{2,1}^{\rm H}(t)$ can  alternatively be written in terms of a second order linear differential
    equation with a Painlev\'e transcendent relating to $q_0(s)$ occurring in the coefficients; see Proposition \ref{p:Corr2DE} below.
    Moreover we show that this can be further simplified to give agreement with the result of Baik and Jenkins \cite{BJ13}.

There are symmetrised versions of the longest increasing subsequence problem and the corresponding
Hammersley process that permit analogues of the limit laws  \eqref{1.1e} and (\ref{1.1k}), and which also admit analogues of Proposition \ref{P1} \cite{BR01,BR01a,BR01b}. To specify these, we first recall that a permutation of
$\{1,\dots,N\}$ can be represented as an $N \times N$ matrix, $P$ say, of zeros and ones with exactly $N$ ones,
distributed so each row and column has a single one. For convenience, number the rows of $P$ starting
from the bottom, and suppose that whenever there is an entry one in position $(i,j)$, there is also an
entry one in position $(j,i)$. If furthermore there are no entries on the diagonal, $N$ must be even
and the permutation consists entirely of two cycles. For such a permutation chosen uniformly at random, 
denote the random variable corresponding to the length of the longest increasing subsequence by
$\linc_{N}$. Now, in the corresponding two line presentation, suppose the order of the second line is reversed,
which is equivalent to rotating $P$ by ninety degrees clockwise. Assuming again that the original permutation
of two cycles was chosen uniformly at random,  the random variable corresponding to the longest increasing subsequence  length of this rotated permutation (equivalently, the longest \textit{decreasing} subsequence of the original permutation) is
to be denoted by $\ldec_{N}$. 

The Hammersley model relating to $\linc_{N}$ has only the points
in the unit square below the diagonal from $(0,0)$ to $(1,1)$ independent; the points above the diagonal are
reflections of these points. The longest up/right path length in this setting will be denoted
$\linc = \linc(z)$.
Similarly, for the rotated Hammersley model relating to $\ldec_{N}$, only the points in the unit square below the diagonal
from $(0,1)$ to $(1,0)$ are chosen independently, with the remaining points the reflection in this diagonal of
those points, and we denote by  $\ldec = \ldec(z)$  the longest up/right path length.

The analogue of the limit laws (\ref{1.1e}) and (\ref{1.1k}) are now \cite{BR01b}
\begin{equation}\label{1.1n}  
\lim_{N \to \infty} {\rm Pr} \Big ( {\linc_N - 2 \sqrt{N} \over N^{1/6}}  \le t \Big ) = 
\lim_{z \to \infty} {\rm Pr} \Big ( {\linc - 2 z \over z^{1/3}}  \le t \Big ) 
= \tilde{E}_4^{\rm soft}(0;(t,\infty))
\end{equation}
and 
\begin{equation}\label{1.1o}  
\lim_{N \to \infty} {\rm Pr} \Big ( {\ldec_N - 2 \sqrt{N} \over N^{1/6}}  \le t \Big ) = 
\lim_{z \to \infty} {\rm Pr} \Big ( {\ldec - 2 z \over z^{1/3}}  \le t \Big ) 
= {E}_1^{\rm soft}(0;(t,\infty)).
\end{equation}
The quantities $\tilde{E}_4^{\rm soft}(0;(t,\infty))$ and ${E}_1^{\rm soft}(0;(t,\infty))$ denote 
the probability that upon a soft edge scaling in the neighbourhood of the largest eigenvalue
in the Gaussian $\beta$ ensemble with $\beta = 4$ and $\beta = 1$, the interval
$(t,\infty)$ is free of eigenvalues. The tilde symbol on    $\tilde{E}_4^{\rm soft}(0;(t,\infty))$
indicates a rescaling of the natural soft edge  Gaussian $\beta$ ensemble variables;
see \cite[displayed equation below (9.139)]{Fo10}.
As in the case of the limit laws associated with $l_N^\Box$ and $l^\Box$, we seek corrections
to these limit laws. Our results are contained in Conjecture \ref{C2} in relation to
$\linc_N, \ldec_N$, and in Propositions \ref{P1a}, \ref{P3.3} and Corollaries \ref{C3.2}, \ref{C3.4} for $\linc, \ldec$.
Baik and Jenkins \cite[Theorem 1.2]{BJ13} contains a result that can be interpreted as corresponding
to  Proposition \ref{P1a}, but with a different functional form for $F_{1,1}^{{\rm H}}(t)$. In Section
\ref{S3.2} we discuss the latter in the context of our Painlev\'e characterisation of $F_{1,1}^{{\rm H}}(t)$.

\section{Large $z$ expansion of $   \Pr \left(\frac{l^{\Box} -2z}{z^{1/3}} \leq t \right)$}
\subsection{Proof of Proposition \ref{P1}} 
Following \cite{BF03} our strategy is to analyse the large $z$ form of the LHS of (\ref{1.1m}) by making
use of (\ref{e:lBoxE2Hard}). In the latter, $l$ is a simpler variable to work with than $z$, so to begin we
need to express $z$ in terms of $l$. We see by matching the LHS of the respective equations that
\begin{equation}\label{2.1a}
l = [2z + t z^{1/3}],
\end{equation}
where $[ \cdot ]$ denotes the integer part.
Introduce
\begin{equation}\label{2.1a+}
\tilde{l} = 2z + t z^{1/3}.
\end{equation}
It is convenient to consider (\ref{2.1a+}) with $t$ replaced by $X$ so that $z = z(\tilde{l};X)$. This function of
$\tilde{l}$ and $X$ is uniquely determined by (\ref{2.1a+}) and the requirement that $z \sim \tilde{l}/2$, independent
of $X$, for $\tilde{l}$ large. Furthermore we introduce notation for the square of $z(\tilde{l};X)$ and note the large
$\tilde{l}$ expansion of the latter
\begin{equation}\label{2.1b}
Q(\tilde{l};X) = (2 z(\tilde{l};X))^2; \qquad
2   z(\tilde{l};X)) = \tilde{l} - X (\tilde{l}/2)^{1/3} + {X^2 \over 6} (\tilde{l}/2)^{-1/3}   + {\rm O}(\tilde{l}^{-5/3}).
\end{equation}
We note furthermore that
\begin{equation}\label{2.1c}
Q({l};\tilde{t}) = 4 z^2, \quad Q({l};0) =  {l}^2,
\end{equation}
where in the first of these $z$ refers to (\ref{2.1a}) and $\tilde{t}$ is from \eqref{1.1u}, and we note too that $Q(\tilde{l};X)$ is a decreasing
function of $X$.

The quantity $E_2^{\rm hard}$ in (\ref{e:lBoxE2Hard}) permits a Fredholm determinant form analogous
to the first line in (\ref{1.1f}) \cite{Fo93a}
\begin{equation}\label{2.1d}
E_2^{\rm hard}(0;(0,4z^2);l) = \det \Big( \mathbb I - \mathbb K_{(0,4z^2)}^{{\rm hard},l} \Big),
\end{equation}
where $\mathbb{K}_{(0,4z^2)}^{{\rm hard},l}$ is the integral operator on $(0,4z^2)$ with kernel 
\begin{equation}
\label{e:KhardxyFull} K_{}^{\rmhard ,a} (x,y) = \frac{J_a (x^{1/2}) y^{1/2} J_a' (y^{1/2}) - J_a' (x^{1/2}) x^{1/2} J_a (y^{1/2})}{2(x-y)}
\end{equation}
and $J_a (x)$ is the Bessel function of the first kind. It is a standard result in the theory
of Fredholm integral equations \cite{WW65} that the determinant in (\ref{2.1d}) can be
expanded as a sum over $k$-dimensional integrals, with the integrand a $k \times k$
determinant with entries (\ref{e:KhardxyFull})
\begin{equation} \label{e:E2HardPf1} 
E_{2}^{\rmhard} \Big(0; (0, 4 z^2 ); l \Big)
 =1 + \sum_{n=1}^{\infty} \frac{(-1)^n}{n!} \int_0^{4z^2} dx_1 \cdots \int_0^{4 z^2} dx_n \det\left[ K^{\rmhard, l} (x_j, x_k) \right]_{j,k=1}^n.
\end{equation}

Now require that $l$ and $z$ are related by (\ref{2.1a}).
We next change variables in each integrand of the series in (\ref{e:E2HardPf1}), $x_l = Q(l;X)$. Taking into consideration
(\ref{2.1c}) the integral in the $n$-th term reads 
 \begin{equation}\label{2.1e}
(-1)^n \int_{\tilde{t}}^{l^2} d X_1 \,  Q'(l; X_1)  \cdots  \int_{\tilde{t}}^{l^2} d X_n \, Q'(l; X_n) \det \Big [ K^{\rmhard, l}(Q(l; X_j), Q(l; X_k)) \Big ]_{j,k=1}^n.
\end{equation}
The point here is that it follows from asymptotic expansions associated with the functional form 
(\ref{e:KhardxyFull}) that for large $l$ the integrand is of order unity in the neighbourhood of
the lower terminal of integration only.
These asymptotic expansions \cite[(9.3.23) \& (9.3.27)]{AS72} give that for large $\nu$
 \begin{align}\label{2.1f}
 J_{\nu}(\nu + u \nu^{1/3}) & \sim {2^{1/3} \over \nu^{1/3}} {\rm Ai}(-2^{1/3} u) \sum_{k=0}^\infty
 {P_k(u) \over \nu^{2k/3}} +    {2^{1/3} \over \nu } {\rm Ai}'(-2^{1/3} u) \sum_{k=0}^\infty
 {Q_k(u) \over \nu^{2k/3}}  \nonumber \\
  J_{\nu}' (\nu + u \nu^{1/3}) & \sim - {2^{2/3} \over \nu^{2/3}} {\rm Ai}' (-2^{1/3} u) \sum_{k=0}^\infty
 {R_k(u) \over \nu^{2k/3}} +    {2^{1/3} \over \nu^{4/3} } {\rm Ai}(-2^{1/3} u) \sum_{k=0}^\infty
 {S_k(u) \over \nu^{2k/3}},
  \end{align} 
 for certain polynomials $P_k(u), Q_k(u), R_k(u), S_k(u)$ of increasing degree. Moreover these expansions
 are uniform for $u \in (-\infty, u_0]$ for any fixed $u_0$. Specifically, upon inserting the explicit values
 of these polynomials for low order, and slightly changing the notation,
 \begin{align}
 J_l \Big( l- x(l/2)^{1/3} \Big) &\mathop{\sim}\limits_{l \to \infty} \frac{2^{1/3}}{l^{1/3}} \rmAi(x) + \frac{1}{10 l}\left( 2x \rmAi(x) + 3x^2 \rmAi' (x)   \right)
+ {\rm O}\Big ( {1 \over l^{5/3}} \Big ) {\rm O}(e^{-x}) \nonumber
\\
\label{e:JdAsympt} J_l' \Big( l- x(l/2)^{1/3} \Big) &\mathop{\sim}\limits_{l \to \infty} -\frac{2^{2/3}}{l^{2/3}} \rmAi' (x) - \frac{2^{1/3}}{10l^{4/3}} \Big( 8x \rmAi'(x) + \left(3x^3 +2 \right) \rmAi(x) \Big)
+ {\rm O}\Big ( {1 \over l^{2}} \Big ) {\rm O}(e^{-x}),
\end{align}
uniformly valid for $x \in [x_0, \infty)$. Recalling the form of the numerator in (\ref{e:KhardxyFull}), we see in particular that
\begin{align}\label{2.1g}
&J_l  ( x^{1/2} ) y^{1/2} J_l'  ( y^{1/2} ) \Big |_{x^{1/2} \mapsto l - x(l/2)^{1/3} \atop y^{1/2} \mapsto l - y(l/2)^{1/3}} 
 \mathop{\sim}\limits_{l \to \infty} - 2\rmAi(x)\rmAi' (y)   + \frac{2}{5 l^{2/3}} \Bigg[ \frac{(y -x)}{2^{1/3}} \rmAi(x)\rmAi'(y)  \nonumber \\
& - \left( 2^{-1/3} +\frac{3y^3}{2^{4/3}} \right) \rmAi(x) \rmAi(y) - \frac{3x^2}{2^{4/3}} \rmAi'(x) \rmAi'(y) \Bigg]
+ {\rm O}\Big ( {1 \over l^{4/3}} \Big ) {\rm O}(e^{-x})  {\rm O}(e^{-y}).
\end{align}

For applicability to (\ref{2.1e}), taking into consideration the second equation in (\ref{2.1b}) and (\ref{e:KhardxyFull}), we see that we require in
(\ref{e:JdAsympt}) that
 \begin{equation}\label{2.1h}
 x = x(l) = X \Big ( 1 - (X/6) (l/2)^{-2/3} + {\rm O}(l^{-2})  \Big ).
 \end{equation}
To account for this in (\ref{2.1g}) we must use the Taylor expansions with bounds on error terms valid for $x \in [x_0,\infty)$
  \begin{align}\label{2.1i}
 \rmAi \left( x+ \frac{a}{2^{1/3} l^{2/3}} \right) & \mathop{\sim}\limits_{l \to \infty} \rmAi (x)+ \frac{a}{2^{1/3} l^{2/3}}\rmAi' (x) + {\rm O}(l^{-4/3})  {\rm O}(e^{-x}), \nonumber \\
\rmAi' \left( x+ \frac{a}{2^{1/3} l^{2/3}} \right)  & \mathop{\sim}\limits_{l \to \infty} \rmAi' (x)+ \frac{ax}{2^{1/3} l^{2/3}}\rmAi (x)  + {\rm O}(l^{-4/3})  {\rm O}(e^{-x}),
\end{align}
where we made use of the differential equation satisfied by the Airy function $\rmAi''(x)= x \rmAi(x)$ in deriving the second expression.
We can now use (\ref{2.1g}) in (\ref{e:KhardxyFull}) to conclude that for large $l$
  \begin{multline}\label{2.1j} 
-(Q' (l; X_j)  Q'(l; X_k) )^{1/2}  K^{\rmhard, l}(Q(l; X_j), Q (l; X_k))  \\ \mathop{\sim}\limits_{l \to \infty} 
K^{\rmsoft} (X_j,X_k) +L(X_j,X_k) l^{-2/3} +  {\rm O}( l^{-4/3}) {\rm O}(e^{-X_j})  {\rm O}(e^{-X_k}).
 \end{multline}
 Substituting in (\ref{2.1e}) with the upper terminals replaced by $\infty$ (this is permissible by the error bounds)
 gives that for large $l$, and $z$ related to $l$ by (\ref{2.1a}),
  \begin{multline}\label{2.1k}  
 E_{2}^{\rmhard} \Big(0; (0, 4 z^2 ); l \Big) = 1 + \\
 \sum_{n=1}^{\infty} \frac{(-1)^n}{n!} \int_{\tilde{t}}^{\infty} dX_1 \cdots \int_{\tilde{t}}^{\infty} dX_n \det\left[ K^{\rmsoft} (X_j, X_k) +   l^{-2/3} L(X_j,X_k) \right]_{j,k=1}^n+
  {\rm O}( l^{-4/3} ). 
\end{multline} 
Recalling now (\ref{e:lBoxE2Hard}),  then rewriting the RHS of (\ref{2.1k}) as in the reverse of going from (\ref{2.1d}) to (\ref{e:E2HardPf1}), this tells us that
for large $l$
 \begin{equation}\label{2.1l}
  \Pr \left(\frac{l^{\Box} -2z}{z^{1/3}} \leq t \right) = 
   \det \Bigg( \mathbb I - \Big( \mathbb{K}^{\rmsoft}_{(\tilde{t}, \infty)} + l^{-2/3} \mathbb{L}_{(\tilde{t}, \infty)} \Big) \Bigg)  + {\rm O}( l^{-4/3} ).
\end{equation}
The stated result (\ref{1.1m}) now follows from \cite[Lemma 1]{BFM17}, and in the term proportional to $z^{-2/3}$ replacing
$\tilde{t}$ by $t$, which is valid since they are equal to leading order in $z$.

 \subsection{A differential equation characterisation of  $F_{2,1}^{\rm H}(t)$}\label{S2.2}
 As mentioned below (\ref{e:Corr2t}), the quantity $ F_{2,1}^{\rm H}(t) $ in (\ref{1.1m}),
 defined in terms of  Fredholm integral operators in (\ref{e:Corr2t}), also permits a characterisation
 as the solution of a particular second order linear differential equation, with coefficients given
 in terms of a particular ($\sigma$ form) Painlev\'e II transcendent. In preparation, we first recall
 that an alternative to the second expression in (\ref{1.1f}) is the evaluation \cite{TW94a} (see also
 \cite[\S 8.3.2]{Fo10})
 \begin{equation}\label{2.2a} 
 E_{2}^{\rmsoft} \Big(0; (s,\infty) \Big) = \exp \left( -\int_s^{\infty} u_0(r) dr\right),
\end{equation}
where $u_0(r)$ satisfies the particular $\sigma$-PII equation  and boundary condition
\begin{align}
\label{e:LeadDE} (u'')^2 + 4 u' \Big( (u')^2 -r u' +u \Big) =0, \quad  u_0 (r) \mathop{\sim}\limits_{r \to \infty} {\rm Ai}'(r)^2 - r  {\rm Ai}(r)^2.
\end{align}

\begin{prop}\label{p:Corr2DE}
Consider the quantity $ F_{2,1}^{\rm H}(t) $ in the expansion (\ref{1.1m}). As an alternative to
 (\ref{e:Corr2t}) we have
  \begin{equation}\label{2.2c} 
  F_{2,1}^{\rm H}(t) = -\exp \left(- \int_{t}^{\infty} u_0 (r) dr \right) \left( \int_{t}^{\infty} u_1 (r) dr \right).
\end{equation}
Here $u_0$ is specified by (\ref{2.2a}). The function $u_1$ is specified as the solution of
the inhomogeneous second order linear differential equation
\begin{align}
\label{e:SubleadDE} A_2 (r) u_1'' + B_2(r) u_1' + C_2 (r) u_1 = D_2 (r)
\end{align}
with the coefficients 
\begin{align}
A_2 (r) & = u_0'' (r), \quad 
B_2 (r)  = 2u_0 (r) - 4r u_0' (r) + 6 (u_0' (r))^2 , \quad
C_2 (r) = 2 u_0' (r),      \nonumber \\
D_2 (r) &= -\frac{1}{3 (2^{1/3})} \Bigg[ u_0' (r) \Big( r^2 u_0' (r) + 6u_0 (r) u_0' (r) -2r u_0 (r) + 3u_0'' (r) \Big) -2(u_0 (r))^2 \Bigg],
\end{align}
and with boundary condition
\begin{align}
\label{e:u1tinfty} u_1(r) \mathop{\sim}\limits_{r \to \infty} \frac{1}{(2^{1/3}) 30} \Bigg[ 12 \rmAi(r) \rmAi'(r) + 3 r^2 (\rmAi(r))^2 - 2r (\rmAi'(r))^2 \Bigg].
\end{align}

\end{prop}

\proof 
We require knowledge \cite{TW94b} (see also
 \cite[\S 8.3.3]{Fo10}) of an alternative to (\ref{2.1d}), telling us that
\begin{align}
\label{e:E2hardDEa} E_2^{\rmhard} (0; (0,s) ;a) = \exp \left( \int_0^{s} \frac{v (r; a)}{r} dr \right),
\end{align}
where $v$ satisfies the particular $\sigma$-PIII$'$ equation and boundary condition
\begin{align}
\label{e:vDE} (r v'')^2 - (a v')^2-v' (4 v' +1) (v - r v' )=0, \quad  v(r;a) \mathop{\sim}\limits_{r\to 0^+} -\frac{r^{1+a}} {2^{2(1+a)} \Gamma (1+a) \Gamma (2+a)}.
\end{align}

With $z$ related to $t$ and $l$ by (\ref{2.1a}), and with $Q(l;X)$ given by (\ref{2.1b}), we can change variables
in (\ref{e:E2hardDEa}) to obtain
\begin{equation}\label{2.2d} 
E_2^{\rmhard} (0; (0, 4z^2) ;l)  = \exp \bigg ( - \int_{\tilde{t}}^{l^2} {v (Q(l;s)) \over Q(l;s)} Q'(l;s) \, ds  \bigg ),
\end{equation}
where $\tilde{t}$ is defined in \eqref{1.1u}. To be consistent with  (\ref{1.1m}) and (\ref{2.2a}) we must have that for large $l$
\begin{equation}\label{2.2e} 
 {v (Q(l;s)) \over Q(l;s)} Q'(l;s) = u_0(s) + {u_1(s) \over l^{2/3}} + \cdots
\end{equation} 
Rearranging this gives a functional form for $v (Q(l;s))$, which is to be substituted in (\ref{e:vDE}) after first
changing variables $r = Q(l;s)$. These steps are readily carried out using computer algebra.  Equating terms
at leading powers of $l$ gives the equation (\ref{e:LeadDE}) for $u_0$ at order $l^0$, and the differential equation (\ref{e:SubleadDE}) at order
$l^{-2/3}$.

In relation to the boundary condition, we reconsider the above working, and the working which gave (\ref{2.1l}),
in the case of $l$ continuous. The only difference is that the discrete variable $\tilde{t}$ should be replaced by the
continuous variable $t$. Taking the logarithmic derivative of the RHS of (\ref{2.1l}) in this setting gives
 \begin{multline}\label{2.2f} 
\frac{d}{dt} \log \det \Bigg( \mathbb I  - \Big( \mathbb{K}^{\rm soft}_{(t, \infty)} + l^{-2/3} \mathbb{L}_{(t, \infty)} \Big) \Bigg) = \frac{d}{dt} \Tr \log \Bigg( \mathbb I - \Big( \mathbb{K}^{\rm soft}_{(t, \infty)} + l^{-2/3} \mathbb{L}_{(t, \infty)} \Big) \Bigg)\\
\mathop{\sim}\limits_{t, l \to \infty }  - \frac{d}{dt} \int_t^{\infty} \Big ( K^{\rmsoft} (x,x) + L(x,x) l^{-2/3}\Big )  dx
= K^{\rmsoft} (t,t) + L(t,t) l^{-2/3}.
\end{multline}
On the other hand, it follows by substituting (\ref{2.2e}) in (\ref{2.2d})  that  this same quantity is also equal to
\begin{align}\label{ep.2}
u_0 (t) + \frac{u_1 (t) } {l^{2/3}}.
\end{align}
Comparing (\ref{2.2f}) and (\ref{ep.2}) at leading order gives $u_0 (t) \mathop{\sim}\limits_{t \to \infty} K^{\rmsoft} (t,t)$,
and thus the boundary condition for $u_0$ in (\ref{e:LeadDE}).
At ${\rm O}(l^{-2/3})$ this comparison
gives $u_1 (t) \mathop{\sim}\limits_{t \to \infty} L (t,t)$. Taking the limit $x\to y=r$ 
in \eqref{d:Lxy} we then obtain \eqref{e:u1tinfty} for $u_1$.

\hfill $\Box$

As  noted in the Introduction, earlier Baik and Jenkins \cite{BJ13} obtained an evaluation of $F_{2,1}^{\rm H}(t)$ relating to
Painlev\'e transcendents. This is simpler than our (\ref{2.2c}) as it involves only $u_0(t)$, 
\begin{equation}\label{B.1}
F_{2,1}^{\rm H}(t) = - {2^{2/3} \over 10} \bigg ( {d^2 \over d t^2} + {t^2 \over 6} {d \over dt} \bigg ) \exp \Big ( - \int_t^\infty u_0(r) \, dr \Big ).
\end{equation} 
Comparing with (\ref{2.2c}), it follows that we must have
\begin{equation}\label{B.2}
u_1(r) = - {2^{2/3} \over 10} \bigg ( {d^2 u_0(r) \over d r^2} + \Big ( 2 u_0(r) + {r^2 \over 6} \Big )  {d u_0(r) \over d r} + {r \over 3} u_0(r) \bigg ).
\end{equation} 
A similar circumstance arose in the study \cite[discussion below (3.42)]{FPTW19}, which suggests how (\ref{B.2}) can be
verified from the characterisation of $u_1(r)$ in Proposition \ref{p:Corr2DE}.

First, we verify from the boundary condition of $u_0(r)$ in (\ref{e:LeadDE}), and that of $u_1(r)$ in 
(\ref{e:u1tinfty}) that they are compatible with (\ref{B.2}). It remains then to verify that (\ref{B.2}) satisfies
the differential equation (\ref{e:SubleadDE}). This can be done be direct substitution, then substituting for
the third and fourth derivatives of $u_0(r)$. In relation to the latter, we note
 that differentiating the differential equation (\ref{e:LeadDE}) for $u_0$, and simplifying, gives us
\begin{equation}\label{B.3}
u_0'''(r) = -2u_0(r) + 4 r u_0'(r) - 6 (u_0'(r))^2.
\end{equation}
Again differentiating this equation, and making further use of  (\ref{e:LeadDE}), we can express the fourth derivative of $u_0(r)$ in terms of $u_0(r)$ and $u_0'(r)$. Once the substitutions have been performed, the resulting equation only
involves the second derivative of $u_0(r)$ in the form of $(u_0'' (r))^2$, which we eliminate in
favour of $u_0(r)$ and $u_0'(r)$ using (\ref{e:LeadDE}). These steps, performed using computer
algebra, verify that (\ref{B.2}) solves (\ref{e:SubleadDE}), as required.

 \subsection{Comparison with numerical calculations}\label{sec:NumsBeta2}
 
We numerically calculate the correction term $F_{2,1}^{\rm H}(t)$ in \eqref{1.1m} using both the integral operator characterisation of Proposition \ref{P1} and the expression \eqref{2.2c} in terms of the solution to a differential equation. We compare these to calculations of the difference
\begin{align}
\label{d:F21ell}  \delta_2^{\rm H} (t) := l^{2/3} \Bigg( E_{2}^{\rmhard} \Big( 0; (0,Q(l;t)); l \Big) - E_{2}^{\rmsoft} \Big( 0; (t,\infty) \Big) \Bigg)
\end{align}
for $l =20$, where, for the purposes of comparison, we use the continuous variable $t$. 

To numerically evaluate the integral operators we use the Fredholm determinant Matlab toolbox by Folkmar Bornemann \cite{Bornemann2010}, and a Mathematica implementation by Allan Trinh,
coauthor on some related works along the theme of finite size corrections to limit formulas in random matrix theory
\cite{FT18,FT19,FT19a,FLT20}.  For the DE solutions $u_0 (r), u_1 (r)$ needed for \eqref{2.2c} we use a sequence of Taylor series expanded about various $r$ points, beginning on the right (near $+\infty$, to match the DE boundary conditions) and proceeding to the left. For $u_0(r)$ we use a sequence of 600 series of degree 11, while for $u_1(r)$ we use a sequence of 500 series of degree 6. To calculate the finite $l=20$ correction \eqref{d:F21ell} we also need a sequence of Taylor series solutions for $v(r; 20)$ from \eqref{e:vDE} --- we use a sequence of 15,446 series of degree $10$. The results are plotted in the left panel of Figure \ref{f:Corrapproxell20_2}. In the right panel we plot a numerical estimate of the next order correction in \eqref{1.1m} by calculating
\begin{align}
\label{d:delta2} E_{2}^{\rmhard} \Big( 0; (0,Q(l;t)); l \Big) - E_{2}^{\rmsoft} \Big( 0; (t,\infty) \Big) - \frac{1}{l^{2/3}} F_{2,1}^{\rm H}(t).
\end{align}

\begin{figure}
\centering
     \begin{minipage}{0.45\textwidth}
	\begin{tikzpicture}[scale=1, every node/.style={transform shape}]
         \node[inner sep=0pt] at (0,0) {\includegraphics[height=4cm, align=t]{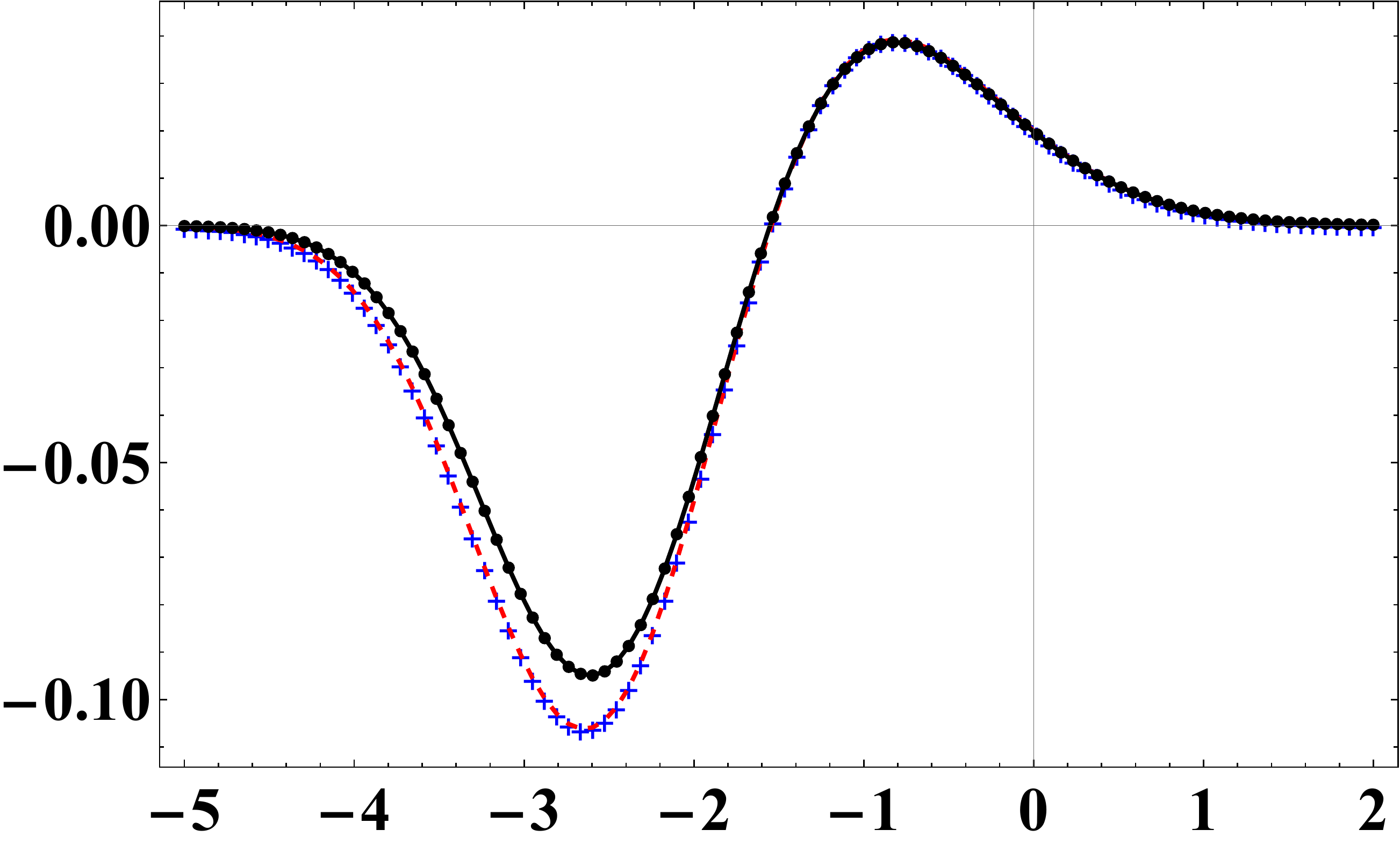}};
         \node at (3.6,-1.7) {\textbf{\textit{t}}};
         \node at (0.3,2.4) {$F_{2,1}^{\rm H}(t)$};
         \end{tikzpicture}
	\end{minipage} \quad
	\begin{minipage}{0.45\textwidth}
	\begin{tikzpicture}[scale=1, every node/.style={transform shape}]
         \node[inner sep=0pt] at (0,0) {\includegraphics[height=4cm, align=t]{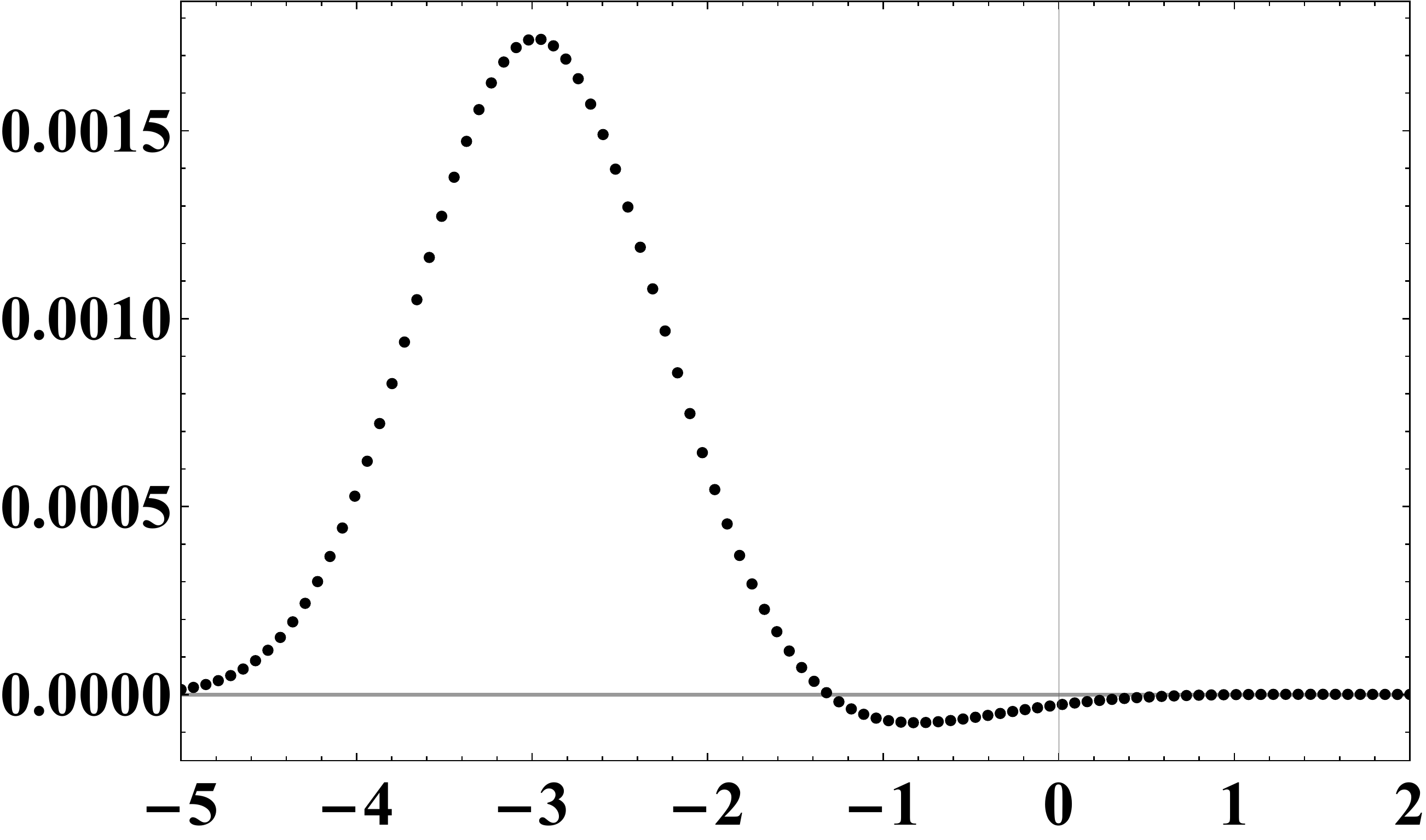}};
         \node at (3.6,-1.7) {\textbf{\textit{t}}};
         \node at (0.3,2.4) {Expression \eqref{d:delta2}};
         \end{tikzpicture}
	\end{minipage}
\caption{In the left panel we have the correction term $F_{2,1}^{\rm H}(t)$ in \eqref{1.1m} calculated using \eqref{2.2c} [blue crosses], and using \eqref{e:Corr2t} [red dashed line]. On the same axes we also plot $\delta_2^{\rm H} (t)$, the difference \eqref{d:F21ell} for $l=20$, using the Fredholm determinant expressions \eqref{2.1d} and the first equality in \eqref{1.1f} [black line] and also using the expressions in terms of solutions to differential equations \eqref{e:E2hardDEa} and the second equality in \eqref{1.1f} [black dots]. In the right panel we plot the difference \eqref{d:delta2}, which is a numerical approximation to the unknown $O(z^{-4/3})$ term in \eqref{1.1m}.}
\label{f:Corrapproxell20_2}
\end{figure}
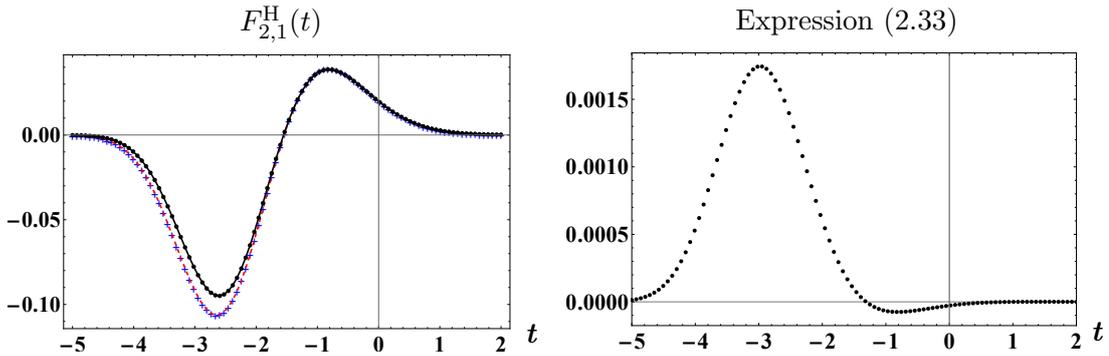

\section{Large $z$ expansion of $\Pr \left(\frac{{\linc} -2z}{z^{1/3}} \leq t \right)$ and  $   \Pr \left(\frac{{\ldec} -2z}{z^{1/3}} \leq t \right)$ }
 \subsection{Fredholm determinant form}
 The analogues of (\ref{e:lBoxE2HardAv}) and (\ref{e:lBoxE2Hard}) are the formulas
 \begin{align}
\label{e:lBoxE4Hard} \Pr(\linc \leq l) &= e^{-z^2/2}\left\langle e^{z \Tr \bU} \right\rangle_{\bU \in O(l)}= \tilde{E}_4^{\rmhard} \Big( 0; (0, 4z^2); l \Big)\\
\label{e:lBoxE1Hard} \Pr(\ldec \leq 2l) &= e^{-z^2/2}\left\langle e^{z \sum_{j=1}^l 2 \cos \theta_j} \right\rangle_{\mathrm{Sp} (2l)}= E_1^{\rmhard} \Big( 0; (0, 4z^2); l \Big).
\end{align}
The first equality in both is due to Rains \cite{Ra98}, while the second were found in \cite{FW04}. We note too that the validity of the second formula in (\ref{e:lBoxE4Hard}) 
as derived in \cite{FW04} is restricted to
$l$ even. However, by the different strategy of expressing both the average over $\bU \in O(l)$ and $\tilde{E}_4^{\rmhard}$ in terms of a generalised
hypergeometric function of $l$ variables based on zonal polynomials --- see \cite{Ja64} in relation to the former and
\cite{Fo93c} in relation to the latter --- the validity can be established independent of the parity of $l$.
The use of a tilde in the notation $ \tilde{E}_4^{\rmhard}$ indicates the use of
 a rescaling of the natural hard edge  Laguerre  $\beta$ ensemble variables;
see \cite[second displayed equation \S 9.8]{Fo10} or (\ref{9.1b}) below.
The analogues of (\ref{2.1d}) are the formulas  \cite{DF06,Fo06}
\begin{align}
\label{e:FredDetHard1} E_1^{\rmhard} \left( 0; (0,s);\frac{a-1}{2} \right)&= \det \left( \mathbb I - \mathbb{V}^{{\rm hard}, a}_{s, (0,1)} \right)\\
\label{e:FredDetHard4} \tilde{E}_4^{\rmhard} \Big( 0; (0,s); a+1 \Big)&= \frac{1}{2} \left[ \det \left( \mathbb I - \mathbb{V}^{{\rm hard}, a}_{s, (0,1)} \right) + \det \left( \mathbb I + \mathbb{V}^{{\rm hard},a}_{s, (0,1)} \right) \right],
\end{align}
where $\mathbb{V}^{{\rm hard}, a}_{s, (0,1)}$ is the integral operator on $(0,1)$ with kernel
\begin{align}
V_{s}^{{\rm hard},a} (x,y):= \frac{\sqrt{s}}{2} J_a (\sqrt{xys}).
\end{align}

In relation to the probabilities in (\ref{e:lBoxE4Hard}) and (\ref{e:lBoxE1Hard}) we have the known limit theorems (\ref{1.1n}) and (\ref{1.1o})
in terms of certain soft edge gap probabilities. As for $E_2^{\rm soft}$ the latter admit evaluations in terms of Fredholm determinants,
and Painlev\'e transcendents. The Fredholm determinant forms read \cite{Sa05,Fo06}
\begin{align}
\label{e:E1softFD} E_{1}^{\rmsoft} \Big(0; (s,\infty) \Big)&= \det\left( \mathbb I - \mathbb{V}_{s, (0, \infty)}^{\rmsoft} \right) \\
\label{e:E4softFD} \tilde{E}_{4}^{\rmsoft} \Big(0; (s,\infty) \Big)&= \frac{1}{2} \left[ \det\left( \mathbb I - \mathbb{V}_{s, (0, \infty)}^{\rmsoft} \right) + \det\left( \mathbb I + \mathbb{V}_{s, (0, \infty)}^{\rmsoft} \right) \right],
\end{align}
where  $\mathbb{V}_{s, (0, \infty)}^{\rmsoft}$ is the integral operator on $(0,\infty)$ with kernel
\begin{align}
\label{d:Vsoft} V^{\rmsoft}_{s} (x,y) := {\rm Ai} (x+y+s).
\end{align}
The forms in terms of Painlev\'e transcendents, assuming (\ref{2.2a}) or the second expression in (\ref{1.1f}), are \cite{TW96}
\begin{align}
\label{e:E1softDE} E_1^{\rmsoft} \Big( 0; (s, \infty) \Big)&= E_2^{\rmsoft} \Big( 0; (s,\infty) \Big)^{1/2} \exp\left( - \frac{1}{2} \int_s^{\infty} q_0(r) dr \right) \\
\label{e:E4softDE} \tilde{E}_4^{\rmsoft} \Big( 0; (s, \infty) \Big)&= E_2^{\rmsoft} \Big( 0; (s,\infty) \Big)^{1/2} \cosh \left( \frac{1}{2} \int_s^{\infty}   q_0(r) dr \right),
\end{align}
where $q_0(r)$ satisfies the particular PII equation and boundary condition
as noted below (\ref{1.1g}).

We will first consider the Fredholm determinant forms and obtain the analogues of Proposition \ref{P1}.

\begin{prop}\label{P1a}
 For large $z$ we have
 \begin{equation}\label{2.1mX}   
   \Pr \left(\frac{{\ldec} +1 -2z}{z^{1/3}} \leq t \right) = F_{1,0}^{ {\rm H}}(\tilde{t}) + {1 \over (2z)^{2/3}} F_{1,1}^{{\rm H}}(t) + {\rm O}(z^{-4/3})
  \end{equation}
  with $\tilde{t}$ specified by (\ref{1.1u}), $ F_{1,0}^{{\rm H}}(\tilde{t})  = E_1^{\rm soft}(0,(\tilde{t},\infty))$ and 
  \begin{align}
 \label{e:Corr1t1}    F_{1,1}^{{\rm H}}(t) &= -\det \left( \mathbb I - \mathbb{V}_{t, (0, \infty)}^{\rmsoft} \right) \Tr \left( ( \mathbb I - \mathbb{V}_{t, (0, \infty)}^{\rmsoft})^{-1} \mathbb{M}_{t, (0, \infty)} \right),
\end{align}
where $ \mathbb{V}_{t, (0, \infty)}^{\rmsoft}$ is specified as in (\ref{e:E1softFD}) and $\mathbb{M}_{t, (0, \infty)}$ is the integral operator on $(0, \infty)$ with kernel
\begin{align}
\nonumber M_t (x,y)&= \frac{1}{(2^{1/3}) 10} \Bigg[ \Big( 2x + 2 y -8t \Big) \rmAi(t+ x + y) \\
\label{d:Mtkernel} & + \frac{1}{3} \Big( 24 x^2+ 24y^2 -12x t -12 x y -12 y t -t^2 \Big) \rmAi'(t+ x + y) \Bigg].
\end{align}
\end{prop}   

\begin{proof}
We start with the Fredholm determinant \eqref{e:FredDetHard1}
\begin{multline}
E_1^{\rmhard} \left( 0; (0,4z^2);\frac{l-1}{2} \right)= \det \left( \mathbb I - \mathbb{V}^{\rmhard, l}_{4 z^2, (0,1)} \right)\\
=1+ \sum_{n=0}^{\infty} \frac{(-1)^n}{n!} \int_0^{1} dx_1 \dots \int_0^{1} dx_n \det \left[   z J_l ( 2 z\sqrt{ x_j x_k}) \right]_{j,k=1}^n.
\end{multline}
In each variable in the integrand of the series we change variables $x_j=  1 - X_j (2/l)^{2/3}$, which transforms the
corresponding multi-dimensional integral to read
 \begin{multline}\label{2.1n}   
 (-1)^n    \int_0^l d X_1 \,  \cdots
 \int_0^l  d X_n \,   
  \det \left[  (2/l)^{2/3} z J_l \Big ( 2 z \sqrt{ (1 - X_j (2/l)^{2/3})   (1 - X_k (2/l)^{2/3}) } \Big ) \right]_{j,k=1}^n.
  \end{multline}
  Introducing now $z = z(l;\tilde{t})$ as expanded for large $l$ according to (\ref{2.1b}) with $X = \tilde{t}$ shows that the
  argument of the Bessel function in (\ref{2.1n}) has the large $l$ expansion
  \begin{align}
l- \left( \frac{l}{2} \right)^{1/3} \left( \tilde{t}+ X_j+ X_k + \frac{\gamma}{2^{1/3} l^{2/3}} \right) + O(l^{-1}),
\end{align}
with
\begin{align}
\gamma:= \frac{X_j^2}{2}+ \frac{X_k^2}{2} - X_j  X_k  - X_j \tilde{t} - X_k \tilde{t}- \frac{\tilde{t}^2}{3},
\end{align}
which from the first formula in \eqref{e:JdAsympt} gives the large $l$ behaviour of the Bessel function in (\ref{2.1n}) itself
\begin{multline}
 \left( \frac{2}{l} \right)^{1/3} \rmAi \left( \tilde{t}+ X_j + X_k  + \frac{\gamma}{2^{1/3} l^{2/3}} \right )
+ \frac{1}{10 l} \bigg [ 2 \left( \tilde{t}+ X_j + X_k \right) \rm Ai(\tilde{t}+ X_j + X_k)  \\
 + 3 \left( \tilde{t}+ X_j + X_k \right)^2 \rmAi'(\tilde{t}+ X_j  + X_k ) \bigg ]  + {\rm O}(l^{-5/3} ) {\rm O}(e^{-X_j - X_k}).
\end{multline}

Now making further use of the large $l$ expansion of $z = z(l; \tilde{t})$  we see from this  that the argument of the determinant
in (\ref{2.1n}) has the large $l$ expansion
\begin{multline}
\rmAi \left( \tilde{t}+ X_j + X_k + \frac{\gamma}{2^{1/3} l^{2/3}} \right) + 
 \frac{1}{(2^{1/3}) 10 l^{2/3}} \bigg [ (2 X_j + 2 X_k  -8\tilde{t} )\rmAi(\tilde{t}+ X_j  +  X_k ) \\ + 3 \left( \tilde{t}+ X_j  +  X_k  \right)^2 \rmAi'(\tilde{t}+ X_j + X_k) \bigg ] + {\rm O}(l^{-4/3} ) {\rm O}(e^{-X_j - X_k}).
\end{multline}
Expanding the argument of the Airy function in the first term according to the first formula in (\ref{2.1i}) shows that this reduces to
 \begin{equation}
\label{e:Vlimkernel}  V_{\tilde{t}}^{\rmsoft} (X_j, X_k) + M_{\tilde{t}} (X_j, X_k) l^{-2/3} + {\rm O}(l^{-4/3} ) {\rm O}(e^{-X_j - X_k}).
 \end{equation}
 The result (\ref{2.1mX}) now follows  from \cite[Lemma 1]{BFM17}, where, as in the proof of Proposition \ref{P1}, we replace $\tilde{t}$ by $t$ in the second-order term since they are of the same order in $l$.

\end{proof}

We see from (\ref{e:FredDetHard4}) that knowledge of the scaled asymptotics of the Fredholm determinant in (\ref{e:FredDetHard1}) is sufficient to compute the same for the quantity $ \tilde{E}_{4}^{\rmhard}$ and thus from (\ref{e:lBoxE4Hard}) the scaled asymptotics of Pr$(\linc \le l)$.

\begin{cor}\label{C3.2}
 For large $z$ we have
 \begin{equation}\label{2.1m}   
   \Pr \left(\frac{\linc -1 -2z}{z^{1/3}} \leq t \right) = F_{4,0}^{ {\rm H}}(\tilde{t}) + {1 \over (2z)^{2/3}} F_{4,1}^{{\rm H}}(t) + {\rm O}(z^{-4/3}),
  \end{equation}
  where $ F_{4,0}^{{\rm H}}(\tilde{t})  = \tilde{E}_4^{\rm soft}(0;(\tilde{t},\infty))$ and 
\begin{multline}
\label{e:Corr4t} F_{4,1}^{{\rm H}}(t) = \frac{1}{2} \left[ \det \left( \mathbb I + \mathbb{V}_{t, (0, \infty)}^{\rmsoft}\right) \Tr \left( (\mathbb I  + \mathbb{V}_{t, (0, \infty)}^{\rmsoft})^{-1} \mathbb{M}_{t, (0, \infty)} \right) \right.\\
\left. -\det \left( \mathbb I - \mathbb{V}_{t, (0, \infty)}^{\rmsoft}\right) \Tr \left( ( \mathbb I  - \mathbb{V}_{t, (0, \infty)}^{\rmsoft})^{-1} \mathbb{M}_{t, (0, \infty)} \right) \right].
\end{multline}
\end{cor}

\begin{remark}
For general $\beta > 0$, define the Laguerre $\beta$ ensemble in terms of the eigenvalue PDF proportional to
\begin{equation}\label{9.1a}
\prod_{l=1}^N x_l^a e^{-\beta x_l/2} \prod_{1 \le j < k \le N} | x_k - x_j|^\beta, \qquad x_l \ge 0.
 \end{equation}
Let $E_{\beta, N}^{\rm L}(0;(0,t);a)$ denote the probability that the interval $(0,t)$ has no eigenvalues in this ensemble.
The corresponding hard edge scaled limit is specified by
\begin{align}
E_\beta^{\rm hard}(0;(0,t);a) := \lim_{N \to \infty} E_{\beta, N}^{\rm L}(0;(0,t/4N);a).
\end{align}
With $\beta = 1,2$ this agrees with the meaning of $E_1^{\rm hard}$ and $E_2^{\rm hard}$ as appear above, while
\begin{equation}\label{9.1b}
\tilde{E}_4^{\rm hard}(0;(0,t);a) := \lim_{N \to \infty} E_{4, N/2}^{\rm L^*}(0;(0,t/4N);a),
 \end{equation}
 where ${\rm L}^*$ refers to the ensemble with eigenvalue PDF proportional to (\ref{9.1a}) but with $e^{-\beta x_l/2}$ replaced by $e^{-x_l}$. 
 As used in \cite{DF06a} in the context of the spectral density, specify the soft edge scaled limit by
 \begin{equation}\label{9.1c}
 {E}_\beta^{\rm soft}(0;(t,\infty)) := \lim_{N \to \infty}  E_{\beta, N}^{\rm L}(0;(4N +2 (2N)^{1/3} t, \infty);a).
  \end{equation}
  Note here that the quantity on the RHS is the  probability that the interval at the far end of the spectrum $(4N +2 (2N)^{1/3} t, \infty)$ contains no eigenvalues,
  and that the limit is independent of $a$.
  The significance of the value $4N +2 (2N)^{1/3} t$ is that this centres and scales the coordinates so that in the variable $t$
  the largest eigenvalue is near the origin, and has spacing of order unity with its neighbours.
  
  In this random matrix setting, our results suggest that in relation to the hard to soft edge transition, we have that for large $\alpha$
   \begin{equation}\label{9.1d}
   {E}_\beta^{\rm hard}(0;(0,4 z^2); \beta(\alpha + 1 - 2/\beta)) \Big |_{\alpha = 2z + t z^{1/3}} =
   E_\beta^{\rm soft}(0;(t,\infty)) + {\rm O} \Big ( {1 \over \alpha^{2/3}} \Big ),
   \end{equation}  
   with the main point being the order of the correction term. The limit law itself was established in \cite{BF03} for $\beta = 1,2$ and
   4, and, using different techniques, for general $\beta > 0$ in \cite{RR09}. In relation to the correction, as already pointed out in 
   \cite{BJ13} in the context of  the Hammersley process corresponding to the $\beta = 1$ case, the shift $\alpha \mapsto \alpha + 1 - 2/\beta$
   on the LHS is crucial for the optimal rate of convergence ${\rm O}(1/\alpha^{2/3})$.
\end{remark}

 \subsection{Differential equation form}\label{S3.2} 
 The scaled asymptotics  of $E_2^{\rm hard}$ were obtained in the proof of Proposition \ref{p:Corr2DE} in  terms of the solution of a differential equation, starting from
 knowledge of the  Painlev\'e transcendent evaluation (\ref{e:E2hardDEa}).
For $E_1^{\rm hard}$ and $\tilde{E}_4^{\rm hard}$ the analogues of (\ref{e:E2hardDEa}) are \cite{Fo00}
\begin{align}
\label{e:E1hardDE} E_1^{\rmhard} \left( 0; (0, s);\frac{a-1}{2} \right) &= E_2^{\rmhard} (0; (0, s); a)^{1/2} \exp \left(- \frac{1}{4} \int_0^{s} \frac{p_{\rmhard} (r; a)}{\sqrt{r}} dr \right)\\
\label{e:E4hardDE} \tilde{E}_4^{\rmhard} (0; (0,s); a+1) &= E_2^{\rmhard} (0; (0,s); a)^{1/2} \cosh \left( \frac{1}{4} \int_0^s \frac{p_{\rmhard} (r; a)}{\sqrt{r}} dr \right),
\end{align}
where $p_{\rmhard} (r; a)$ satisfies the particular Painlev\'e III$'$ equation and boundary condition
\begin{align}
\label{e:E2hardDEb} r (1- p^2) (rp')' + p (rp')^2 +\frac{1}{4}(r- a^2)p +\frac{1}{4} rp^3 (p^2 -2)=0, \quad  p_{\rmhard} (r; a) \mathop{\sim}\limits_{r\to 0^+} \frac{r^{a/2}}{2^a \Gamma (1+a)}.
\end{align}
From Proposition \ref{p:Corr2DE} we already know how to express the leading two terms in the scaled limit of the factors $E_2^{\rmhard} (0; (0, s); a)$ in (\ref{e:E1hardDE}) and (\ref{e:E4hardDE}) in terms of
quantities satisfying differential equations. Our primary task then is to do the same for the second factor in (\ref{e:E1hardDE}).

\begin{prop}\label{P3.3}
Let $z, l$ and $t$ be related by (\ref{2.1a}), and $\tilde{t}$ be defined by \eqref{1.1u}. We have the large $z$ and $l$ expansion
 \begin{equation}\label{2.1o}   
  \exp \left(- \frac{1}{4} \int_0^{4z^2} \frac{p_{\rmhard} (r; l)}{\sqrt{r}} dr \right) = \exp \Big ( - {1 \over 2} \int_{\tilde{t}}^\infty q_0(r) \, dr \Big )
  \Big ( 1 - {1 \over 2  l^{2/3}} q_1(t) + \cdots \Big ).
 \end{equation}
 Here $q_0$ is specified as in (\ref{1.1f}) and $q_1 (r)$ satisfies the DE
\begin{align}
\label{e:SymmDE} A_1(r) q_1'' + B_1 (r) q_1' +C_1(r) q_1 = D_1(r),
\end{align}
where
\begin{align}
A_1(r)&:= \frac{1}{2},\quad
 B_1(r):= 0, \quad
  C_1(r) := - \frac{r}{2} - 3 q_0^2(r),   \nonumber \\
\label{2.1pDE}  D_1(r)&:= \frac{1}{2^{1/3}} \left( -\frac{r^2 q_0(r)}{12} +r q_0(r)^3 +q_0(r)^5 - \frac{q_0'(r)}{2} -q_0 (r) q_0'(r)^2 \right),
\end{align}
with boundary condition
\begin{align}\label{e:w1BC}
 q_1(r) \mathop{\sim}\limits_{r\to \infty} - \frac{1}{30(2^{1/3})} \Big( 14r \rmAi(r) + r^2 \rmAi'(r) \Big).
\end{align}
Furthermore, for large $l,z$
 \begin{multline}\label{2.1p} 
 E_1^{\rm hard} \left( 0;(0,4z^2); \frac{l-1}{2} \right)  =   \Pr \left(\frac{\linc +1 -2z}{z^{1/3}} \leq t \right)  \\
 = \exp \Big ( - {1 \over 2} \int_{\tilde{t}}^\infty (u_0(r) + q_0(r)) \, dr \Big ) \Big (
 1 - {1 \over 2 l^{2/3}} \int_{{t}}^\infty (u_1(r) + q_1(r)) \, dr + \cdots \Big ),
\end{multline} 
with $u_0$ and $u_1$ from Proposition \ref{p:Corr2DE}, and hence
 \begin{equation}\label{2.1pX} 
 F_{1,1}^{\rm H}(t) = - {1 \over 2}  \exp \Big ( - {1 \over 2} \int_{{t}}^\infty (u_0(r) + q_0(r)) \, dr \Big )   \int_t^\infty (u_1(r) + q_1(r)) \, dr.
 \end{equation}
\end{prop}
 
 \begin{proof}
 Analogous to (\ref{2.2d}), with $Q(l;X)$ given by (\ref{2.1b}), we can change variables to obtain
\begin{equation}\label{2.2dH} 
 \exp \left(- \frac{1}{4} \int_0^{4z^2} \frac{p_{\rmhard} (r; l)}{\sqrt{r}} dr \right) = \exp \bigg ( {1 \over 4}  \int_{\tilde{t}}^{l^2} { p_{\rm hard}(Q(l;s);l) \over \sqrt{Q(l;s)}} Q'(l;s) \, ds  \bigg ).
\end{equation}
To be consistent with  (\ref{2.1mX}) and (\ref{e:E1softDE}) we must have that for large $l$
\begin{equation}\label{2.2eH} 
 {1 \over 2} {  p_{\rm hard}(Q(l;s);l) \over \sqrt{Q(l;s)}} Q'(l;s) = - q_0(s) - {q_1(s) \over l^{2/3}} + \cdots
\end{equation} 

Rearranging this gives a particular functional form for $p_{\rm hard}(Q(l;s);l)$. This is to be substituted in the differential
equation (\ref{e:E2hardDEb}) with the 
change of variable $r = Q(l;s)$, which we do  using computer algebra.  Equating terms
at leading powers of $l$ gives the differential equation stated below (\ref{1.1g}) for $q_0$ at order $l$, and equation \eqref{e:SymmDE} at order
$l^{1/3}$. 

From the working of the proof of Proposition \ref{p:Corr2DE} we have
\begin{equation}\label{2.1q} 
 E_2^{\rm hard}(0;(0,4z^2);l) = \exp \Big ( - \int_{\tilde{t}}^\infty u_0(r) \, dr \Big ) \Big ( 1 - {1 \over l^{2/3}} \int_{\tilde{t}}^\infty u_1(r) \, dr + {\rm O}( l^{-1}) \Big ).
 \end{equation}
 The expansion (\ref{2.1p}) now follows from this, (\ref{2.1o}) and (\ref{e:E1hardDE}).

In relation to the boundary condition, we allow $l$ in both the above working, and that of the proof of  Proposition \ref{P1a}  to be continuous. The effect is to replace the discrete variable
$\tilde{t}$ by the continuous variable $t$. The working of the proof of Proposition \ref{P1a} then tells us
\begin{equation}\label{2.1r} 
 E_1^{\rm hard}(0;(0,4z^2);(l-1)/2)  
 = \det \Big ( \mathbb I - \big(  \mathbb{V}_{{t}, (0, \infty)}^{\rm soft} + (1/ l^{2/3}) \mathbb{M}_{t, (0, \infty)} \big) + \cdots \Big ),
 \end{equation}
 From this latter formula, it follows that for large $t$
 \begin{equation}\label{2.1s} 
 \log E_1^{\rm hard}(0;(0,4z^2);(l-1)/2) \sim - \int_0^\infty \Big ( V_{t}^{\rm soft}(x,x) +  (1/ l^{2/3}) M_t(x,x) \Big ) \, dx.
  \end{equation}
 Differentiating (\ref{2.1s}), simplifying using the explicit forms of the kernels (\ref{d:Vsoft}) and (\ref{d:Mtkernel}), then
 comparing to the logarithmic derivative of (\ref{2.1p}) with $\tilde{t}$ replaced by $t$ we obtain
 \begin{align}\label{3.36}
-q_0(r) -u_0(r) &\mathop{\sim}\limits_{r \to \infty} -\rmAi(r), \nonumber \\
-q_1(r) - u_1(r) &\mathop{\sim}\limits_{r \to \infty} \frac{1}{(2^{1/3}) 30} \Big( 14r \rmAi(r) +r^2 \rmAi'(r) \Big).
\end{align}
We already have the asymptotic behaviour for $u_0$ in (\ref{e:LeadDE}) and $u_1$ in \eqref{e:u1tinfty}, from which we can check that $u_0$ and $u_1$
fall off faster than the RHS's in \eqref{3.36}. This implies the boundary conditions as stated below (\ref{1.1g}) for $q_0$, and \eqref{e:w1BC} for $q_1$.
 \end{proof}
 
 \begin{remark}\label{R.E}
It is known that \cite{TW94a} that
\begin{align}
 \int_t^\infty u_0(r) \, dr = \int_t^\infty (r - t) q_0(r)^2 \, dr.
\end{align}
 Using this in (\ref{B.1}) and comparing with (\ref{2.2c}) shows 
\begin{align}
\int_t^\infty u_1(r) \, dr = {2^{2/3} \over 10} \Big ( - q_0(t)^2 + \Big ( \int_t^\infty  q_0(x)^2 \, dx \Big )^2 +  {t^2 \over 6} \int_t^\infty q_0(x)^2 \, dx \Big ).
\end{align}
  Hence $F_{1,1}^{\rm H}(t)$ can be expressed entirely in terms of $q_0(r)$ and $q_1(r)$.
  \end{remark}
  
As for $F_{2,1}^{\rm H}(t)$, the earlier work of \cite{BJ13} gives an expression for $F_{1,1}^{\rm H}(t)$ in terms
  of Painlev\'e transcendents simpler than (\ref{2.1pX}). This reads \cite[Th.~1.2]{BJ13}
 \begin{equation}\label{BJ.2}
 F_{1,1}^{\rm H}(t) = - {2^{2/3} \over 10} \bigg ( 2 {d^2 \over d t^2} + {t^2 \over 6} {d \over dt} \bigg )
 \exp \bigg ( - {1 \over 2} \int_t^\infty  \Big ( (t - r)  q_0(r)^2 + q_0(r) \Big ) \, dr \bigg ).
 \end{equation}
 A direct verification of the consistency of (\ref{BJ.2}) and (\ref{2.1pX}) can, upon use of Remark \ref{R.E}, be attempted
 along the same lines of that used in relation to verifying the identity (\ref{B.2}). However the implied identity for
 $q_1(r)$ has extra terms and a more complex structure relative to (\ref{B.2}), and this has not been carried through. We remark that graphical agreement
 is readily verified.

 Knowledge of (\ref{2.1p}) and the structural relation between (\ref{e:E1hardDE}) and (\ref{e:E4hardDE}) allows for the
 differential equation companion to Corollary \ref{C3.2} to be presented.
 
 \begin{cor}\label{C3.4}
Let $\{u_0, u_1, q_0, q_1 \}$ be as in Proposition  \ref{P3.3}.  For large $l,z$ we have
 \begin{multline} \label{e:Corr4DE}
F_{4,1}^{{\rm H} } (t)
= \frac{1}{2} \exp \left(- \frac{1}{2} \int_{t}^{\infty} u_0(r) dr \right) \Bigg[ \sinh \left( \frac{1}{2} \int_{t}^{\infty} q_0(r) dr \right) \int_t^{\infty} q_1 (r) dr\\
  - \cosh \left( \frac{1}{2} \int_{t}^{\infty} q_0(r) dr \right) \int_{t}^{\infty} u_1 (r) dr \Bigg].
\end{multline}
\end{cor} 

 \begin{remark}\label{R.E1}
 This quantity was not considered in \cite{BJ13}. There is no evidence of an analogue of the simplified formulas
(\ref{B.1})  and  (\ref{BJ.2}).
\end{remark}

 \subsection{Comparison with numerical calculations}
 
Here we perform similar numerical calculations to those in Section \ref{sec:NumsBeta2} above, calculating the corrections $F_{1,1}^{{\rm H}} $ and $F_{4,1}^{{\rm H}}$, from \eqref{2.1mX} and \eqref{2.1m} respectively, using both the Fredholm determinant expressions and the expressions in terms of solutions to differential equations. We compare them to the differences
\begin{align}
\label{d:Corr1ell}  \delta_1^{\rm H} (t) := l^{2/3} \Bigg( E_{1}^{\rmhard} \Bigg( 0; (0, Q(l; t)); \frac{l -1}{2} \Bigg) - E_{1}^{\rmsoft} \Big( 0; (t,\infty) \Big) \Bigg)
\end{align}
and
\begin{align}
\label{d:Corr4ell}   \delta_4^{\rm H} (t) :=  l^{2/3} \Bigg( \tilde{E}_{4}^{\rmhard} \Big( 0; (0, Q(l; t)); l+1 \Big) - \tilde{E}_{4}^{\rmsoft} \Big( 0; (t,\infty) \Big) \Bigg)
\end{align}
for $l=20$.

For the Fredholm determinants expressions we again use the toolbox of \cite{Bornemann2010}. For the differential equation solutions $q_0$ and $q_1$ we obtain a sequence of Taylor series solutions; $500$ series of degree $6$ for $q_0$ and $1,000$ series of degree $8$ for $q_1$. Lastly, we compute a sequence of $5,400$ series of degree $6$ for $p_{\rmhard} (r, 20)$. With these sequences, and the corresponding sequences of DE solutions from Section \ref{sec:NumsBeta2} we obtain the graphs in the left panels of Figure \ref{f:Corrapproxell20_1} for $F_{1,1}^{{\rm H}} $ and of Figure \ref{f:Corrapproxell20_4} for $F_{4,1}^{{\rm H}}$. In the right panels of these figures we present plots of
\begin{align}
\label{d:delta1} E_{1}^{\rmhard} \Bigg( 0; (0, Q(l; t)); \frac{l -1}{2} \Bigg) - E_{1}^{\rmsoft} \Big( 0; (t,\infty) \Big) - \frac{1}{l^{2/3}} F_{1,1}^{\rm H}(t)
\end{align}
and
\begin{align}
\label{d:delta4} \tilde{E}_{4}^{\rmhard} \Big( 0; (0, Q(l; t)); l+1 \Big) - \tilde{E}_{4}^{\rmsoft} \Big( 0; (t,\infty) \Big) - \frac{1}{l^{2/3}} F_{4,1}^{\rm H}(t).
\end{align}
as approximations to the higher order corrections in \eqref{2.1mX} and \eqref{2.1m} respectively.

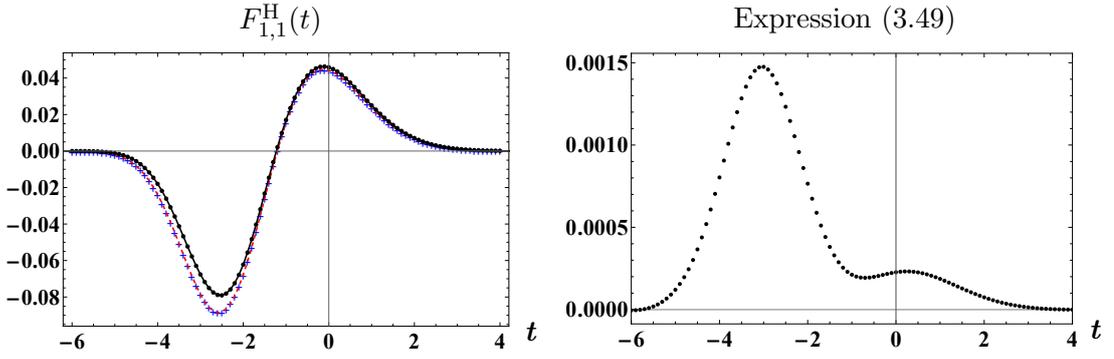
\begin{figure}
\centering
     \begin{minipage}{0.45\textwidth}
	\begin{tikzpicture}[scale=1, every node/.style={transform shape}]
         \node[inner sep=0pt] at (0,0) {\includegraphics[height=4cm, align=t]{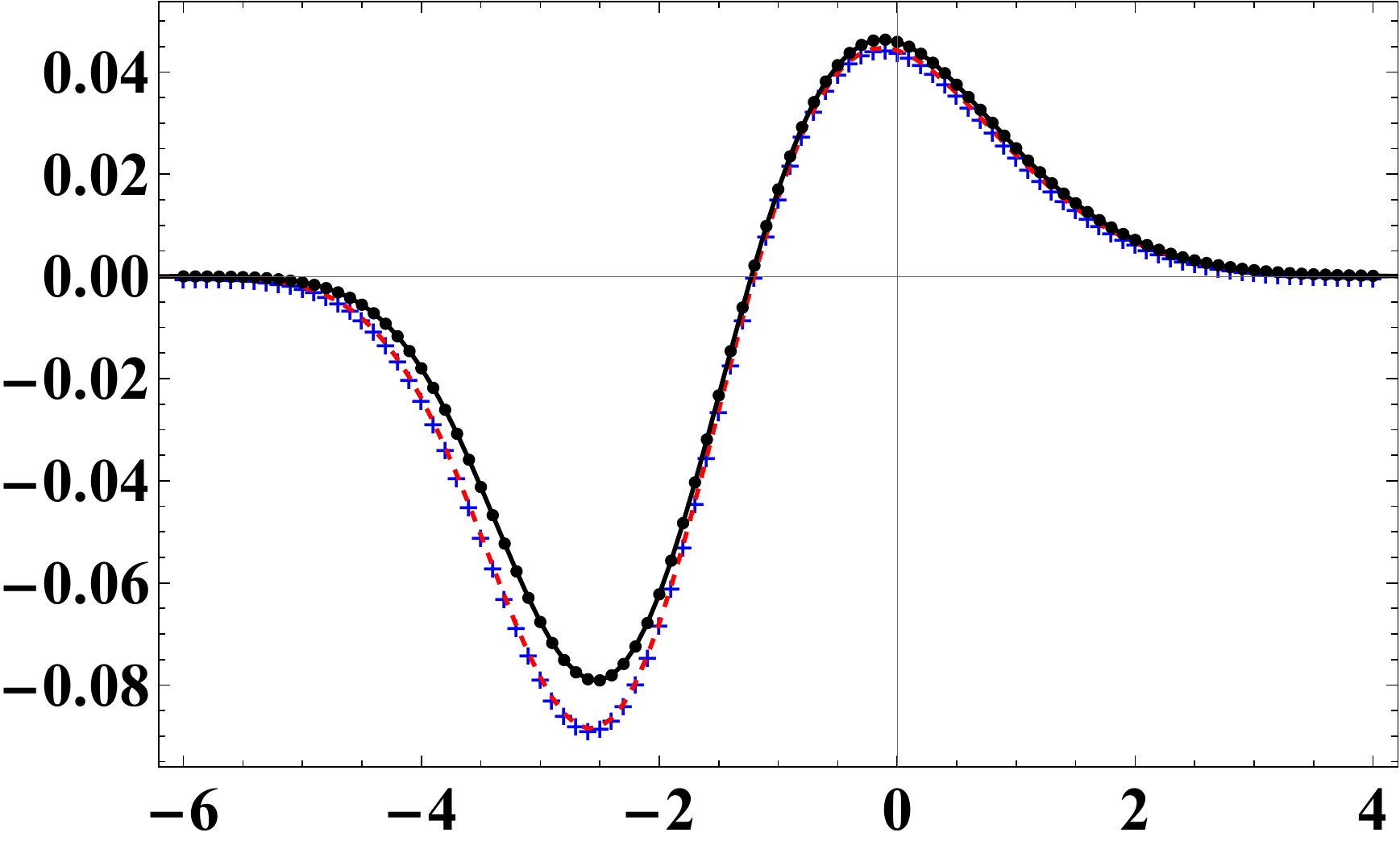}};
         \node at (3.6,-1.7) {\textbf{\textit{t}}};
         \node at (0.3,2.4) {$F_{1,1}^{\rm H}(t)$};
         \end{tikzpicture}
	\end{minipage} \quad
	\begin{minipage}{0.45\textwidth}
	\begin{tikzpicture}[scale=1, every node/.style={transform shape}]
         \node[inner sep=0pt] at (0,0) {\includegraphics[height=4cm, align=t]{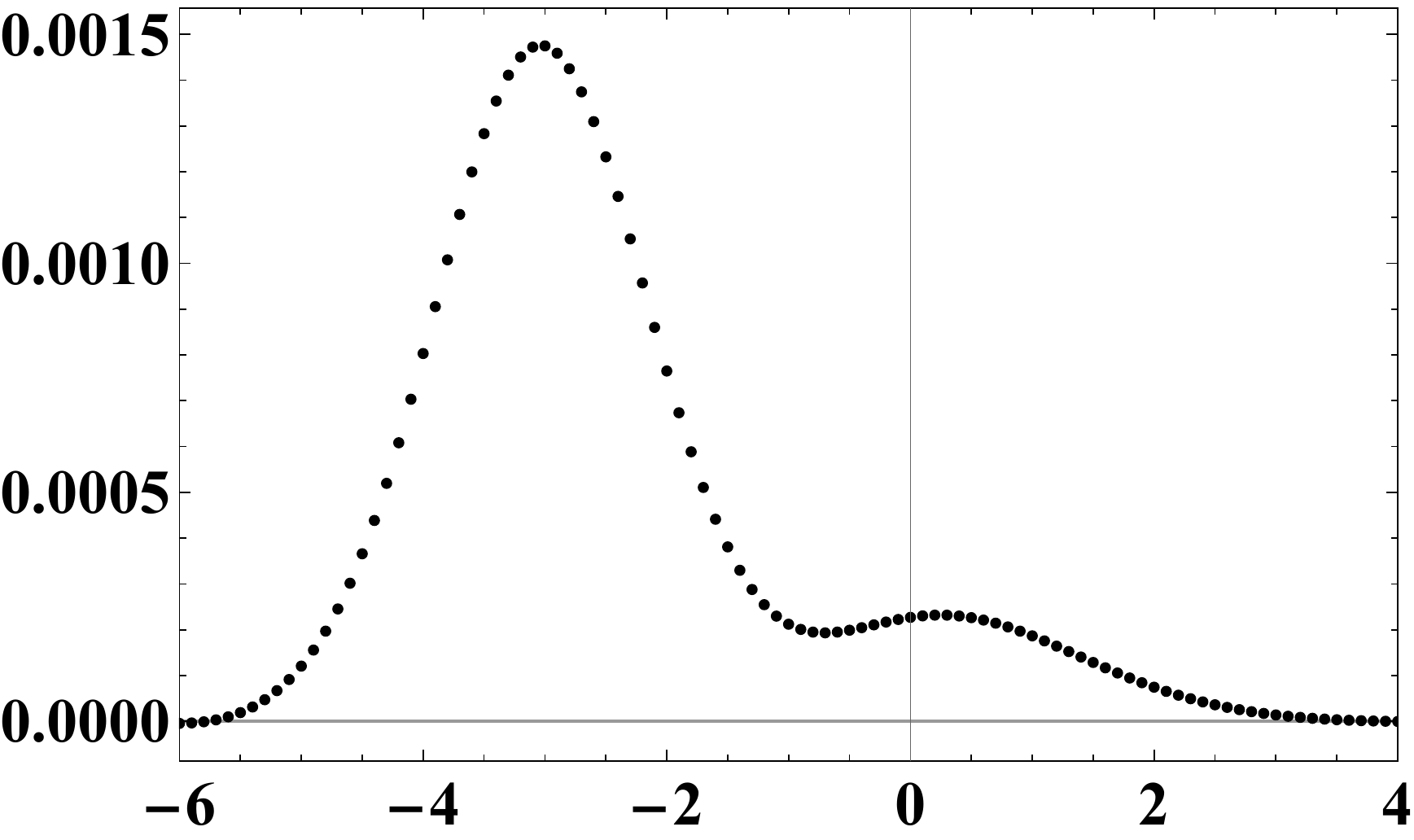}};
         \node at (3.6,-1.7) {\textbf{\textit{t}}};
         \node at (0.3,2.4) {Expression \eqref{d:delta1}};
         \end{tikzpicture}
	\end{minipage}
\caption{In the left panel we have the correction term $F_{1,1}^{\rm H}(t)$ in \eqref{2.1mX} calculated using \eqref{2.1pX} [blue crosses], and using \eqref{e:Corr1t1} [red dashed line]. On the same axes we also plot $\delta_1^{\rm H} (t)$ from \eqref{d:Corr1ell} for $l=20$ using the Fredholm determinant expressions \eqref{e:FredDetHard1} and \eqref{e:E1softFD} [black line] and also using the expressions in terms of solutions to differential equations \eqref{e:E1softDE} and \eqref{e:E1hardDE} [black dots]. In the right panel we plot the difference \eqref{d:delta1}, which is a numerical approximation to the remaining terms in \eqref{2.1mX}.}
\label{f:Corrapproxell20_1}
\end{figure}

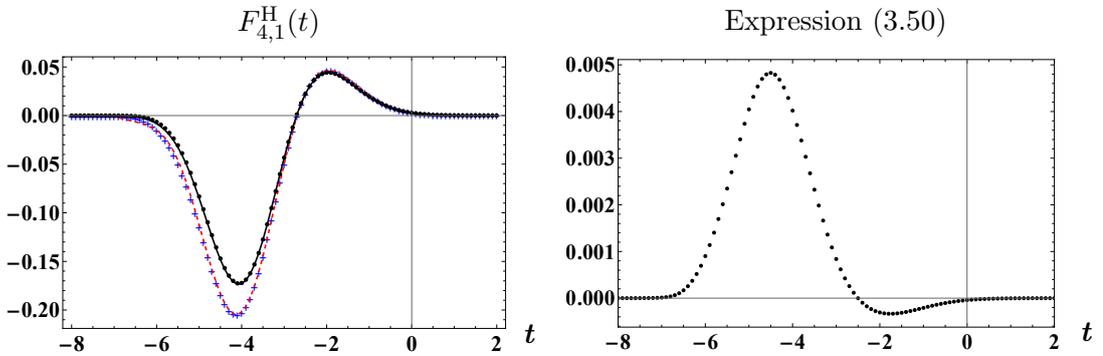
\begin{figure}
\centering
     \begin{minipage}{0.45\textwidth}
	\begin{tikzpicture}[scale=1, every node/.style={transform shape}]
         \node[inner sep=0pt] at (0,0) {\includegraphics[height=4cm, align=t]{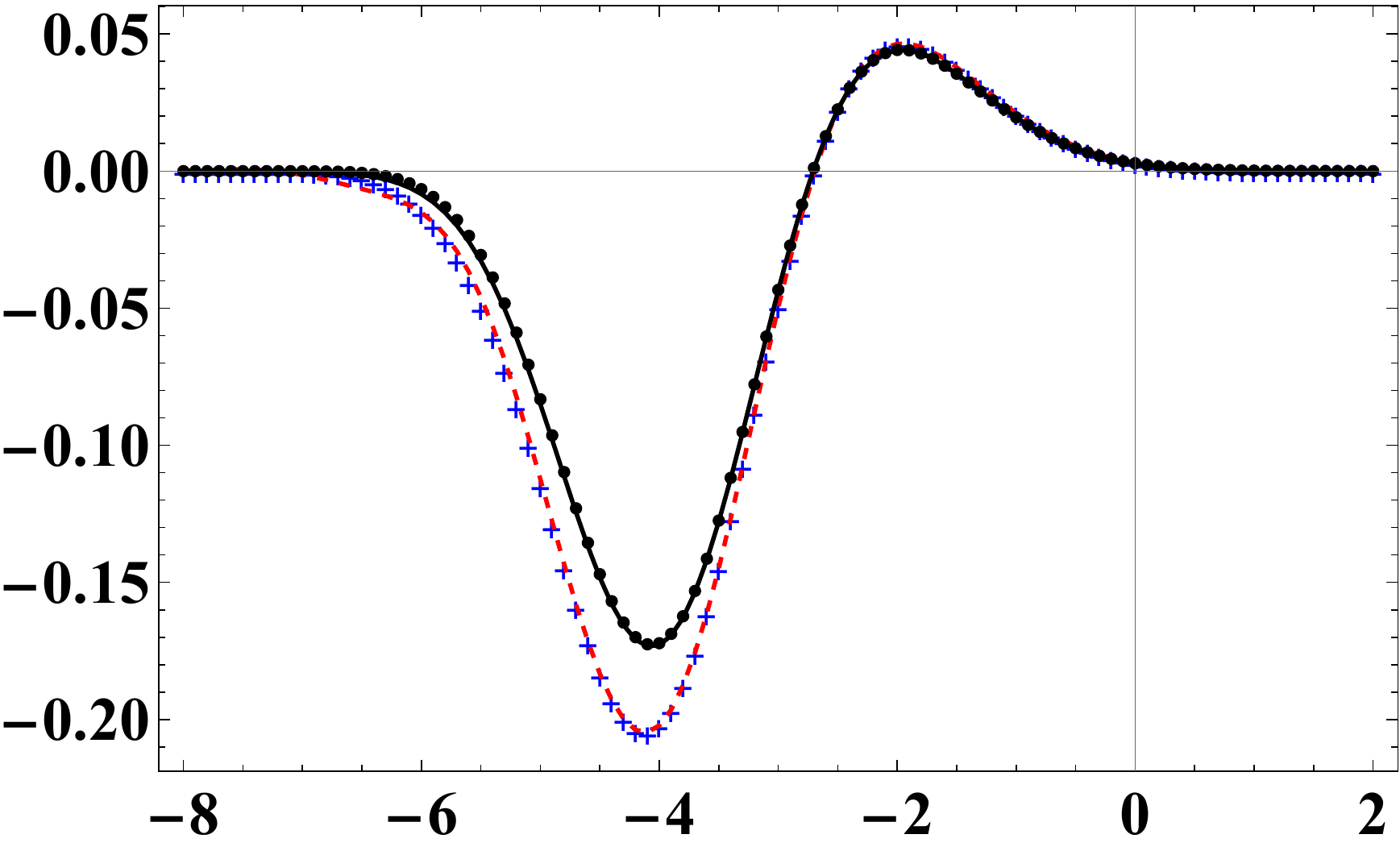}};
         \node at (3.6,-1.7) {\textbf{\textit{t}}};
         \node at (0.3,2.4) {$F_{4,1}^{\rm H}(t)$};
         \end{tikzpicture}
	\end{minipage} \quad
	\begin{minipage}{0.45\textwidth}
	\begin{tikzpicture}[scale=1, every node/.style={transform shape}]
         \node[inner sep=0pt] at (0,0) {\includegraphics[height=4cm, align=t]{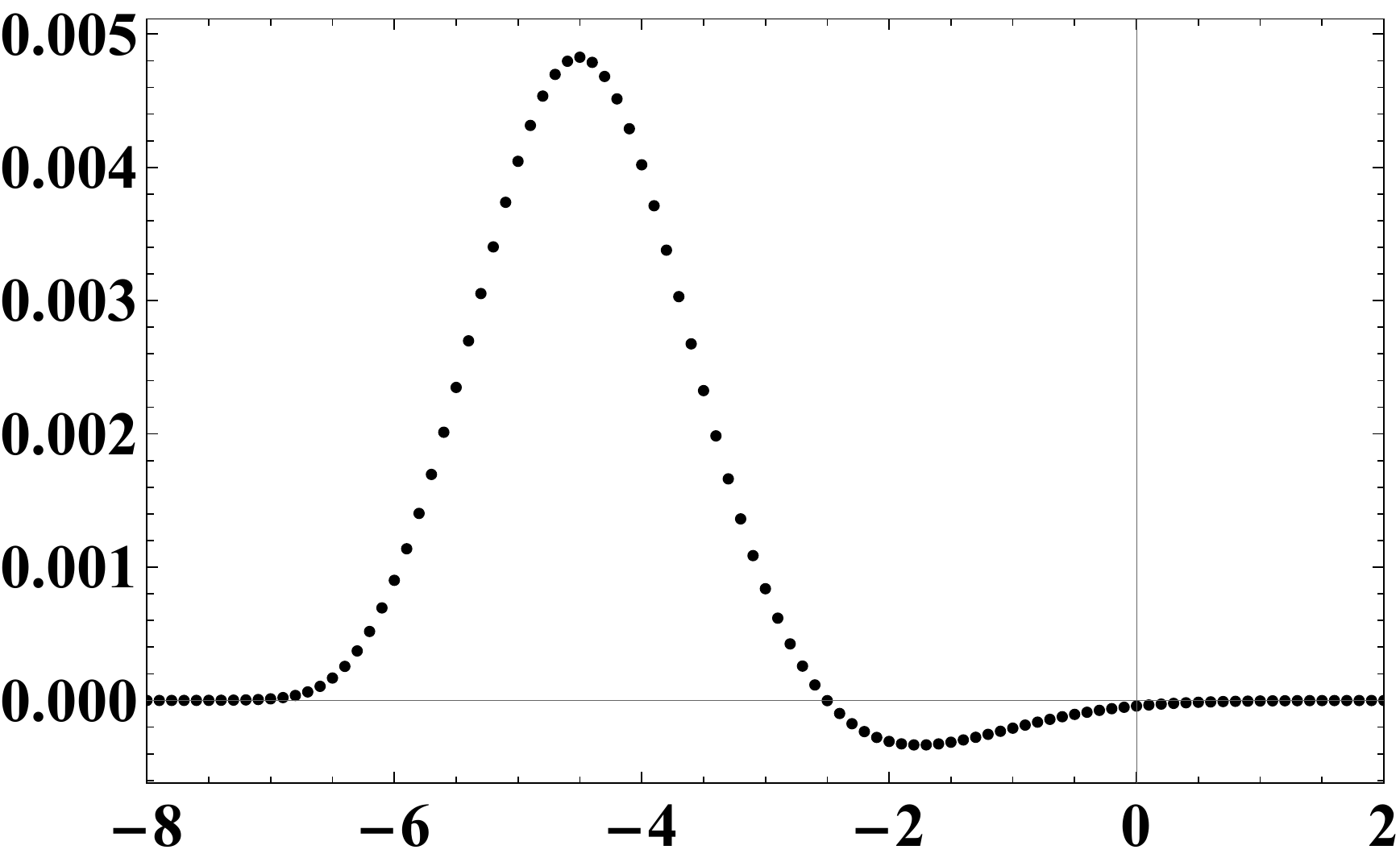}};
         \node at (3.6,-1.7) {\textbf{\textit{t}}};
         \node at (0.3,2.4) {Expression \eqref{d:delta4}};
         \end{tikzpicture}
	\end{minipage}
\caption{These are the plots for $F_{4,1}^{\rm H}(t)$ from \eqref{e:Corr4t} and \eqref{e:Corr4DE} analogous to those in Figure \ref{f:Corrapproxell20_1}, using the same Fredholm determinant calculations and differential equation solutions. We also have $\delta_4^{\rm H} (t)$ from \eqref{d:Corr4ell} for $l=20$ again using the Fredholm determinant expressions and the expressions in terms of solutions to differential equations. Similarly, in the right panel is the difference \eqref{d:delta4}.}
\label{f:Corrapproxell20_4}
\end{figure}

\section{Large $N$ expansion of  ${\rm Pr} \Big ( {l_N^\square - 2 \sqrt{N} \over N^{1/6}}  \le t \Big )$ and symmetrised analogues}
\subsection{Relationship to large $z$ form of $   \Pr \left(\frac{l^{\Box} -2z}{z^{1/3}} \leq t \right)$ and conjecture}

The longest increasing subsequence problem has been described in the paragraph including (\ref{1.1e}).
Equating the latter with (\ref{1.1k}) shows the coincidence of limit laws with the maximal up/right path length
in the Hammersley process,
\begin{equation}\label{4.1a}
\lim_{N \to \infty}  {\rm Pr} \Big ( {l_N^\square - 2 \sqrt{N} \over N^{1/6}}  \le t \Big ) = \lim_{z \to \infty}  \Pr \left(\frac{l^{\Box} -2z}{z^{1/3}} \leq t \right) = E_2^{\rm soft} \Big( 0;(t,\infty) \Big).
\end{equation}
To understand why these two limits coincide, first recall from (\ref{1.1j}) that $ {\rm Pr} ( l^\square  \le l) $ is an exponential generating function for
${\rm Pr} ( l_N^\square  \le l)$. Furthermore the latter is a decreasing function of $N$ that takes values between $0$ and $1$. In this general setting Johansson \cite{Jo98} proved what has been referred to as a de-Poissonisation lemma.

\begin{prop}\label{P4.1}
Let the sequence $\{q_n\}_{n=0,1,\dots}$ satisfy the bounds
$0 \le q_n \le 1$ and be monotonically decreasing so that
$q_n \ge q_{n+1}$. Let
\begin{align}
\phi(\xi) := e^{-\xi} \sum_{n=0}^\infty q_n {\xi^n \over n!}
\end{align}
and for given $d>0$ write
\begin{align}
\mu_n^{(d)} = n + (2\sqrt{d+1}+1)\sqrt{n\log n}, \qquad \nu_n^{(d)} = n - (2\sqrt{d+1}+1)\sqrt{n\log n}.
\end{align}
One has
\begin{equation}\label{bp.johk}
\phi(\mu_n^{(d)}) - C n^{-d} \le q_n \le
\phi(\nu_n^{(d)}) + C n^{-d}
\end{equation}
for all $n \ge n_0$, where $C$ is some positive constant.
\end{prop}

Rewriting (\ref{1.1j}) so that it reads
\begin{equation}\label{4.1b}
  \Pr \left(\frac{l^{\Box} -2z}{z^{1/3}} \leq t \right) = e^{-z^2} \sum_{N=0}^\infty  {z^{2N} \over N!}   {\rm Pr} \left ( {l_N^\square - 2 z \over z^{1/3}}  \leq t \right) ,
\end{equation}
then applying Proposition \ref{P4.1} establishes the first equality in (\ref{4.1a}).
From Proposition \ref{P1} we know details of further terms in the large $z$ expansion of ${\rm Pr} \Big ( {l^\square - 2 z \over z^{1/3}} \leq t \Big ) $, with the leading correction to the limit formula (\ref{1.1k}) being ${\rm O}(1/z^{2/3})$. However, this knowledge used in Proposition \ref{P4.1} does not give
information on details of further terms in the large $N$ expansion of ${\rm Pr} \Big ( {l_N^\square - 2 \sqrt{N} \over N^{1/6}}  \le t \Big )$. In fact we don't know of
any analytic approach to this question. Nonetheless there are numerical methods that allow for data to be obtained leading to a conjecture.

\begin{conj}\label{C1}
Set $ F_{2,0}(t)  = E_2^{\rm soft}(0;(t,\infty))$. For some $ F_{2,1}(t)$ we  have
\begin{equation}\label{1.1iXa}
 {\rm Pr} \left ( {l_N^\square - 2 \sqrt{N} \over N^{1/6}}  \le t \right ) =  F_{2,0}(t^*) +  {1\over N^{1/3}}   F_{2,1}(t) + \cdots,
  \end{equation}
  where $t^*$ is defined in \eqref{1.1i}.
\end{conj}

\begin{remark}
With $\sqrt{N}$ identified as $z$ the expansion (\ref{1.1iXa}) is consistent with (\ref{1.1m}).
\end{remark}

\subsection{Data from the Painlev\'e characterisation}\label{S4.2}

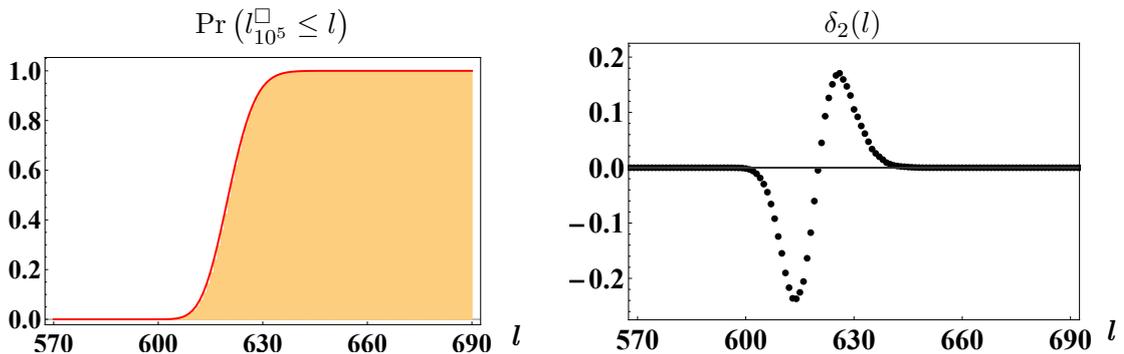
\begin{figure}[b]
\centering
     \begin{minipage}{0.45\textwidth}
	\begin{tikzpicture}[scale=1, every node/.style={transform shape}]
         \node[inner sep=0pt] at (0,0) {\includegraphics[height=4cm, align=t]{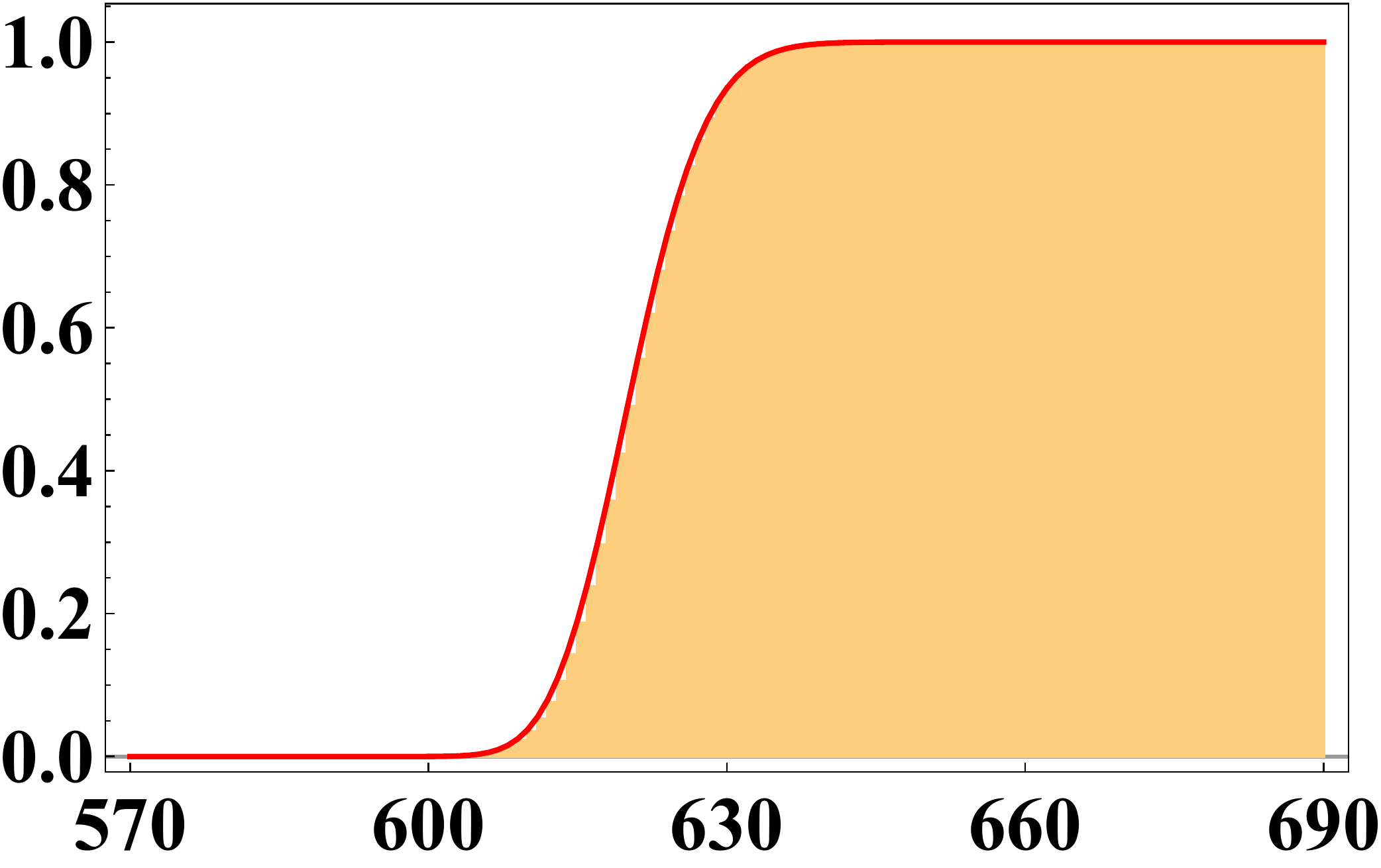}};
         \node at (3.5,-1.7) {\textbf{\textit{l}}};
         \node at (0.3,2.4) {$\Pr\big( l^{\Box}_{10^5} \leq l\big)$};
         \end{tikzpicture}
	\end{minipage} \quad
	\begin{minipage}{0.45\textwidth}
	\begin{tikzpicture}[scale=1, every node/.style={transform shape}]
         \node[inner sep=0pt] at (0,0) {\includegraphics[height=4.17cm, align=t]{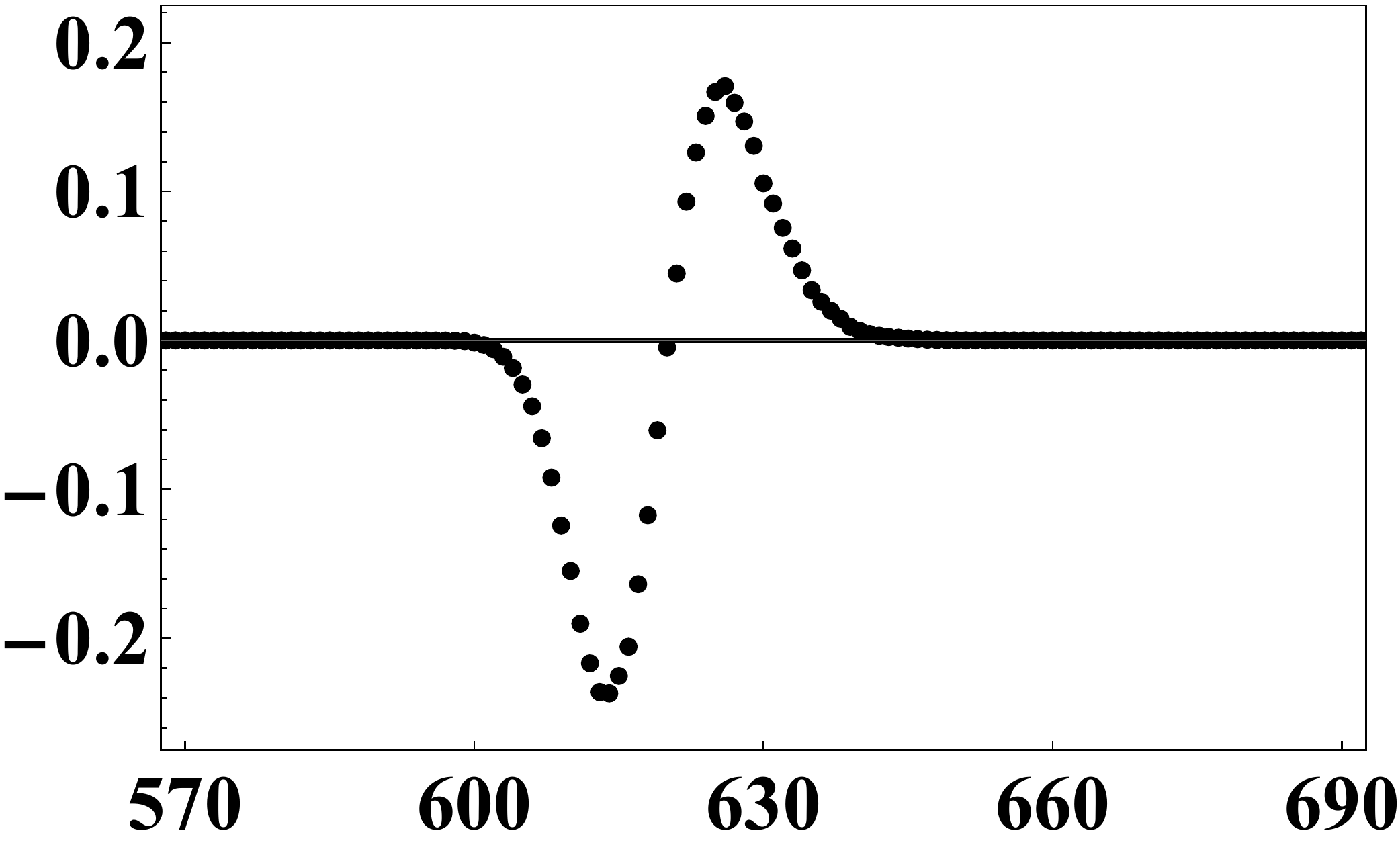}};
         \node at (3.7,-1.75) {\textbf{\textit{l}}};
         \node at (0.3,2.3) {$\delta_2 (l)$};
         \end{tikzpicture}
	\end{minipage}
        \caption{Here we have the analogous plots as those in Figure \ref{f:AsymmLISn700}, only now with $5\times 10^6$ samples of random permutations with length $N= 10^5$, again with the limiting CDF given by the second term in (\ref{1.1h}) [red curve]. On the right is plotted an approximation to the quantity $\delta_2 (l)$ from \eqref{SC}, where we estimate $\Pr \left( l^{\Box}_{N} \leq l \right)$ by using the empirical CDF from the left panel.}
        \label{f:AsymmLISnLarge}
\end{figure}

Use of (\ref{e:E2hardDEa}) in (\ref{e:lBoxE2Hard}) and recalling (\ref{1.1j}) tells us that
\begin{align}\label{4.5}
\sum_{N=0}^{\infty} \frac{z^{2N} }{N!} \Pr \left( l_N^{\Box} \leq l \right) = e^{z^2} \exp \left( \int_0^{4z^2} \frac{v (r; l)}{r} dr \right)=: G^{\Box} (z; l),
\end{align}
where $v(r;l)$ is the solution of the particular $\sigma$-PIII$'$ equation specified by (\ref{e:vDE}).
This shows that if we expand $G^{\Box}(z; l)$ in powers of $z^2$
\begin{align}
\label{e:Gtlexp} G^{\Box}(z;l)= \sum_{N=0}^{\infty} z^{2N} c_N^{\Box} (l), \qquad c_N^{\Box} (l) N! := \Pr \left( l_N^{\Box} \leq l \right),
\end{align}
then we have a practical method to compute $\{  \Pr \left( l_N^{\Box} \leq l \right) \}$.
Thus our approach is to use the characterisation (\ref{e:vDE}) to carry out the series expansion
\begin{align}
v (r; l) = r^{l+1} \sum_{k=0}^M a_k^{\Box}(l)  r^k,
\end{align}
up to some cutoff $M$. We were able to carry out the computation for $M=700$, allowing for the computation
of the CDF for all $l_N^{\Box} $ up to $N = 700$. 
The data for this quantity can be stored as exact integers, by multiplying each of the probabilities by $N!$;
see \cite[Table 2--4]{OR00} for some examples, with the largest value of $N$ there being $N = 60$.
Recall Figure \ref{f:AsymmLISn700}, where on the left we displayed the data for the case $N = 700$ in a graphical form --- the histogram is the empirical CDF while the black dots are calculated using the $c_N^{\Box} (l)$ from \eqref{e:Gtlexp}. On the right of the figure we plotted the difference \eqref{1.1h}. Multiplying this difference by the conjectured order of the correction term in \eqref{1.1iXa} we obtain the scaled difference
\begin{equation}\label{SC}
 \delta_2 (l) := N^{1/3} \left  [ \Pr \left( l^{\Box}_{N} \leq l \right) - E_2^{\rmsoft} \left(0; \left(\frac{l -2 \sqrt{N}}{N^{1/6}}, \infty \right) \right) \right  ],
\end{equation}
which will allow us to compare the data to other values of $N$.

\subsection{Data from simulations}\label{S4.3}
For values of $N$ beyond $N = 700$ data can be generated by Monte Carlo simulations. The C code used to generate samples of $l_N^{\Box}$ was given to the authors by Eric Rains, based on the code used for the simulations in \cite{OR00}, which uses the algorithm of \cite{BB68}. The most expensive part of the code is the pseudo random number generation, for which Rains' code uses a Marsaglia-style multiply-with-carry bit-shifting algorithm. To generate the value of $l_N^{\Box}$ from $5 \times 10^6$ trials with $N = 10^5$ took approximately $51,000$ seconds. Without this code, an alternative method to generate the simulation data is to simply use the inbuilt Mathematica command {\tt LongestOrderedSequence}, which has comparable runtime for this value of $N$.

In Figure \ref{f:AsymmLISnLarge} we display the data in a graphical form, along with the estimate of the scaled difference $\delta_2 (l)$ from (\ref{SC}), where $\Pr \left( l^{\Box}_{N} \leq l \right)$ is taken to be the empirical CDF. In Figure \ref{f:CompDelta2} we compare $\delta_2 (t)$ with $N=700, \, 20000$ and $10^5$, where we have rescaled $t= (l- 2\sqrt{N})/N^{1/6}$ --- the agreement in the plots suggests that $N^{-1/3}$ is indeed the correct order of the next-to-leading term in \eqref{1.1iXa}.

\begin{figure}
\centering

\begin{minipage}{0.7\textwidth}
	\begin{tikzpicture}[scale=0.8, every node/.style={transform shape}]
         \node[inner sep=0pt] at (0,0) {\includegraphics[width=\textwidth, align=t]{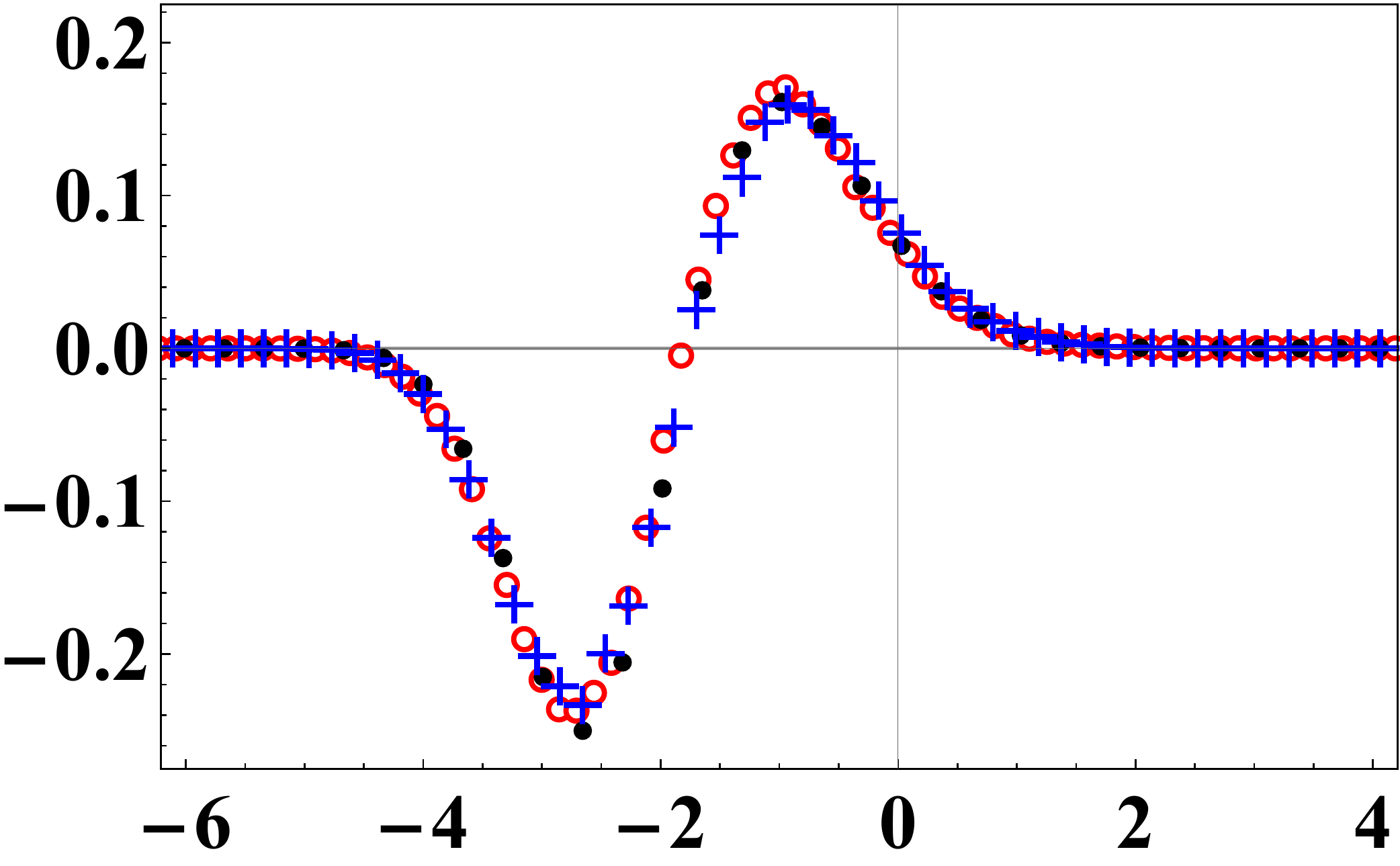}};
         \node at (5.8,-2.8) {\Large \textbf{\textit{t}}};
         \node at (0.6,3.8) {\Large $\delta_2 (t)$};
         \end{tikzpicture}
	\end{minipage}
	\caption{Comparison of $\delta_2 (t)$ from \eqref{SC} with $N=700$ [black dots], $20000$ [blue crosses] and $10^5$ [red circles], where the horizontal axis has been rescaled by $t= (l- 2\sqrt{N})/N^{1/6}$.}
\label{f:CompDelta2}
\end{figure}

\subsection{Large $N$ expansion of the mean and variance of $l^\Box_N$}

We turn our attention now to the large $N$ form of the mean and variance.
From the leading term in (\ref{1.1iXa}) it follows that for large $N$ \cite{BDJ98}
\begin{align}
\label{e:limExpVal2} \bbE [l^{\Box}_N]  \mathop{\sim}\limits_{N\to \infty} 2 \sqrt{N} +m_2^{(1)} N^{1/6} + \cdots , \quad  m_{2}^{(1)} \approx -1.771086807,
\end{align}
where, with ${d F_2(r) \over dr} = {d \over dr} E_2^{\rm soft}(0;(r,\infty))$,
\begin{align}
m_{2}^{(k)}:= \int_{-\infty}^{\infty} r^k d F_{2} (r),
\end{align}
and the numerical value follows from a computation based on (\ref{1.1f}) in \cite{TW94a}. On the other hand the correction term in (\ref{1.1iXa}) does not immediately reveal information about higher order terms in (\ref{e:limExpVal2}),
the reason being that $l^{\Box}_N$ is a discrete quantity, while the right hand side of \eqref{1.1iXa} corresponds to rescaling and smoothing of the discrete distribution. This is similarly true of the variance, for which the 
limit theorem (\ref{1.1e}) gives that
\begin{align}
 {\rm Var} [ l^{\Box}_N]    \mathop{\sim}\limits_{N\to \infty} \Big( m_2^{(2)} - (m_2^{(1)})^2 \Big) N^{1/3}, \qquad m_2^{(2)} - (m_2^{(1)})^2 \approx 0.81319,
 \end{align}
with the nature of higher order terms not immediately determined by the  correction term in (\ref{1.1iXa}). Our data can be used to investigate the corrections to $\bbE [l^{\Box}_N]$ and ${\rm Var} [ l^{\Box}_N]$ at a numerical level.

For this purpose, we note from elementary probability theory that
\begin{align}
\label{e:ExpValcnell} \bbE [l^{\Box}_N] = \sum_{k=0}^{N-1} (k+1) \Big ( {\rm Pr} (l^{\Box}_N \leq k + 1) - {\rm Pr} (l^{\Box}_N \leq k)  \Big ) = 
\sum_{k=0}^N \Big( 1- \Pr(l^{\Box}_N \leq k)\Big)
\end{align}
and
\begin{align}
 {\rm Var} [ l^{\Box}_N]  & =  \sum_{k=0}^{N-1} (k+1)^2 \Big ( {\rm Pr} (l^{\Box}_N \leq k + 1) - {\rm Pr} (l^{\Box}_N \leq k ) \Big ) 
- \Big (  \bbE [l^{\Box}_N] \Big )^2  \nonumber \\ & = 
 \label{Var1} 1+ \sum_{k=1}^N  (2k + 1) \Big( 1- \Pr(l^{\Box}_N \leq k)\Big) -  \Big (\bbE [l^{\Box}_N] \Big )^2.
\end{align}
From the theory of Section \ref{S4.2}, for $N$ up to $700$ we have exact knowledge of $\{ {\rm Pr}(l^{\Box}_N \le k) \}$.
Already a consequence of these distributions being discrete shows itself.
Thus if we approximate the CDF $\Pr (l_N^{\Box} \leq l)$ by the limiting expression $E_2^{\rmsoft} \Big( 0; \big( (l-2\sqrt{N})/N^{1/6}, \infty \big) \Big)$, and substitute this into \eqref{e:ExpValcnell} and \eqref{Var1} for the expected mean and variance, which we denote $\bbE_{\infty} [l_N^{\Box}]$ and $\Var_{\infty} [l_N^{\Box}]$ respectively, then we obtain the numerical estimates
\begin{align}
\label{e:Einfty} \bbE_{\infty} [l^{\Box}_N]  - \Big (  2 \sqrt{N} +m_2^{(1)} N^{1/6}  \Big ) &\mathop{\to}\limits_{N \to \infty} \frac{1}{2}\\
\label{e:Vinfty} \Var_{\infty} [l^{\Box}_N]  -\Big( m_2^{(2)} - ( m_2^{(1)})^2 \Big) N^{1/3} &\mathop{\to}\limits_{N \to \infty} \frac{1}{12}.
\end{align}
We recognise the values $1/2$ and $1/12$  as the mean and variance of the continuous uniform distribution on $[0,1]$.

Tabulating the quantities
\begin{align}\label{4.15}
\hat{\mu}_2(N)  :=  \bbE [l^{\Box}_N]  - \Big (  2 \sqrt{N} +m_2^{(1)} N^{1/6}  \Big ), \quad
 \hat{\sigma}^2_2(N)  := {\rm Var} [ l^{\Box}_N] - \Big( m_2^{(2)} - (m_2^{(1)})^2 \Big) N^{1/3},
\end{align}
leads us to believe that for large $N$ both these quantities are of order unity. Making an ansatz
$\hat{\mu}_2 (N) = c + d N^{-\alpha}$ and choosing between $\alpha = 1/6$ or $1/3$ as suggested by their
appearance already in this problem, we found that the choice $\alpha = 1/3$ gives the better fit.
Notice that the latter exponent is precisely the one appearing in Conjecture \ref{C1} for the CDF.
Performing a least squares analysis from our tabulation with $N$ from 10 up to 700 then gives
\begin{align}
\hat{\mu}_2(N)  \approx 0.5065 + {0.222 \over N^{1/3}}, \qquad
\label{c:musigma} \hat{\sigma}^2_2(N)  \approx -1.206 + {0.545 \over N^{1/3}}.
\end{align}
In keeping with \eqref{e:Einfty}, we expect the value $0.5065$ in relation to $\hat{\mu}_2(N)$ is exactly $1/2$. Note that the value $-1.206$, being distinct from $1/12$ in \eqref{e:Vinfty}, can be understood as being due to the square of the mean occurring in \eqref{Var1}. This gives a mechanism for the coupling of terms decaying in $N$ in the expansion of the mean, with terms that increase, which are not taken into consideration in deriving \eqref{e:Vinfty}.

In Figures \ref{f:muhat} and \ref{f:sigmahat} we plot $\hat{\mu}_2 (N)$ and $\hat{\sigma}_2^2(N)$ respectively, along with the conjectures in \eqref{c:musigma}. In the right panel of each we also plot the differences
\begin{align}
\hat{\mu}_2(N)- \left( 0.5065 + {0.222 \over N^{1/3}} \right), \qquad
\label{e:sigmahatdiff} \hat{\sigma}^2_2 (N)- \left( -1.206 + {0.545 \over N^{1/3}} \right).
\end{align}

\begin{figure}
\centering

	\begin{minipage}{0.45\textwidth}
	\begin{tikzpicture}[scale=1, every node/.style={transform shape}]
         \node[inner sep=0pt] at (0,0) {\includegraphics[height=4.05cm, align=t]{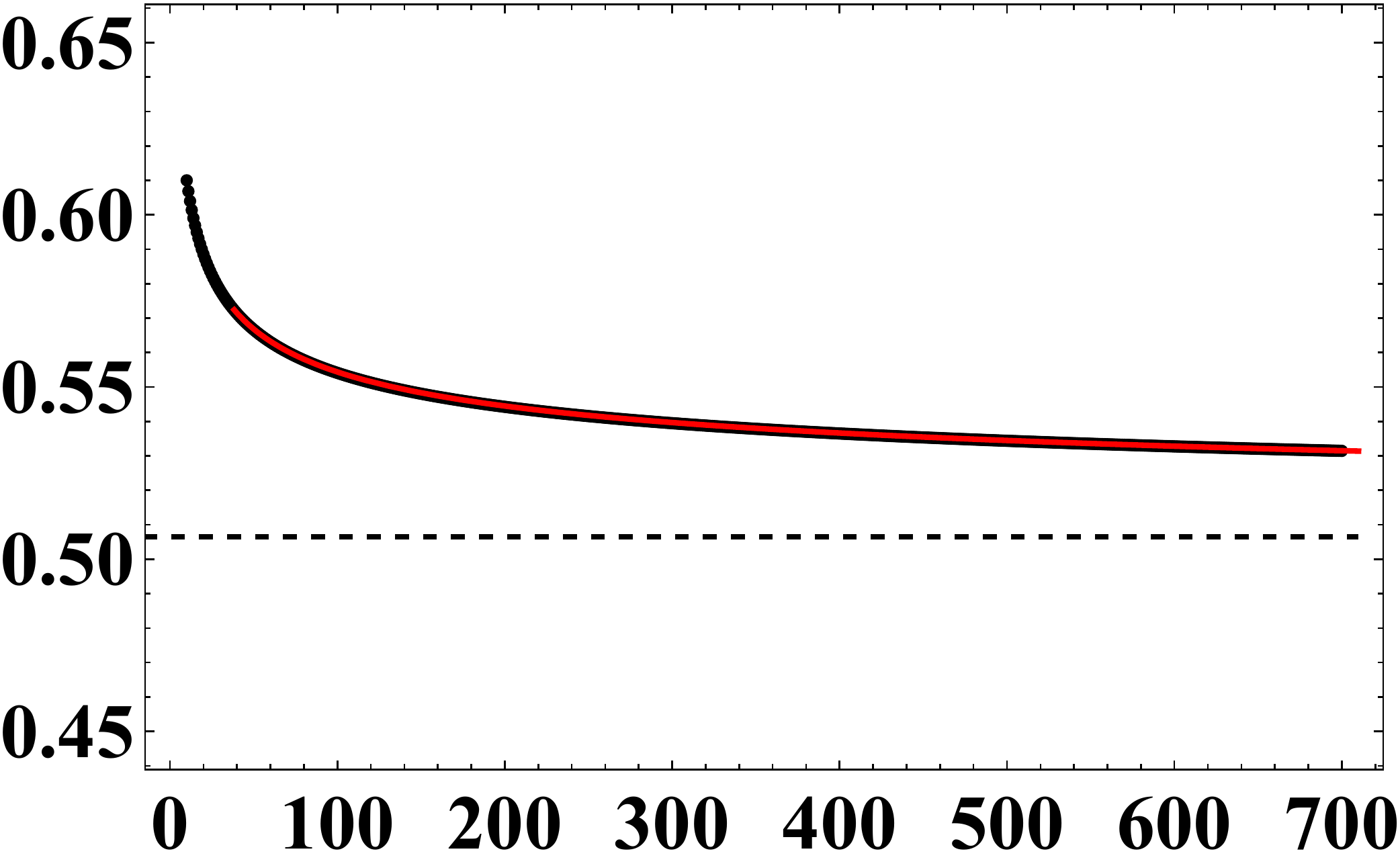}};
         \node at (3.6,-1.7) {\textbf{\textit{N}}};
         \node at (0.3,2.4) {$\hat{\mu}_2 (N)$};
         \end{tikzpicture}
	\end{minipage} \quad
	\begin{minipage}{0.45\textwidth}
	\begin{tikzpicture}[scale=1, every node/.style={transform shape}]
         \node[inner sep=0pt] at (0,0) {\includegraphics[height=4cm, align=t]{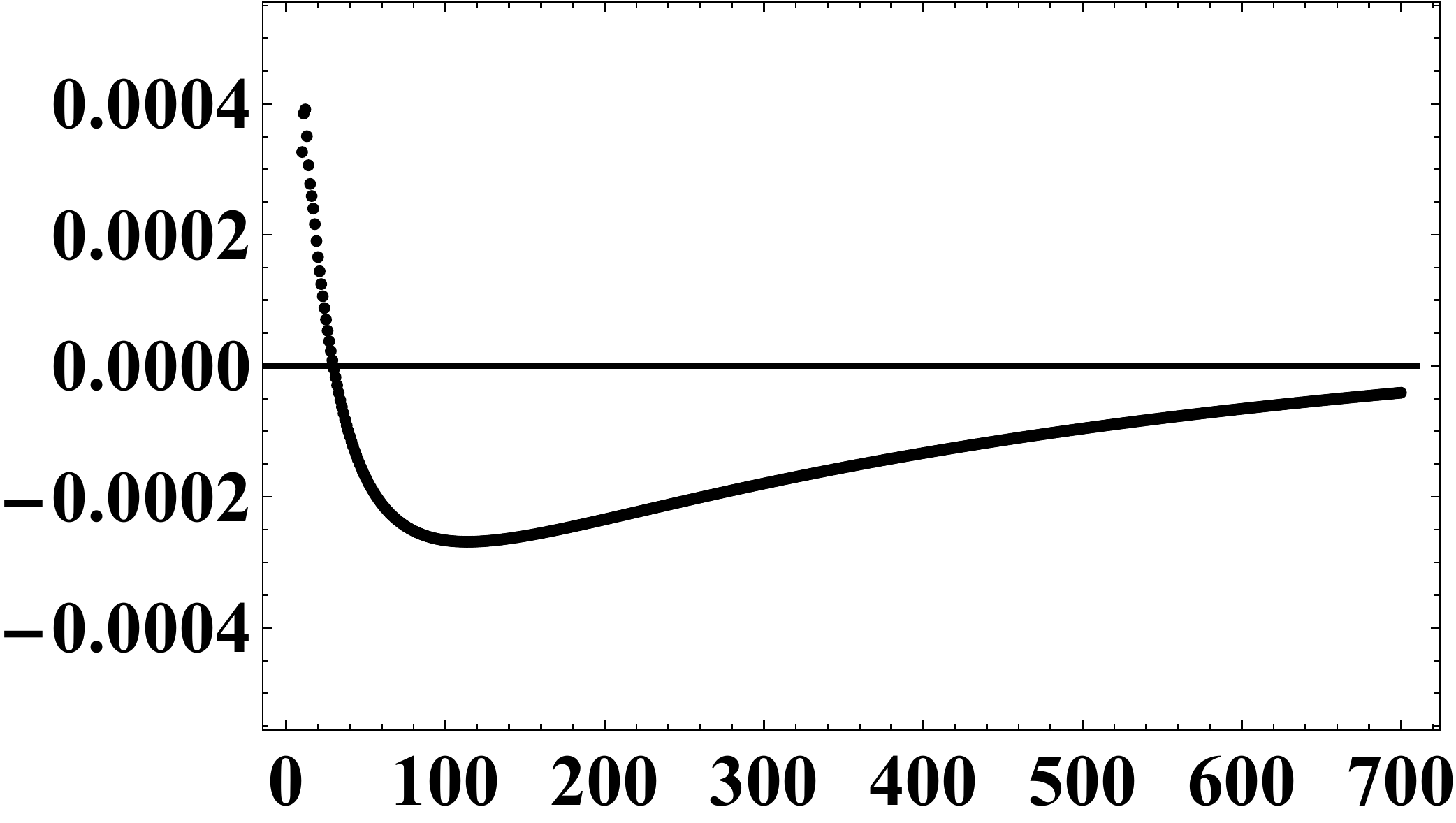}};
         \node at (3.85,-1.75) {\textbf{\textit{N}}};
         \node at (0.3,2.3) {Difference \eqref{e:sigmahatdiff}};
         \end{tikzpicture}
	\end{minipage}
	\caption{On the left we have the exact values of $\hat{\mu}_2 (N)$ from \eqref{4.15} for $N=10, 11, \dots, 700$ calculated using \eqref{e:ExpValcnell} [black dots] compared to the conjectured form in \eqref{c:musigma} [red line]. The difference between these, the quantity \eqref{e:sigmahatdiff}, is plotted on the right.}
	\label{f:muhat}
\end{figure}
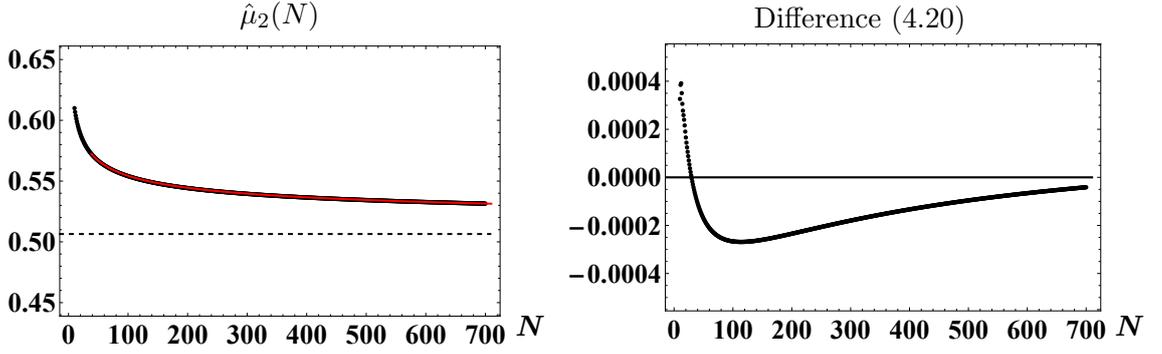

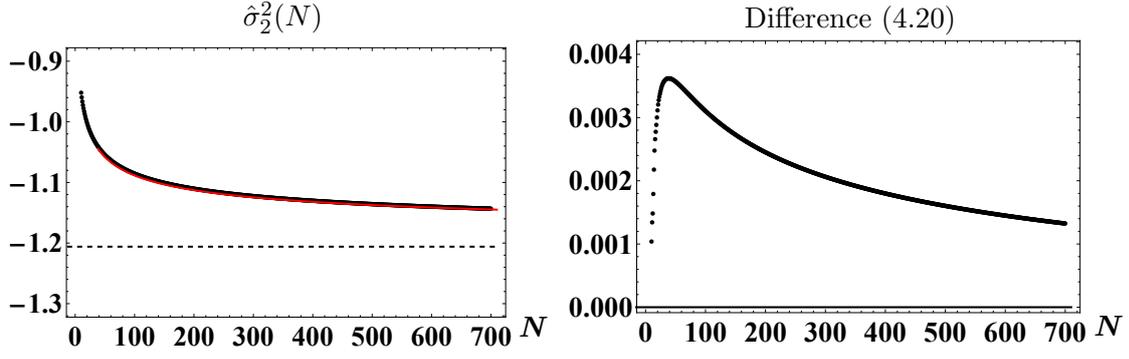
\begin{figure}
\centering

	\begin{minipage}{0.45\textwidth}
	\begin{tikzpicture}[scale=1, every node/.style={transform shape}]
         \node[inner sep=0pt] at (0,0) {\includegraphics[height=4.05cm, align=t]{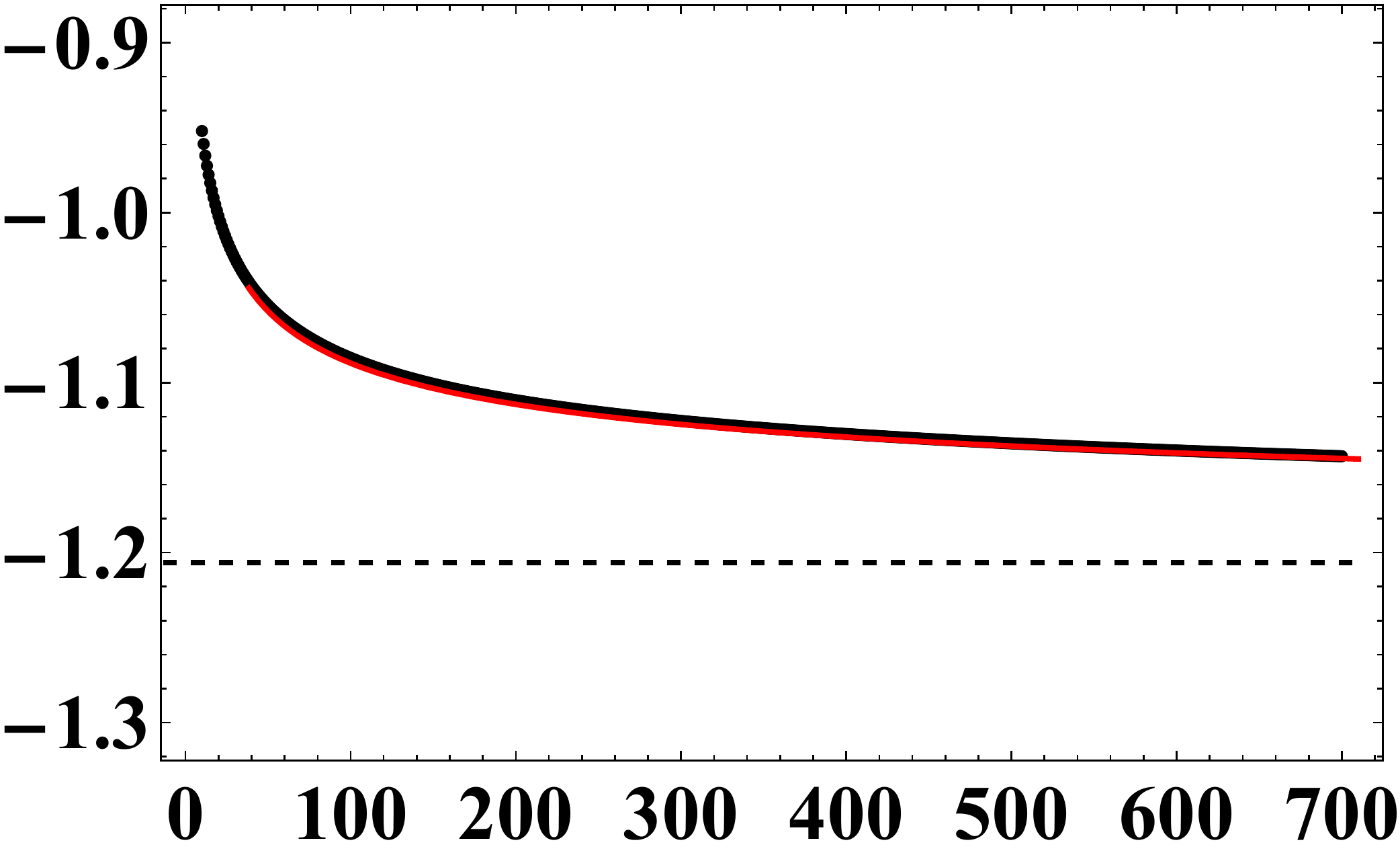}};
         \node at (3.6,-1.7) {\textbf{\textit{N}}};
         \node at (0.3,2.4) {$\hat{\sigma}^2_2 (N)$};
         \end{tikzpicture}
	\end{minipage} \quad
	\begin{minipage}{0.45\textwidth}
	\begin{tikzpicture}[scale=1, every node/.style={transform shape}]
         \node[inner sep=0pt] at (0,0) {\includegraphics[height=4.1cm, align=t]{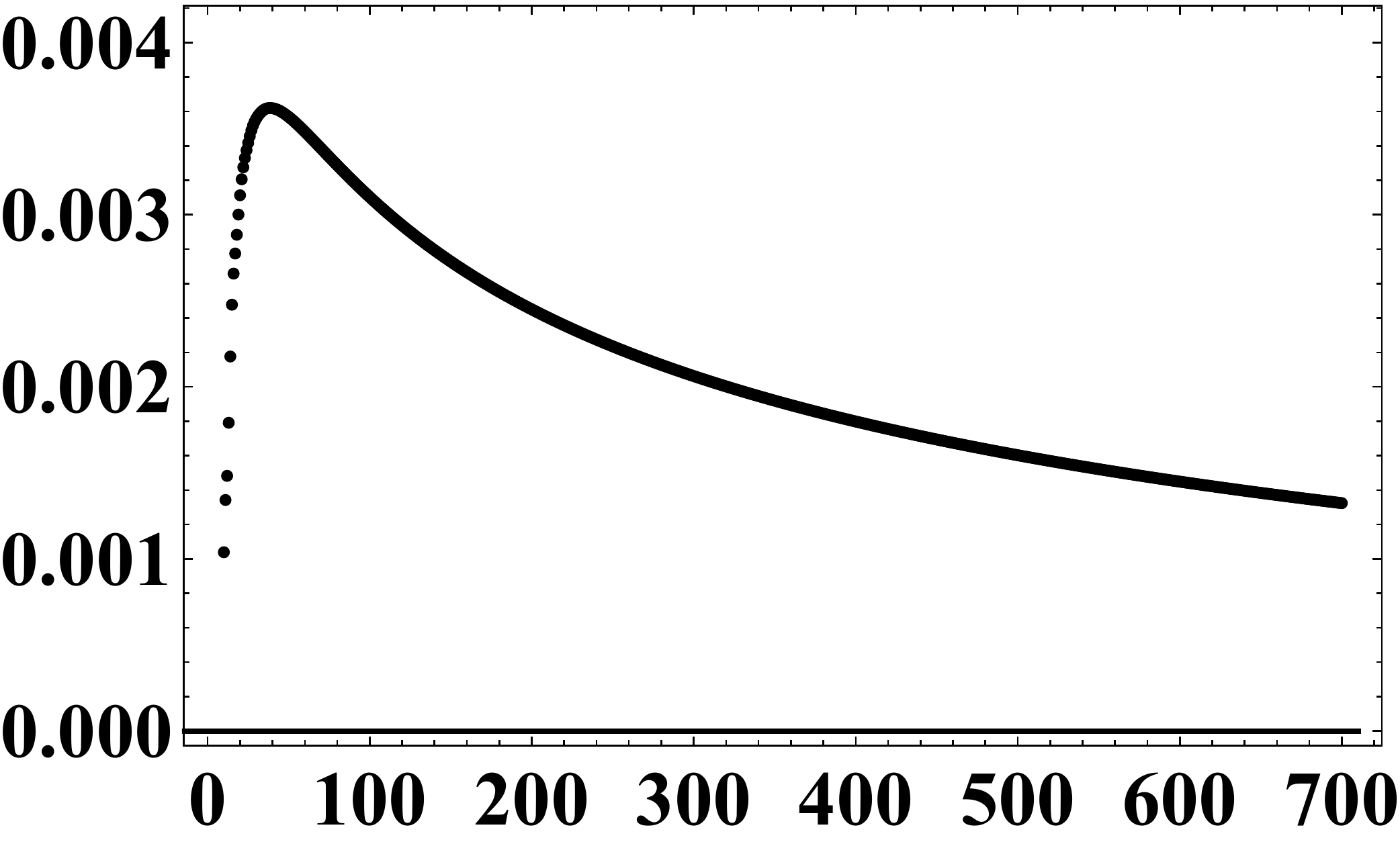}};
         \node at (3.7,-1.75) {\textbf{\textit{N}}};
         \node at (0.3,2.3) {Difference \eqref{e:sigmahatdiff}};
         \end{tikzpicture}
	\end{minipage}
	\caption{As in Figure \ref{f:muhat}, on the left we have the exact values of $\hat{\sigma}^2_2 (N)$ from \eqref{4.15} for $N=10, 11, \dots, 700$ calculated using \eqref{Var1} [black dots] compared to the conjectured form in \eqref{c:musigma} [red line]. On the right is plotted the difference \eqref{e:sigmahatdiff}.}
	\label{f:sigmahat}
\end{figure}

\subsection{The quantities  ${\rm Pr}\Big ( {\linc_N  - 2 \sqrt{N} \over N^{1/6}}  \le t \Big )$ and 
${\rm Pr} \Big ( {\ldec_N  - 2 \sqrt{N} \over N^{1/6}}  \le t \Big )$}

Analogous to (\ref{4.5}) we have
\begin{multline}
 \sum_{N=0}^{\infty} \frac{z^N}{N!!} {\rm Pr} \Big(l^{\boxslash}_{N} \leq l \Big) \\ 
=e^{z^2/2}
\label{e:Ginc}  \exp \left( \frac{1}{2} \int_0^{4z^2} \frac{v (r;l-1)}{r} dr \right) \cosh \left(- \frac{1}{4} \int_0^{4z^2} \frac{p_{\rmhard} (r; l-1)}{\sqrt{r}} dr \right) :=  G^{\boxslash} (z;l) 
  \end{multline}
and
\begin{multline} 
  \sum_{N=0}^{\infty} \frac{z^N}{N!!} {\rm Pr} \Big(l^{\boxbslash}_{N} \leq 2l \Big) \\  
   = e^{z^2/2} \exp \left( \frac{1}{2} \int_0^{4z^2} \frac{v (r; 2l+1)}{r} dr \right) \exp \left(- \frac{1}{4} \int_0^{4z^2} \label{e:Gdec} \frac{p_{\rmhard} (r; 2l+1)}{\sqrt{r}} dr \right) :=  G^{\boxbslash} (z;l) 
  \end{multline}
  These follow from (\ref{e:lBoxE4Hard}),  (\ref{e:lBoxE1Hard}), (\ref{e:E1hardDE}) and  (\ref{e:E4hardDE}); for example see \cite[\S 10.7]{Fo10}.  Hence
 \begin{align}
\label{G1}   G^{\boxslash} (z;l)    = \sum_{N=0}^{\infty} z^{2N} c_N^{\boxslash} (l), \qquad c_N^{\boxslash} (l) N!! := \Pr \left( l_N^{\boxslash} \leq l \right),
\end{align}
and 
 \begin{align}
\label{G2}   G^{\boxbslash} (z;l)    = \sum_{N=0}^{\infty} z^{2N} c_N^{ \boxbslash } (l), \qquad c_N^{\boxbslash} (2l) N!! := \Pr \left( l_N^{\boxbslash} \leq 2l \right),
\end{align}
We now proceed as detailed in Section \ref{S4.2}, which provides us with the exact values of $\{ c_N^{\boxbslash} (l) \}$ and $\{ c_N^{\boxslash} (l) \}$ for $N$ up to $400$. That is, we find a series solution of degree $400$ to the differential equation in \eqref{e:E2hardDEb}, and use it (along with the $v(r;l)$ from Section \ref{S4.2}) to expand \eqref{e:Ginc} and \eqref{e:Gdec} in powers of $z$.

In Figures \ref{f:SymmLDSn400} and \ref{f:SymmLISn400} we display the cases $N = 400$ in graphical form, along with the scaled differences
\begin{align}
\label{d:SymmLDSerror}   \delta_1 (l) := N^{1/3} \left [  \Pr \left( \ldec_{N} \leq l \right) - E_1^{\rmsoft} \left(0; \left(\frac{l + 1  -2 \sqrt{N}}{N^{1/6}}, \infty \right) \right) \right ],\\
\label{d:SymmLISerror}  \delta_4 (l) := N^{1/3} \left [  \Pr \left( \linc_{N} \leq l \right) - \tilde{E}_4^{\rmsoft} \left(0; \left(\frac{l  - 1 -2 \sqrt{N}}{N^{1/6}}, \infty \right) \right) \right ].
\end{align}


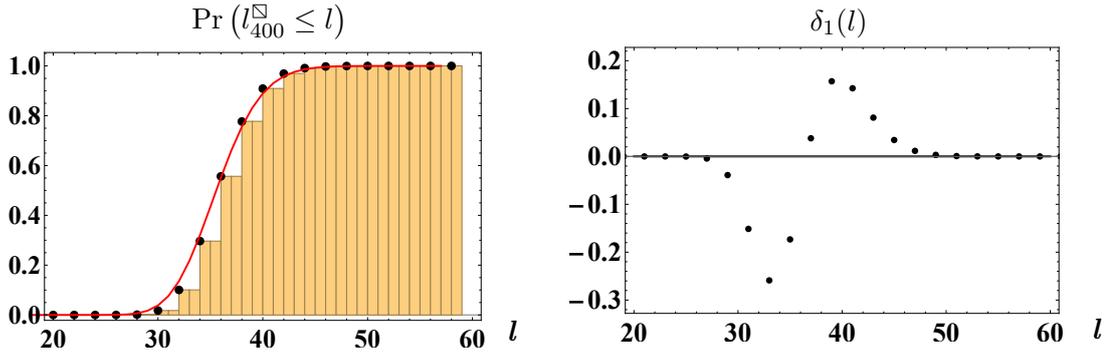
\begin{figure}
\centering
     \begin{minipage}{0.45\textwidth}
	\begin{tikzpicture}[scale=1, every node/.style={transform shape}]
         \node[inner sep=0pt] at (0,0) {\includegraphics[height=4cm, align=t]{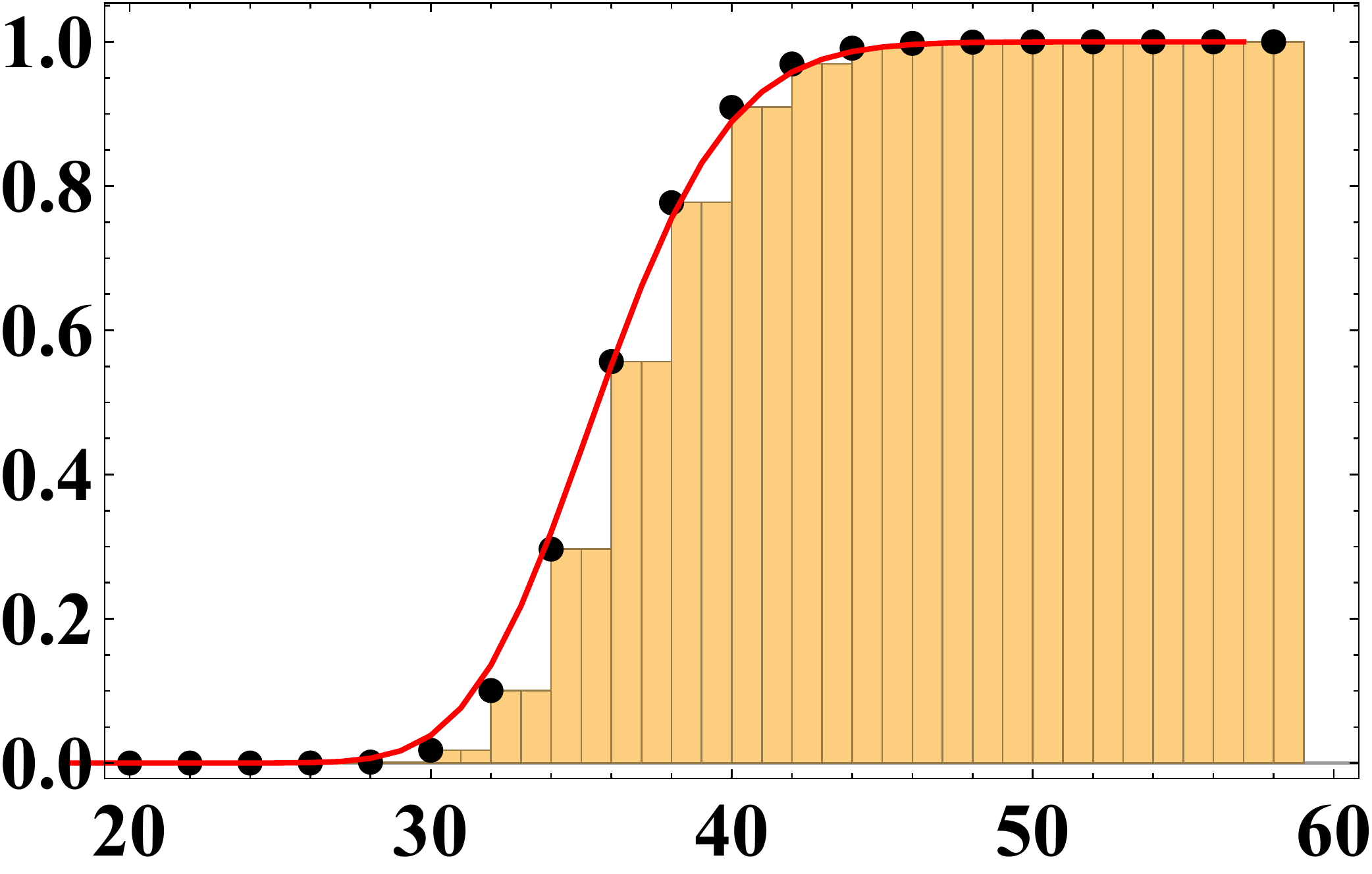}};
         \node at (3.5,-1.7) {\textbf{\textit{l}}};
         \node at (0.3,2.4) {$\Pr\big( l^{\boxbslash}_{400} \leq l\big)$};
         \end{tikzpicture}
	\end{minipage} \quad
	\begin{minipage}{0.45\textwidth}
	\begin{tikzpicture}[scale=1, every node/.style={transform shape}]
         \node[inner sep=0pt] at (0,0) {\includegraphics[height=4cm, align=t]{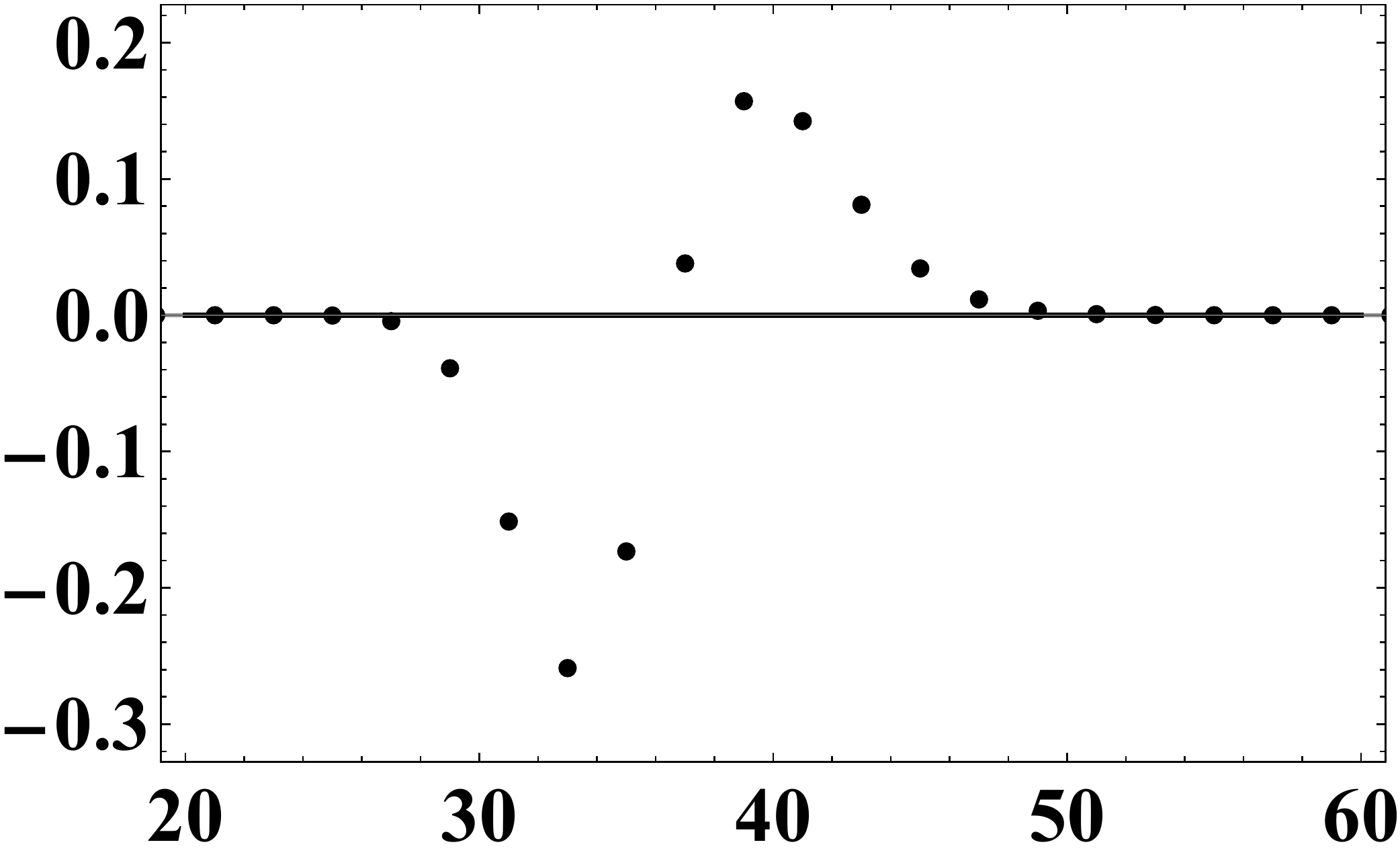}};
         \node at (3.7,-1.75) {\textbf{\textit{l}}};
         \node at (0.3,2.3) {$\delta_1 (l)$};
         \end{tikzpicture}
	\end{minipage}
        \caption{On the left we have the empirical CDF of the longest decreasing subsequences of 1,000,000 random permutations (which consist entirely of two cycles) of length $N= 400$ along with the calculation of the exact CDF using $c_{400}^{\boxbslash} (l)$ in \eqref{G2} [black dots] and the limiting CDF given by the right hand side of (\ref{1.1o}) with $t= (l -1 -2\sqrt{N})/N^{1/6}$ [red curve]. On the right is plotted the quantity $\delta_1 (l)$ from \eqref{d:SymmLDSerror}.}
        \label{f:SymmLDSn400}
\end{figure}

\begin{figure}
\centering
     \begin{minipage}{0.45\textwidth}
	\begin{tikzpicture}[scale=1, every node/.style={transform shape}]
         \node[inner sep=0pt] at (0,0) {\includegraphics[height=4cm, align=t]{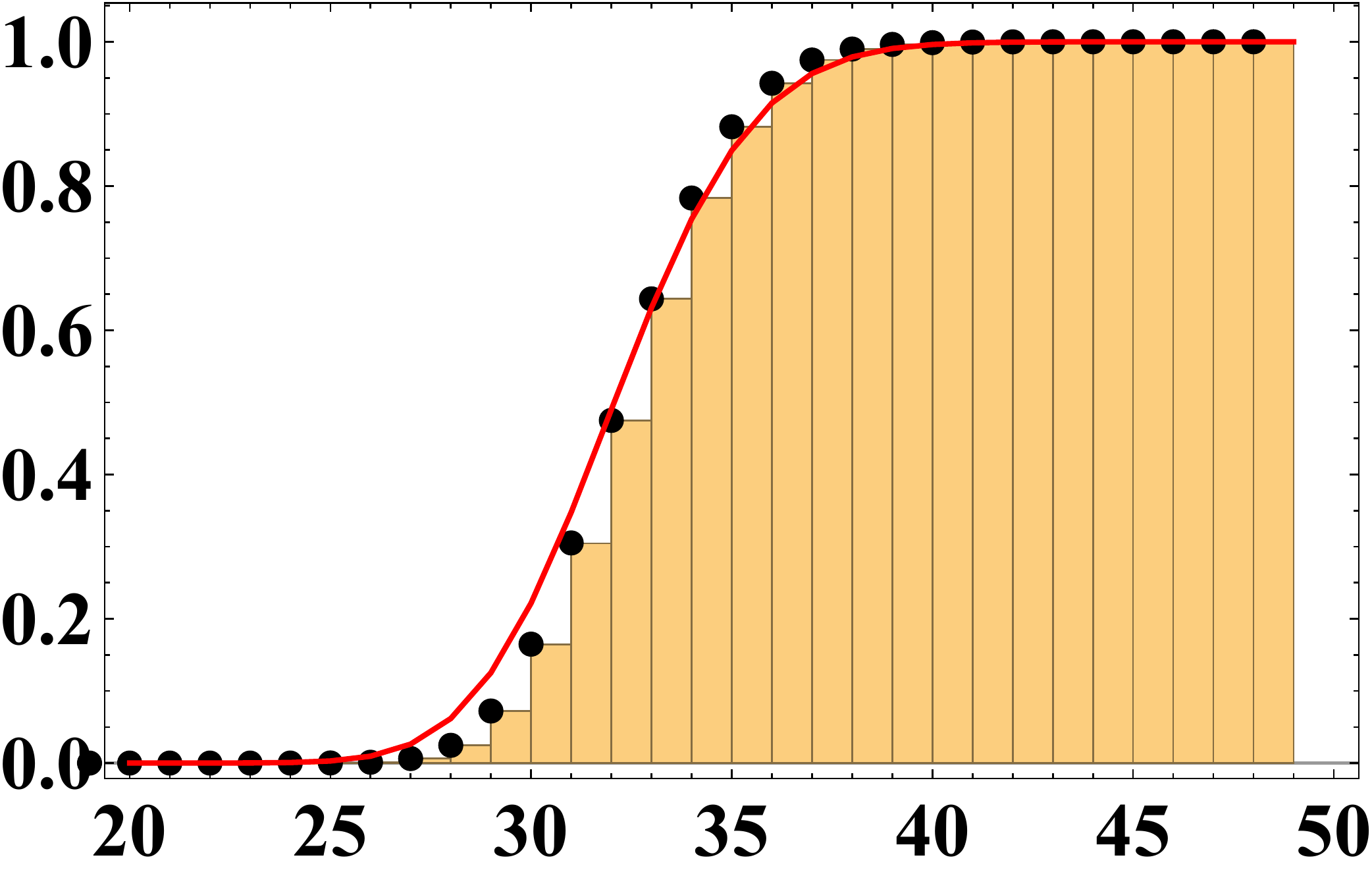}};
         \node at (3.5,-1.7) {\textbf{\textit{l}}};
         \node at (0.3,2.4) {$\Pr\big( l^{\boxslash}_{400} \leq l\big)$};
         \end{tikzpicture}
	\end{minipage} \quad
	\begin{minipage}{0.45\textwidth}
	\begin{tikzpicture}[scale=1, every node/.style={transform shape}]
         \node[inner sep=0pt] at (0,0) {\includegraphics[height=4cm, align=t]{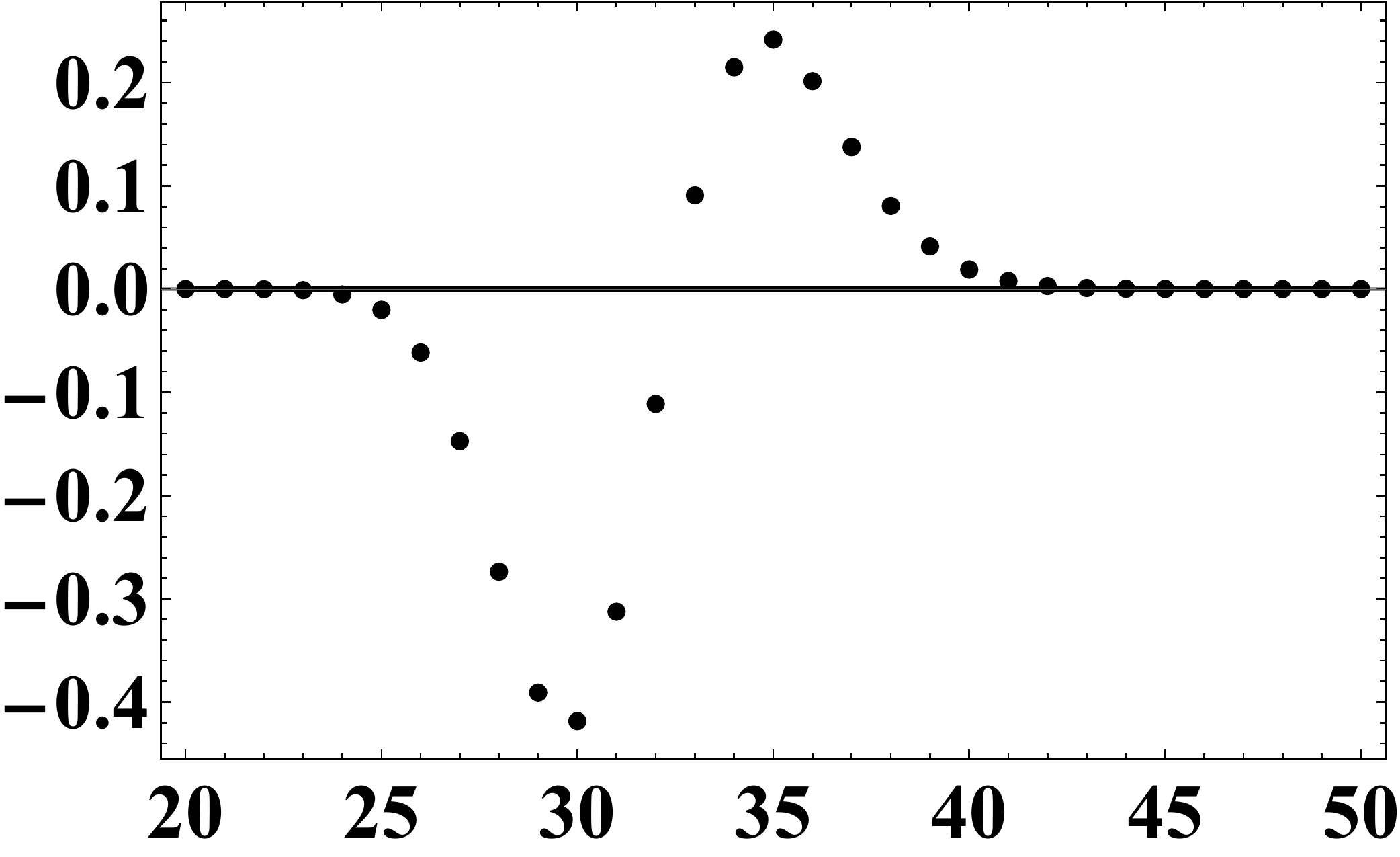}};
         \node at (3.7,-1.75) {\textbf{\textit{l}}};
         \node at (0.3,2.3) {$\delta_4 (l)$};
         \end{tikzpicture}
	\end{minipage}
        \caption{Here we have the plots analogous to those in Figure \ref{f:SymmLDSn400}, now counting the longest \textit{increasing} subsequences, where the limiting CDF is given by the right hand side of \eqref{1.1n} with $t= (l +1 -2\sqrt{N})/N^{1/6}$ [red curve]. The right panel is again the scaled difference between the red curve and the black dots, i.e. $\delta_4 (l)$ from \eqref{d:SymmLISerror}}
        \label{f:SymmLISn400}
\end{figure}
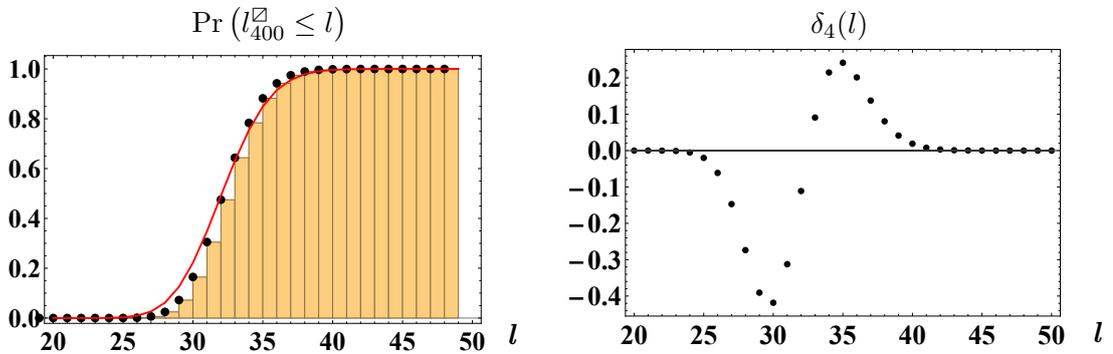

The exact values for $\Pr \Big(l^{\boxbslash}_{N} \leq l \Big)$ and $\Pr \Big(l^{\boxslash}_{N} \leq l \Big)$ can be supplemented by simulations as in Section \ref{S4.3}. For this we use the C++ code for generating self-inverse permutations from \cite{Arndt2011}, which samples the permutations uniformly for very large $N$. To find the longest increasing subsequence we use the C++ implementation from \cite{Sanaka2022} of an optimal algorithm. Plotting the scaled differences (\ref{d:SymmLDSerror}) and (\ref{d:SymmLISerror}) for $N=400, 20000$ and $10^5$, with the horizontal axis rescaled by $t= (l \pm 1- 2\sqrt{N})/N^{1/6}$ --- see Figure \ref{f:CompDelta14} --- gives evidence for the analogue of Conjecture \ref{C1}. 

\begin{conj}\label{C2}
Specify $t^*$ as in (\ref{1.1i}). Set $ F_{1,0}(t)  = E_1^{\rm soft}(0;(t,\infty))$. For some $ F_{1,1}(t)$ we  have
\begin{equation}\label{1.1iX1}
 {\rm Pr} \left ( {l_N^\boxbslash +1 - 2 \sqrt{N} \over N^{1/6}}  \le t \right ) =  F_{1,0}(t^*) +  {1\over N^{1/3}}   F_{1,1}(t) + \cdots
  \end{equation}
  Similarly, with $ F_{4,0}(t)  = \tilde{E}_4^{\rm soft}(0;(t,\infty))$, for some $ F_{4,1}(t)$ we  have
\begin{equation}\label{1.1iX4}
 {\rm Pr} \left ( {l_N^\boxslash -1 - 2 \sqrt{N} \over N^{1/6}}  \le t \right ) =  F_{4,0}(t^*) +  {1\over N^{1/3}}   F_{4,1}(t) + \cdots
  \end{equation}
\end{conj}

\begin{figure}
\centering

     \begin{minipage}{0.45\textwidth}
	\begin{tikzpicture}[scale=1, every node/.style={transform shape}]
         \node[inner sep=0pt] at (0,0) {\includegraphics[height=4cm, align=t]{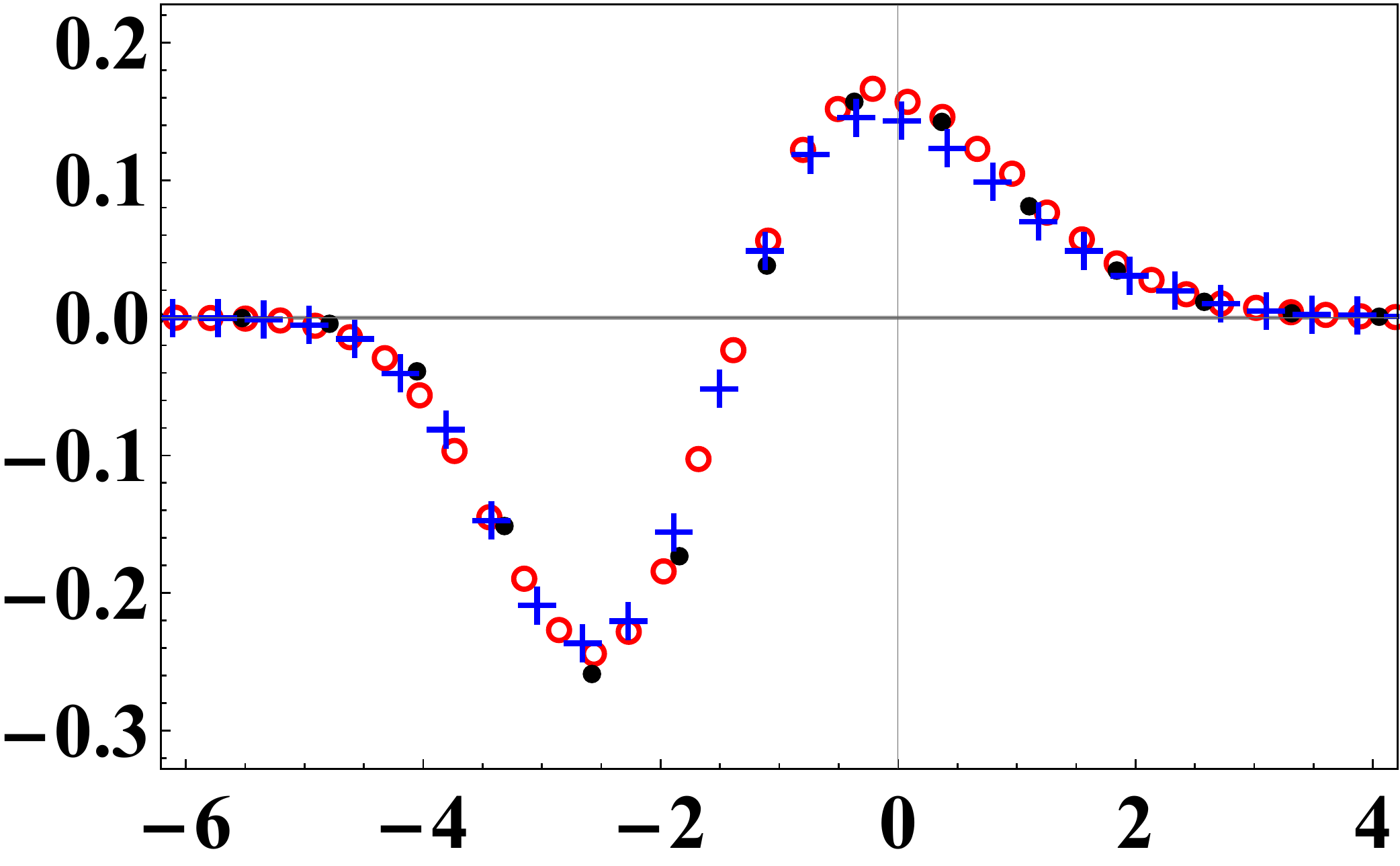}};
         \node at (3.5,-1.7) {\textbf{\textit{t}}};
         \node at (0.3,2.3) {$\delta_1 (t)$};
         \end{tikzpicture}
	\end{minipage} \quad
	\begin{minipage}{0.45\textwidth}
	\begin{tikzpicture}[scale=1, every node/.style={transform shape}]
         \node[inner sep=0pt] at (0,0) {\includegraphics[height=4cm, align=t]{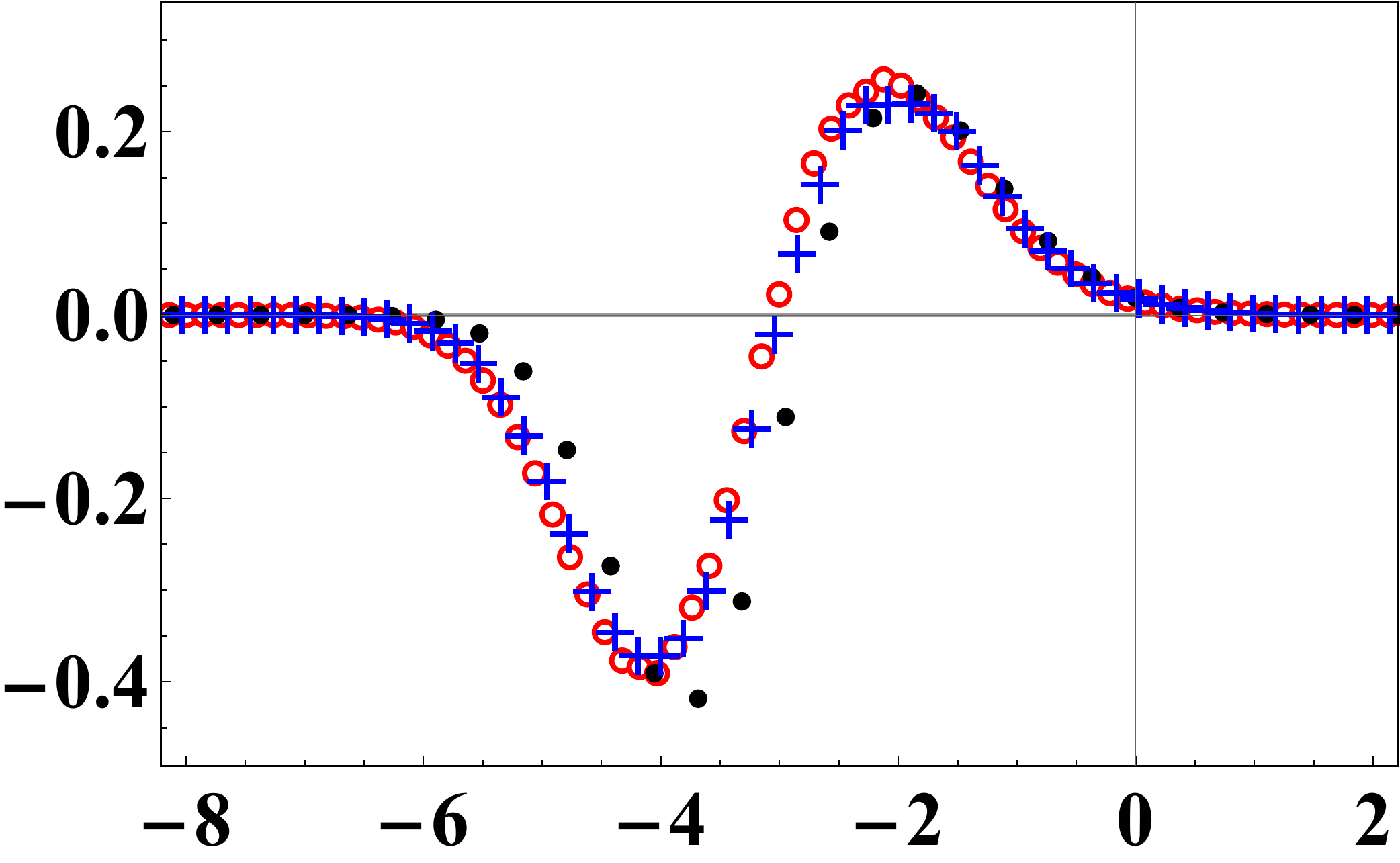}};
         \node at (3.5,-1.7) {\textbf{\textit{t}}};
         \node at (0.3,2.3) {$\delta_4 (t)$};
         \end{tikzpicture}
	\end{minipage}
\caption{Comparison of $\delta_1 (t)$ from \eqref{d:SymmLDSerror} and $\delta_4 (t)$ from \eqref{d:SymmLISerror} rescaled by $t= (l \pm 1 - 2\sqrt{N})/N^{1/6}$, with $N= 400$ [black dots], $20000$ [blue crosses] and $10^5$ [red circles].}
\label{f:CompDelta14}
\end{figure}
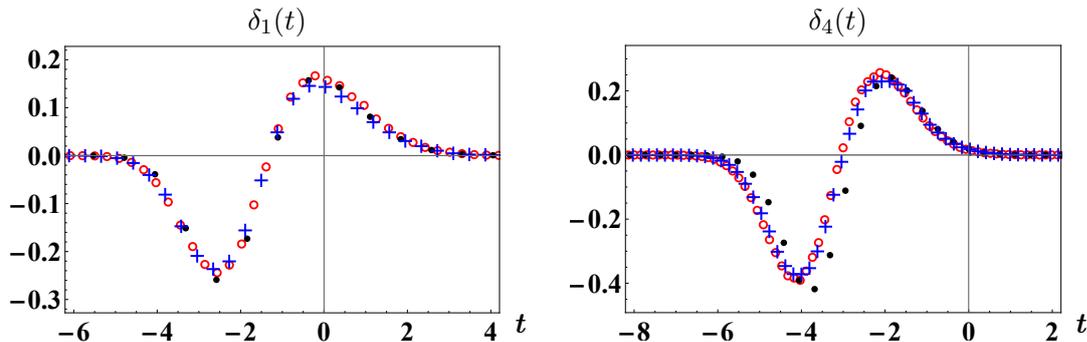

Finally, we comment on  the analogues of (\ref{4.15}) in relation to the large $N$ expansion of the mean and variance. 
We proceeded as for $l_N^\Box$ and postulated that the quantities $\hat{\mu}_1(N),  \, \hat{\sigma}^2_1(N), \, \hat{\mu}_4(N),
\,   \hat{\sigma}^2_4(N)$ each have a large $N$ expansion of the form $c + dN^{-1/3} + \cdots$. The exact data for
$l^{\boxbslash}_N, \, l^{\boxslash}_N$ available for $N$ up to 400 was then used to find best fits for the corresponding
values of $c$ and $d$. However, as distinct from our findings in the case of $l_N^\Box$
seen in Figures \ref{f:muhat} and \ref{f:sigmahat}, when calculating the differences
analogous to (\ref{e:sigmahatdiff}) a decrease to zero as $N$ increased was not observed. Hence, as yet we do not have
convincing evidence for $N$ dependence of higher order terms in the large $N$ expansion of the mean and standard
deviation of $l^{\boxbslash}_N, \, l^{\boxslash}_N$.

  \subsection*{Acknowledgements}
	This research is part of the program of study supported
	by the Australian Research Council Centre of Excellence ACEMS
	and the Discovery Project grant DP210102887.
	We thank Eric Rains for providing the computer program as referenced
	in the text. We are most grateful to Jinho Baik for bringing to our attention  the crucial
	reference \cite{BJ13}, which unfortunately was missed when we prepared the first draft
	of this work. Also, we acknowledge the contribution of Allan Trinh in
	collaborating in the early stages of this project.

\providecommand{\bysame}{\leavevmode\hbox to3em{\hrulefill}\thinspace}
\providecommand{\MR}{\relax\ifhmode\unskip\space\fi MR }
\providecommand{\MRhref}[2]{%
  \href{http://www.ams.org/mathscinet-getitem?mr=#1}{#2}
}
\providecommand{\href}[2]{#2}


\begin{thebibliography}{10}

\bibitem{AS72}
M. Abramowitz and I. A. Stegun, editors. \emph{Handbook of Mathematical Functions}, Dover, New York (1972).



\bibitem{AD99}
D.~Aldous and P.~Diaconis, \emph{Longest increasing subsequences: from patience sorting to
  the {B}aik-{D}eift-{J}ohansson theorem}, Bull. Amer. Math. Soc. \textbf{36}
  (1999), 413--432.

\bibitem{Arndt2011} J.~Arndt, \emph{Matters Computational: Ideas, algorithms and source code}, Springer, Heidelberg (2011).
  
  \bibitem{BB68} R.M.~Baer and P.~Brock, \emph{Natural sorting over permutation spaces},
  Math. Comp. \textbf{22} (1968), 385--410.


\bibitem{BDJ98}
J.~Baik, P.~Deift, and K.~Johansson, \emph{On the distribution of the length of
  the longest increasing subsequence of random permutations}, J. Amer. Math.
  Soc. \textbf{12} (1999), 1119--1178.
  
  \bibitem{BJ13}
 J.~Baik and R. Jenkins,  \emph{Limiting distribution of maximal crossing and nesting of Poissonized
random matchings}, Ann. Probab.,  \textbf{41} (2013), 4359--4406.
  
\bibitem{BR01}
J.~Baik and E.M. Rains, \emph{Symmetrized random permutations}, Random matrix models and their
  applications (P.M. Bleher and A.R. Its, eds.), Mathematical Sciences Research
  Institute Publications, vol.~40, Cambridge University Press, Cambridge (2001),
  pp.~171--208.


\bibitem{BR01a}
J.~Baik and E.M. Rains, \emph{Algebraic aspects of increasing subsequences}, Duke Math. J.
  \textbf{109} (2001), 1--65.

\bibitem{BR01b}
J.~Baik and E.M. Rains, \emph{The asymptotics of monotone subsequences of involutions}, Duke
  Math. J. \textbf{109} (2001), 205--281.

\bibitem{BBLM06}
E.~Bogomolny, O.~Bohigas, P.~Leboeuf, and A.C. Monastra, \emph{On the spacing
  distribution of the {R}iemann zeros: corrections to the asymptotic result},
  J. Phys. A \textbf{39} (2006), 10743--10754.
  
\bibitem{Bornemann2010}
F. Bornemann, \emph{On the numerical evaluation of distributions in random matrix theory: a review}, Markov Processes Relat. Fields \textbf{16} (2010), 803--866.
  
  \bibitem{BFM17}
F. Bornemann, P. Forrester and A. Mays.
\newblock \emph{Finite size effects for spacing distributions in random matrix theory: circular ensembles and Riemann zeros.}
\newblock Stud. Appl. Math., \textbf{138} (2017), 401--437.

\bibitem{BF03}
A.~Borodin and P.J. Forrester, \emph{Increasing subsequences and the
  hard-to-soft transition in matrix ensembles}, J.Phys. A \textbf{36} (2003),
  2963--2981.
  
  
\bibitem{DF06}  
  P.~Desrosiers and P.J.~Forrester,  \emph{Relationships between $\tau$-functions and Fredholm determinant
expressions for gap probabilities in random matrix theory}, Nonlinearity \textbf{19} (2006), 1643--1656.

\bibitem{DF06a}
P.~Desrosiers and P.J. Forrester, \emph{Hermite and {L}aguerre
  $\beta$-ensembles: asymptotic corrections to the eigenvalue density}, Nucl.
  Phys. B \textbf{743} (2006), 307--332.



\bibitem{Dy62a}
F.J.~Dyson, \emph{Statistical theory of energy levels of complex systems {III}},
  J. Math. Phys. \textbf{3} (1962), 166--175.
  
  \bibitem{FF11}
  P.L.~Ferrari and R.~Frings, \emph{Finite time corrections in KPZ growth models}, J. Stat.
Phys. \textbf{144} (2011), 1123--1150.
  
  
  \bibitem{Fo93a}
P.J. Forrester, \emph{The spectrum edge of random matrix ensembles}, Nucl. Phys. B
  \textbf{402} (1993), 709--728.
  
  \bibitem{Fo93c}
P.J. Forrester, \emph{Exact results and universal asymptotics in the {Laguerre} random
  matrix ensemble}, J. Math. Phys. \textbf{35} (1993), 2539--2551.
  
  
 \bibitem{Fo00}  
  P.J. Forrester, \emph{Painlev\'e transcendent evaluation of the scaled distribution of
the smallest eigenvalue in the Laguerre orthogonal and symplectic ensembles}, arXiv:nlin.SI/0005064 (2000).
  
 \bibitem{Fo06}  
P.J. Forrester, \emph{Hard and soft edge spacing distributions for random matrix ensembles with
orthogonal and symplectic symmetry}, Nonlinearity \textbf{19} (2006), 2989--3002.
  
  \bibitem{Fo10}
P.J. Forrester, \emph{Log-gases and random matrices}, Princeton University Press,
  Princeton, NJ (2010).
  
  
   \bibitem{FLT20}
  P.J. Forrester, S.-H.~Li and A.K.~Trinh,  \emph{Asymptotic correlations with corrections for the circular Jacobi
  $\beta$-ensemble}, J. Approximation Th.   \textbf{271} (2021), 105633.
  
  \bibitem{FM15}
P.J. Forrester and A. Mays.
\newblock \emph{Finite-size corrections in random matrix theory and Odlyzko's dataset for the Riemann zeros}.
\newblock Proc. R. Soc. A, \textbf{471} (2015), 20150436.

\bibitem{FPTW19}
P.J.~Forrester, J.H.H.~Perk, A.K.~Trinh and N.S.~Witte, \emph{Leading corrections to the scaling function on the diagonal for the two-dimensional Ising model},
J. Stat. Mech. \textbf{2019} (2019), 023106.


\bibitem{FT18}
P.J.~Forrester and A.K.~Trinh. 
\newblock \emph{Functional form for the leading correction to the distribution of the largest eigenvalue in the GUE and LUE}. \newblock J. Math. Phys., 59(5) (2018), 053302.


 \bibitem{FT19}
		P.J.~Forrester and A.K.~Trinh, 
		\emph{Finite-size corrections at the hard edge for the Laguerre $\beta$ ensemble}, Stud. Applied Math.		 
\textbf{143} (2019), 315--336.


\bibitem{FT19a}
P.J.~Forrester and A. Trinh. 
\newblock \emph{Optimal soft edge scaling variables for the Gaussian and Laguerre even $\beta$ ensembles}. 
\newblock Nuclear Phys. B, 938 (2019), 621--639.





\bibitem{FW04}
P.J. Forrester and N.S. Witte, \emph{Application of the $\tau$-function theory of {Painlev\'e}
  equations to random matrices: {PVI}, the {JUE},{CyUE}, {cJUE} and scaled
  limits}, Nagoya Math. J. \textbf{174} (2004), 29--114.



\bibitem{Ge90}
I.M. Gessel, \emph{Symmetric functions and $p$-recursiveness}, J. Comb. Th. A
  \textbf{53} (1990), 257--285.
  
  \bibitem{Ja64}
A.T. James, \emph{Distributions of matrix variate and latent roots derived from
  normal samples}, Ann. Math. Statist. \textbf{35} (1964), 475--501.
  
\bibitem{Jo98}  
K. Johansson,  \emph{The longest increasing subsequence in a random
permutation and a unitary random matrix model},  Math. Research Lett.
\textbf{5} (1998), 63--82.
 

\bibitem{KS00a}
J.P. Keating and N.C. Snaith, \emph{Random matrix theory and $\zeta(1/2 +
  it)$}, Commun. Math. Phys. \textbf{214} (2001), 57--89.
  
  \bibitem{LS77}
  B.F. Logan and L.A. Shepp, \emph{A variational problem for random Young tableaux},  Advances in
Math.,  \textbf{26} (1977), 206--222.
  
  \bibitem{Od01}
A.M. Odlyzko, \emph{The $10^{22}$-nd zero of the {R}iemann zeta function},
  Dynamical, Spectral, and Arithmeitc Zeta Functions (M.~van Frankenhuysen and
  M.L. Lapidus, eds.), Contemporary Math. \textbf{2001}, Amer.~Math.~Soc,
  Providence, RI, (2001), 139--144.
  
 \bibitem{OR00} 
A.M. Odlyzko and E.M. Rains,  \emph{On Longest Increasing Subsequences
in Random Permutations}, Amer. Math. Soc., Contemp. Math.  \textbf{251}
(2000), 439--451.

\bibitem{Ra98}
E.M. Rains, \emph{Increasing subsequences and the classical groups}, Elect. J. of
  Combinatorics \textbf{5} (1998), \#R12.
  
  \bibitem{RR09}
  J. Ramirez and B. Rider, \emph{Diffusion at the random matrix hard edge}, Commun. Math.
Phys. \textbf{288} (2009), 887--906.

 \bibitem{Ro15}
D. Romik,  \emph{The surprising mathematics of longest increasing subsequences},  Institute of Mathematical Statistics
Textbooks. Cambridge University Press (2015).
  
  \bibitem{Sanaka2022}
  V. Sanaka, \emph{Longest Increasing Subsequence Size (N log N)}, Geeks for Geeks (2012), url:\\ \path{www.geeksforgeeks.org/} \\ \path{longest-monotonically-increasing-subsequence-size-n-log-n/}.
  
  \bibitem{Sa04}  
P.~Sarnak,  \emph{Problems of the Millennium: The Riemann Hypothesis}, Clay Mathematics Institute Annual Report (2004), 5--21. 

\bibitem{Sa05}
T.~Sasamoto, \emph{Spatial correlations of the {1D} {KPZ} surface on a flat
  substrate}, J. Phys. A \textbf{38} (2005), L549--L556.

\bibitem{TW94a}
C.A. Tracy and H.~Widom, \emph{Level-spacing distributions and the {Airy} kernel}, Commun.
  Math. Phys. \textbf{159} (1994), 151--174.
  
  \bibitem{TW96}
C.A. Tracy and H.~Widom, \emph{On orthogonal and symplectic matrix ensembles}, Commun. Math.
  Phys. \textbf{177} (1996), 727--754.
  
  \bibitem{TW94b}
C.A. Tracy and H.~Widom, \emph{Level-spacing distributions and the {Bessel} kernel}, Commun.
  Math. Phys. \textbf{161} (1994), 289--309.
  
  \bibitem{WW65}
E.T. Whittaker and G.N. Watson, \emph{A course of modern analysis}, 2nd ed.,
  Cambridge University Press, Cambridge (1965).



\end{thebibliography}
\end{document}